\documentclass[12pt, draftclsnofoot, onecolumn]{IEEEtran}
\usepackage{amsmath, amssymb, amsbsy,cite}
\usepackage{mathtools}
\usepackage[small]{caption}
\usepackage{amsthm,multirow,color,amsfonts}
\usepackage{tabulary}
\usepackage{subfigure}
\usepackage{graphicx}
\usepackage{setspace}
\usepackage{enumerate}
\usepackage{bbm,bm}
\usepackage[]{algorithm2e}
\usepackage{setspace}
\usepackage{comment}
\usepackage{cleveref}
\usepackage{url}

\makeatletter
\g@addto@macro{\normalsize}{
	\setlength{\abovedisplayskip}{3pt plus 2pt minus 2pt}
	\setlength{\abovedisplayshortskip}{3pt plus 2pt minus 2pt}
	\setlength{\belowdisplayskip}{3pt plus 2pt minus 2pt}
	\setlength{\belowdisplayshortskip}{3pt plus 2pt minus 2pt}
	\setlength{\textfloatsep}{10pt plus 2pt minus 2pt}
}
\makeatother

\title{Tiny Codes for Guaranteeable Delay}

 \author{Derya~Malak,~Muriel~M\'{e}dard, and~Edmund~M.~Yeh
 \thanks{D. Malak and M. M\'{e}dard are with the Research Laboratory of Electronics (RLE), at The Massachusetts Institute of Technology, Cambridge, MA 02139 USA (email: \{deryam, medard\}@mit.edu). E. Yeh is with Electrical and Computer Engineering Department, at Northeastern University, Boston, MA 02115, USA (email: eyeh@ece.neu.edu). \hfill Article last revised: {\today}.}}

\newtheorem{defi}{Definition}

\newtheorem{prop}{Proposition}

\newtheorem{cor}{Corollary}

\newcommand{\matr}[1]{\mathbf{#1}}

\DeclareMathOperator*{\CF}{\sf{CF}} 
\DeclareMathOperator*{\Cod}{\sf{C}} 

\DeclareMathOperator*{\etaU}{\eta_{\sf{ARQ}}}
\DeclareMathOperator*{\etaH}{\eta_{\sf{HARQ}}}
\DeclareMathOperator*{\etaCF}{\eta_{\sf{CF-ARQ}}}
\DeclareMathOperator*{\etaC}{\eta_{\sf{C-ARQ}}}

\DeclareMathOperator*{\tauCF}{\tau_{\sf{CF-ARQ}}}

\DeclareMathOperator*{\avgtauU}{\overline{\tau}_{\sf{ARQ}}}
\DeclareMathOperator*{\avgtauH}{\overline{\tau}_{\sf{HARQ}}}
\DeclareMathOperator*{\avgtauCF}{\overline{\tau}_{\sf{CF-ARQ}}}
\DeclareMathOperator*{\avgtauC}{\overline{\tau}_{\sf{C-ARQ}}}

\DeclareMathOperator*{\DU}{\overline{D}_{\sf{ARQ}}}
\DeclareMathOperator*{\DeH}{\overline{D}_{\sf{HARQ}}}
\DeclareMathOperator*{\DCF}{\overline{D}_{\sf{CF-ARQ}}}
\DeclareMathOperator*{\DC}{\overline{D}_{\sf{C-ARQ}}}

\DeclareMathOperator*{\varDU}{\sigma^2_{{D}_{\sf{ARQ}}}}

\DeclareMathOperator*{\varDC}{\sigma^2_{{D}_{\sf{C-ARQ}}}}

\DeclareMathOperator*{\PhitauU}{\matr{\Phi_{{\sf \tau}_{\sf{ARQ}}}}}
\DeclareMathOperator*{\PhitauH}{\matr{\Phi_{{\sf \tau}_{\sf{HARQ}}}}}
\DeclareMathOperator*{\PhitauCF}{\matr{\Phi_{{\sf \tau}_{\sf{CF-ARQ}}}}}
\DeclareMathOperator*{\PhitauC}{\matr{\Phi_{{\sf \tau}_{\sf{C-ARQ}}}}}

\DeclareMathOperator*{\phitauU}{\Phi_{{\sf \tau}_{\sf{ARQ}}}}

\DeclareMathOperator*{\phitauCF}{\Phi_{{\sf \tau}_{\sf{CF-ARQ}}}}

\DeclareMathOperator*{\PhiDU}{\matr{\Phi_{{\sf D}_{\sf{ARQ}}}}}
\DeclareMathOperator*{\PhiDH}{\matr{\Phi_{{\sf D}_{\sf{HARQ}}}}}
\DeclareMathOperator*{\PhiDCF}{\matr{\Phi_{{\sf D}_{\sf{CF-ARQ}}}}}
\DeclareMathOperator*{\PhiDC}{\matr{\Phi_{{\sf D}_{\sf{C-ARQ}}}}}

\DeclareMathOperator*{\phiDCF}{\Phi_{{\sf D}_{\sf{CF-ARQ}}}}

\DeclareMathOperator*{\tbg}{\sf BG_{\tau}}
\DeclareMathOperator*{\dbg}{\sf BG_{D}}

\begin{document}

\maketitle
\begin{abstract} 
Future 5G systems will need to support ultra-reliable low-latency communications scenarios. From a latency-reliability viewpoint, it is inefficient to rely on average utility-based system design. Therefore, we introduce the notion of guaranteeable delay which is the average delay plus three standard deviations of the mean. We investigate the trade-off between guaranteeable delay and throughput for point-to-point wireless erasure links with unreliable and delayed feedback, by bringing together signal flow techniques to the area of coding. We use tiny codes, i.e. sliding window by coding with just $2$ packets, and design three variations of selective-repeat ARQ protocols, by building on the baseline scheme, i.e. uncoded ARQ, developed by Ausavapattanakun and Nosratinia: (i) Hybrid ARQ with soft combining at the receiver; (ii) cumulative feedback-based ARQ without rate adaptation; and (iii) Coded ARQ with rate adaptation based on the cumulative feedback. Contrasting the performance of these protocols with uncoded ARQ, we demonstrate that HARQ performs only slightly better, cumulative feedback-based ARQ does not provide significant throughput while it has better average delay, and Coded ARQ can provide gains up to about $40\%$ in terms of throughput. Coded ARQ also provides delay guarantees, and is robust to various challenges such as imperfect and delayed feedback, burst erasures, and round-trip time fluctuations. This feature may be preferable for meeting the strict end-to-end latency and reliability requirements of future use cases of ultra-reliable low-latency communications in 5G, such as mission-critical communications and industrial control for critical control messaging.  
\end{abstract}  


\maketitle
       
%%%%%%%%%%%%%%%%%%%%%%%%%%%%%%%%%%%%%%%%%%%%%               
\section{Introduction}
\label{intro}
Ultra reliability and low latency in 5G are key factors for many applications ranging from industrial control (automation) \cite{Popetal2018}, tactile Internet \cite{Fettweis2014}, interactive gaming, remote healthcare, financial services, smart cities, public safety, defense \cite{urllc2018}, to mission-critical communications such as autonomous driving, drones, virtual and enhanced reality (wearable computing devices) \cite{BenDebPoo2018,Linetal2018,WCM2010,3GPPmeeting2016,3GPP2018Mar}. Ultra-reliable low-latency communications (URLLC) for a large number of Internet of Things (IoT) devices will be an important use case for future 5G communication networks \cite{DhiHuVis2015}. Such scenarios have strict requirements in terms of capacity, end-to-end latency (on the order of about a few milliseconds) and reliability (higher than 99.9999\%). 5G will need to support a round-trip time (RTT) of about 1 millisecond, an order of magnitude faster than 4G, along with necessary overhead for resource allocation and access in 5G networks. Such severe latency constraints together with the associated control information introduce a plethora of challenges in terms of the protocol stack design, control/user plane, and the core network \cite{Andrews2014}. 

Inevitable feedback channel impairments may cause unreliability in packet delivery. A negative acknowledgement (NACK), falsely received as positive acknowledgement (ACK) results in undesirable packet outage. Repetition of a packet and forward error-correction (FEC) help repair the loss of the packets over non-deterministic channel conditions. The role of feedback is to increase the reliability in packet delivery and the channel efficiency by limiting the repetitions. Automatic Repeat reQuest (ARQ) and hybrid ARQ (HARQ), which combines FEC and ARQ error control, have been used in 5G mobile networks \cite{3GPP2017}, to boost the performance of wireless technologies such as HSPA, WiMAX and LTE \cite{KhosVis2017}. A network coding based HARQ algorithm, which combines FEC- and network-coding-based ARQ to maximize the throughput and video quality for wireless video broadcast, has been proposed in \cite{LuWuXiDu2011}. ARQ and HARQ perform together to provide robustness in 4G LTE networks, and a system with reliable packet delivery. Failure in HARQ is compensated for by ARQ at the expense of extra latency for the packet \cite{3GPP2016}. While enhanced mobile broadband (eMBB) aims at high spectral efficiency, it can also rely on HARQ retransmissions to achieve high reliability. However, this might not be the case for URLLC due to the hard latency constraints \cite{BenDebPoo2018}. In this case, it is required to depart from average delay-based models, and design wireless systems that provide delay guarantees.

%%%%%%%%%%%%%%%%%%%%%%%%%%%%%%%%%%%%%%%%
\subsection{Related Work}
\label{relatedwork}
Different classes of codes have been proposed to correct errors over packet erasure channels. Block codes require a packet stream to be partitioned into blocks, each block being treated independently from the rest. Block codes for error correction have been considered in \cite{KarLei2014}. Streaming codes, e.g. convolutional codes, have the flexibility of grouping the blocks of information in an appropriate way, and decoding the part of the sequence with fewer erasures. They can correct more errors than classical block codes when considering the erasure channel \cite{KarLei2014}, \cite{Lieb2017}. Fountain codes have efficient encoding and decoding algorithms, and are capacity-achieving. However, they are not suitable for streaming because the decoding delay is proportional to the size of the data \cite{LubMitShoSpiSte1997}.

Network RTT can be estimated using delay measurements reported from the receiver's acknowledgments \cite{ZakPotCheSubGor2015}. RTT can also be estimated using traffic already flowing between two points without requiring precise time synchronization between each point \cite{ZanArmi2013}. However, measuring the time-varying RTT actually experienced by an application is a substantial challenge, particularly when the RTT fluctuates more frequently than one can sample the path. The unpredictability of RTT might have massive effects on upper layer protocols. Therefore, it is required to design robust protocols which are more predictable across statistics, hence are more stable when RTT is unreliable. 

Feedback and coding over a broadcast erasure channel have been combined in \cite{KelDriFra2008} to optimize decoding delay when perfect feedback is available from the receivers. The achievable rate has been optimized using feedback and coding, under the condition that each received packet is either useless or can be immediately decoded by the destination \cite{KatRahHuKatMedCrow2006}. An extension of ARQ for coded networks has been proposed in \cite{SunShaMed2008} to minimize the queue size at the transmitter. This approach combines the benefits of network coding and ARQ by acknowledging degrees of freedom (DoF) instead of original packets. It enables the feedback-based control of the tradeoff between throughput and decoding delay \cite{FraLunMedPak2007}. The proposed scheme in \cite{SunShaMed2008} is robust to delayed or imperfect feedback. However, none of these examples jointly investigates the delay and throughput when the feedback is imperfect. 

For schemes requiring feedback, it is generally assumed that feedback is lossless and delay-free \cite{SunShaMed2008}, \cite{FraLunMedPak2007}. Imperfect and delayed feedback may cause unreliability in packet delivery. Inevitable feedback channel impairments and burst errors may impede the protocol stability. The situation becomes worse under RTT fluctuations along with the delayed feedback. A method of acknowledging packet delivery for retransmission protocols with unreliable feedback has been proposed in \cite{KhosVis2017}. Based on backwards composite acknowledgment from multiple packets, the scheduler can exploit the channel quality to increase reliability at the cost of a small increase in average delay. Attempts to increase feedback reliability via repetition coding might be costly to the receiver node while erroneous feedback detection may increase packet delivery latency and diminish throughput and reliability. In LTE, blind HARQ retransmissions of a packet are proposed to avoid feedback complexity and increase reliability \cite{3GPP2015Mar}. However, this approach can severely decrease resource utilization efficiency. 

Using FEC, in-order delivery delay over packet erasure channels can be reduced \cite{KarLei2014}, and the performance of SR ARQ protocols can be boosted. Delay bounds for convolutional codes have been provided in \cite{TomFitLucPedSee2014}. Packet dropping to reduce playback delay of streaming over an erasure channel has been investigated in \cite{JosKocWor2012}. Delay-optimal codes without feedback for burst erasure channels, and the decoding delay of codes for more general erasure models have been analyzed in \cite{Martinian2004}. Despite all these prior attempts, feedback and coding have been difficult to blend, and to the best of our knowledge, unreliable feedback has not been analyzed in the area of coding before.

Throughput-delay tradeoffs of low-latency communications have been studied at the physical layer. Delay-limited link capacity has been investigated in \cite{HanTse1998}, which focused on minimizing the average delay instead of the worst-case delay. Complexity of various channel coding schemes for URLLC in 5G has been investigated in \cite{SybWesJayVenVuk2016}. At the network layer, recent work includes end-to-end delay bounds in wireless networks using large deviations theory \cite{ZubLieBur2016}, edge caching \cite{BasBenZeuKadKarErDeb2015}, the use of short transmission time interval \cite{PocPedSorLauMog2016}, HARQ retransmissions to meet target reliability rate in the uplink \cite{MalHuaAnd2017} or in the downlink \cite{AnaVec2018}, non-orthogonal multiple access \cite{DaiWanYuaHanLinWan2015}, and delay-limited throughput \cite{XinLiuNAlDinPoo2018}. From a coding perspective, average delay of network coding in downlink has been studied \cite{EryOzdMed2006}. A coded ARQ scheme for delay-free feedback has been proposed in \cite{LunPakFraMedKoe2006}. To the best of our knowledge, coding has only been studied from an average delay perspective, but minimizing the worst-case delay, i.e. providing delay guarantees, is crucial in 5G system design that also supports URLLC.

%%%%%%%%%%%%%%%%%%%%%%%%%%%%%%%%%%%%%%%%%%
\subsection{Contributions}\label{contributions}
We investigate the trade-off between throughput and guaranteeable delay over packet erasure channels with unreliable and delayed feedback. By building on the uncoded baseline scheme in \cite{AusNos2007}, we propose three protocols: (i) Hybrid ARQ (HARQ) with soft combining at the receiver, where the feedback is not cumulative; (ii) Cumulative feedback-based ARQ (CF ARQ) without rate adaptation, where feedback includes the extra information regarding the previous packets; and (iii) Coded ARQ with rate adaptation based on the cumulative feedback.  

We use tiny codes, i.e. sliding window with just $2$ packets. For the Gilbert-Elliott channels, we analyze the distributions of transmission time and delay by exploiting signal-flow techniques. We provide exact closed-form expressions of throughput and delay for memoryless channels, as functions of the system parameters such as the timeout, the RTT and the packet erasure rate. Contrasting the new protocols with uncoded ARQ, we demonstrate that HARQ performs only slightly better. While CF ARQ has lower average delay and has benefits under burst erasures or high erasure rates, it does not provide a significant throughput gain (up to about 18\%). Coded ARQ can provide throughput gains up to about 40\% when the average erasure burst is higher than 3. It also provides delay guarantees, and is robust to various challenges such as imperfect and delayed feedback, high erasure rates, burst erasures, and RTT fluctuations. Coded ARQ is more predictable across statistics, hence is more stable. 
 
To the best of our knowledge, we bring signal-flow techniques for the first time, to the area of coding, melding two areas that have hitherto been quite separate. Our results permit analysis of the perennially vexing problem of accurately accounting for delay when coding.

%%%%%%%%%%%%%%%%%%%%%%%%%%%%%%%%%%%%%%%%%%%%%
\section{System Model}
\label{model}  
We consider a point-to-point channel model consisting of a sender and a receiver. As illustrated in Fig. \ref{GEchannel}, on the forward link, the sender attempts to transmit a packet to the receiver, and upon the successful reception of the packet, on the reverse link, the receiver acknowledges the sender by transmitting a feedback. Erasure errors can occur in both the forward and reverse channels. However, an ACK cannot be decoded as a NACK, and vice versa. For the convenience of the reader, we follow the notation of \cite{AusNos2007}. The status of a transmission at time $t$ is a Bernoulli random variable taking values in $\mathcal{X}=\{0,1\}$, where $0$ denotes an error-free packet, and $1$ means the packet is erased. The erasure rate $\epsilon$ is a function of channel condition. Both for the forward and reverse links we use a Gilbert-Elliott (GE) channel model \cite{Gilbert1960}, which is a binary-state Markov process $S_t$ with states $G$ (good) and $B$ (bad), i.e. $\mathcal{S}=\{G,B\}$, and probability transition matrix $\matr{P}$. The packet erasure rates in states $G$ and $B$ are $\epsilon_G$ and $\epsilon_B$, respectively. We let $\bm{\epsilon}=[\epsilon_G, \epsilon_B]$. Since the channel state is not the same as the channel observation, the process $X_t$ is a hidden Markov model (HMM)\footnote{HMM is a statistical Markov process with unobserved states \cite{BauPet1966}. Although the state is not directly observed, the output dependent on the state can be observed.}, which is driven by the process $S_t$.

The channel state information is not available at the transmitter and the receiver. Hence, the transmitter does not know the state of the forward link at time $t$, but it observes the status of the feedback at time $t-1$, which is a Bernoulli random variable. Similarly, the receiver does not know the status of the reverse link, but it observes the status of a transmission at time $t$. The joint probabilities of channel state and observation at time $t$ can be computed using the state-transition matrix of the GE channel: $\matr{P}=[p_{ij}]\in\mathbb{R}^{2\times2},\, i,j\in\mathcal{S}$, where $p_{GB}=1-p_{GG}=q$, $p_{BG}=1-p_{BB}=r$, i.e., the first and second rows correspond to the transition probabilities of states $G$ and $B$, given the channel state at time $t-1$. Solving $\pi \matr{P} = \pi$ and $\pi \matr{1} = 1$, where $\matr{1}$ is a column vector of ones, the stationary vector of $\matr{P}$ is $\pi=[\frac{r}{r+q},\,\, \frac{q}{r+q}]$. The erasure rate is $\epsilon=\pi\bm{\epsilon}^\intercal$. Given $r$, $\epsilon_G$, $\epsilon_B$, and $\epsilon$, we have $q = r \big(\frac{\epsilon_B-\epsilon_G}{\epsilon_B-\epsilon}-1\big)$. Note that $1/r$ represents the average erasure burst, and burst errors occur when $r$ is low. The joint probabilities of channel state and observation at time $t$, given the channel state at time $t-1$, are 
\begin{align}
\mathbb{P}(S_t=j, X_t=1 \vert S_{t-1}=i)
&=\mathbb{P}(S_t=j \vert S_{t-1}=i)\mathbb{P}(X_t=1 \vert S_t=j)=p_{ij}\epsilon_j,\quad i,j\in\mathcal{S}.\nonumber
\end{align}  
Let $\matr{P}_1=\matr{P}\cdot {\rm diag}\{\bm{\epsilon}\}$ be the error matrix on the forward (or reverse) link. Similarly, $\matr{P}_0=\matr{P}\cdot {\rm diag}\{\matr{1}-\bm{\epsilon}\}$ is the success matrix in either link. The entries of $\matr{P}_0$ and $\matr{P}_1$ are the joint state-transition probabilities given the channel observations \cite{AusNos2007}. Hence, the HMM can be characterized by $\{\mathcal{S},\mathcal{X},\matr{P}_0,\matr{P}_1\}$. 

Consider the forward link $\{\mathcal{S}^{(f)},\mathcal{X}^{(f)},\matr{P}_0^{(f)},\matr{P}_1^{(f)}\}$ and the reverse link $\{\mathcal{S}^{(r)},\mathcal{X}^{(r)},\matr{P}_0^{(r)},\matr{P}_1^{(r)}\}$ that are mutually independent. The composite channel is characterized by $\{\mathcal{S}^{(c)},\mathcal{X}^{(c)},\matr{P}_{00}^{(c)},\matr{P}_{01}^{(c)},\matr{P}_{10}^{(c)},\matr{P}_{11}^{(c)}\}$, where $\mathcal{S}^{(c)}=\mathcal{S}^{(f)}\times \mathcal{S}^{(r)}$ are the composite channel states, i.e. the Cartesian product of forward and reverse states, and $\mathcal{X}^{(c)}=\mathcal{X}^{(f)}\times \mathcal{X}^{(r)}=\{00,01,10,11\}$ is the combined observation set. Note that $X_t^{(c)}=00$ means both the forward and reverse channels are good, while $X_t^{(c)}=10$ means the forward channel is erroneous and the reverse channel is good. For $X_t^{(c)}=11$, the joint probability of the combined observation and the composite state at time $t$, given the composite state at time $t-1$, is
\begin{align}
\mathbb{P}(S_t^{(c)}=(j,m), X_t^{(c)}=11\vert S_{t-1}^{(c)}=(i,k))
=(p_{ij}^{(f)}\epsilon_j^{(f)})\cdot (p_{km}^{(r)}\epsilon_m^{(r)}).
\end{align}
Using the Kronecker product notation $\otimes$, we have $\matr{P}_{ij}^{(c)}=\matr{P}_{i}^{(f)} \otimes \matr{P}_{j}^{(r)}$ for the combined observation at time $t$, i.e., $X_t^{(c)}=ij$, $i,j\in\mathcal{X}$. We assume that both the forward and the reverse channels have the same parameters\footnote{The forward and reverse channels do not necessarily have the same parameters. In practice, data packets and feedback packets typically have different lengths and different coding levels. One can use the same method to obtain the results for different channel parameters \cite{MalMedYeh2018}. Furthermore, data packets generally travel downstream from the sender towards receivers, and feedback packets travel upstream from receivers to the sender \cite{LubVicGemRizHanCrow2002}. Hence, to compensate the channel asymmetry, more bandwidth can be allocated on downlink.} $r$, $\epsilon_G$, $\epsilon_B$, and $\epsilon$. Hence, the state-transition matrix for both the forward and reverse channels is given by $\matr{P}$. In the rest of the paper, we will drop the superscript $^{(c)}$ and denote the observation probability matrices by $\matr{P}_{00}$, $\matr{P}_{01}$, $\matr{P}_{10}$ and $\matr{P}_{11}$. Similar to \cite{AusNos2007}, let $\matr{P}_{0x}=\matr{P}_{00}+\matr{P}_{01}$ and $\matr{P}_{1x}=\matr{P}_{10}+\matr{P}_{11}$ be the success and error probability matrices on the forward channel, respectively, and let $\matr{P}_{x0}=\matr{P}_{00}+\matr{P}_{10}$ and $\matr{P}_{x1}=\matr{P}_{01}+\matr{P}_{11}$ be the success and error probability matrices on the reverse channel, respectively. The matrices $\matr{P}$, $\matr{P}_0$, and $\matr{P}_1$ will denote the composite channel matrices, i.e. the Kronecker product of the forward and reverse channel matrices. The matrices for the GE channel are provided in Appendix \ref{GEmatrices}. %\cite[Appendix A]{MalMedYeh2018tinycodesarxiv}. 

%%%%%%%%%%%%%%%%%%%%%%%%%%%%%%%%%%%%%%%%%%%%% 
\section{Analysis of ARQ}
\label{matrixflowgraphs} 
In this section, we describe the protocol for the proposed channel model, signal-flow graphs, as well as a primer on the MSFGs for throughput and delay of ARQ protocols. We analyze the throughput and guaranteeable delay of uncoded ARQ, and provide exact expressions for memoryless channels. 

%%%
\subsection{Protocol}\label{protocoldetails}
We use a slotted Selective Repeat (SR) ARQ protocol for data transmission. With SR ARQ, the sender sends a number of packets specified by a window size without the need to wait for individual ACK from the receiver. SR ARQ allows the receiver to accept packets out of order, which can be stored in a buffer and sorted at the receiver to ensure in-order delivery\footnote{If there is full feedback, ARQ achieves 100\% throughput and the lowest possible packet delay over an erasure channel, and it is composable across links \cite{SunShaMed2008}. However, when the network is lossy, link-by-link ARQ cannot achieve the capacity of a general network.}. The receiver may selectively reject the packets, and the sender individually retransmits packets that have timed out. All data packets are available at the transmitter prior to any transmission, and the receiver does not have buffer overflows. There is a handshake mechanism between the sender and receiver that initiates a synchronous transmission. After the start of transmission, the RTT is $k$ slots, i.e. it takes $k-1$ time slots between the transmission of a packet and receipt of its feedback.  

At the sender, when a packet is (re)transmitted, the timeout associated with this packet is set to $T$, which is greater than or equal to the RTT $k$. Upon the reception of the first feedback, the waiting will be aborted after the timer expires, i.e., after $d=T-k$ slots. The feedback -- ACK/NACK sent by the receiver indicating if it has correctly received a data packet -- includes the information about all correctly received packets. The ACK/NACK is sent in each slot. Thus, the packet whose ACK is lost will be acknowledged by the subsequent ACKs/NACKs. If a succeeding ACK/NACK is successfully received before the timeout, the packet will not be retransmitted. Otherwise, the sender retransmits the packet until it receives an ACK. Hence, we do not have an upper bound on the maximum number of retransmissions of the same packet to guarantee its reliable delivery. If a packet is lost, the packet will be retransmitted immediately (if its NACK is received), or after the timer expires (if the NACK is also lost). The protocol is shown in Fig. \ref{protocol}.

\begin{figure*}[t!]
\centering
\begin{minipage}[t]{.49\textwidth}
\includegraphics[width=\textwidth]{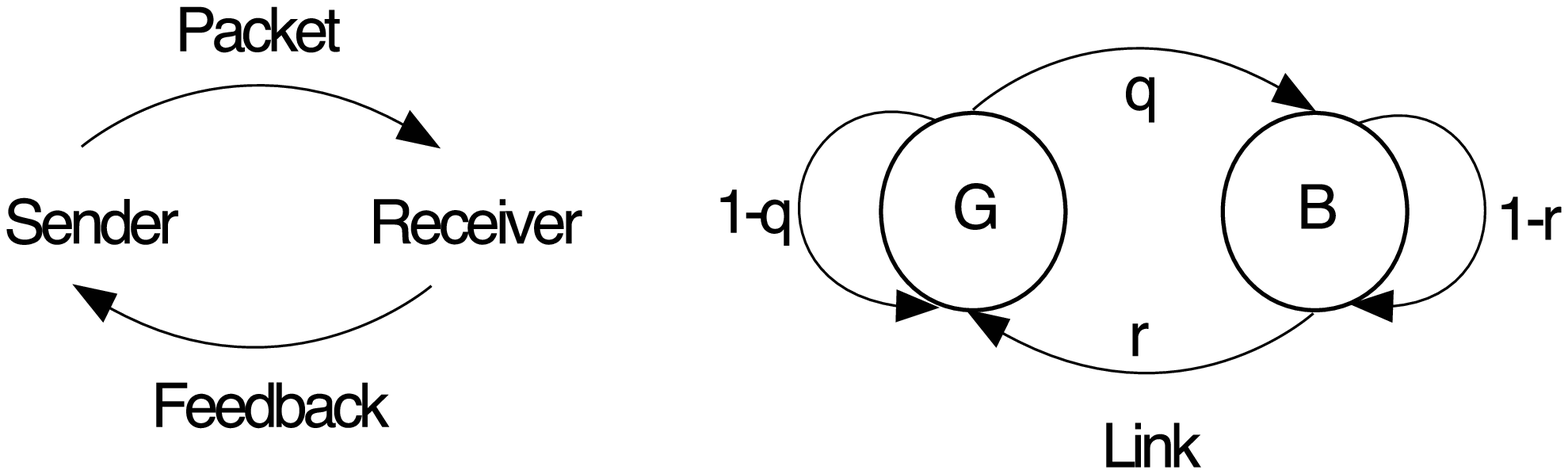}
\caption{\small{Point-to-point channel model.}\label{GEchannel}}
\end{minipage}
\centering
\begin{minipage}[t]{.49\textwidth}
\includegraphics[width=\textwidth]{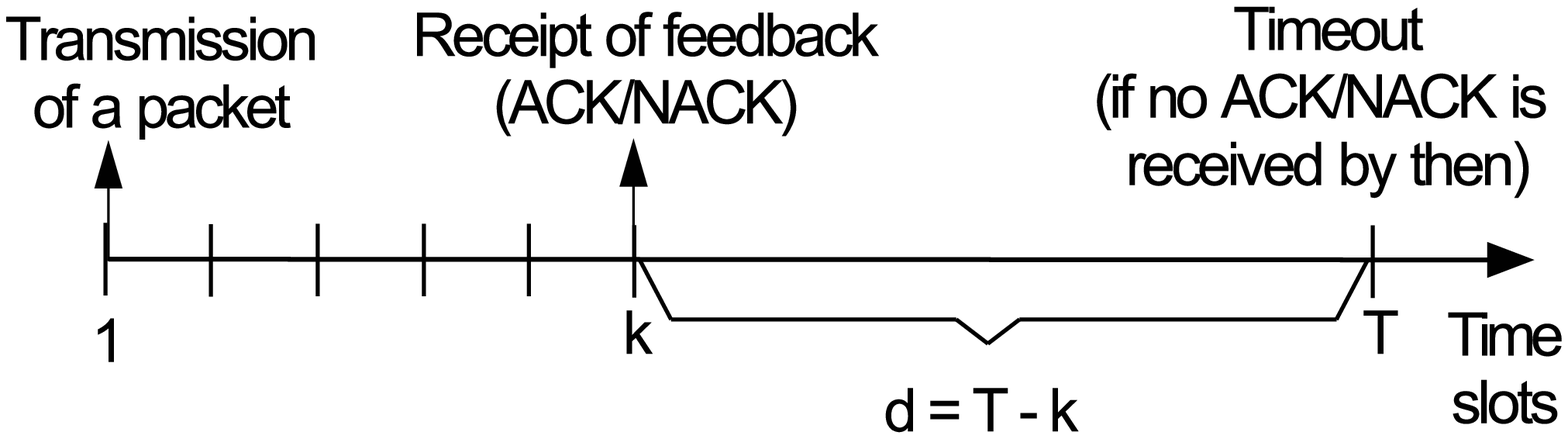}
\caption{\small{SR ARQ protocol description.}\label{protocol}}
\end{minipage}
\end{figure*}

%%%
\subsection{A Primer on Signal-Flow Graphs}\label{signalflowprimer} 
A signal-flow graph is a diagram that consists of a set of nodes that denote the different states of the system, and a set of directed branches that represent the functional relationships among the states. The analysis of finite-state HMMs can be streamlined by using signal-flow graphs, and labeling the branches of flow graphs with observation probabilities \cite{MasZim1960}, \cite{Howard1971}. We next detail how to build the flow graphs for the analysis of SR ARQ.

In the current paper, the nodes of the flow graphs correspond to the states of the transmitter. Upon the initial state that a new packet is transmitted (input node $I$), the transmitter goes from one state to the other. The output node ($O$) represents correct reception of ACK by the sender, and other nodes are hidden states. A certain value for the random variable $X$, that for example models the transmission or delay time for ARQ protocols as in \cite{AusNos2007}, \cite{LuChang1989,LuChang1993,ChoUn1994}, corresponds to a state transition. The value of $X$ along with its probability $p$ appear in the branch gain as $pz^X$. Hence, the input-output gain of the graph is a polynomial in $z$, whose coefficients are the probabilities of corresponding values of $X$. This polynomial denotes  the probability-generating function (PGF) for $X$, i.e., $\mathbb{E}[z^X]$. Flow graphs with vector node values and branches labeled with observation probability matrices are called matrix signal-flow graphs\footnote{MSFGs have been extensively used in the state-space formulation of feedback theory \cite{Chen1991}. They can also be used to model channel erasures, incorporating unreliable feedback.} (MSFGs) \cite{AusNos2007}. The graph can be simplified using the basic equivalence operations, i.e. parallel, series, and self-loop, and the matrix gain can be computed. Then, the input-output relationship is given by the matrix-generating function (MGF) $\mathbf{\Phi}(z)$.

%%%%%%%%%%%%%%%%%%%%%%%%%%%%%%%%%%%%%%%%%%
\subsection{Probability Distributions of the Transmission Time and the Delay}\label{probdistributionstransmissiontimeanddelay}
We derive the MGFs for the transmission time and the delay of different SR ARQ  protocols. The transmission time $\tau$ is defined as the number of packets transmitted per successful packet, while the delay $D$ is the time from when a packet is first transmitted to when its ACK is successfully received at the sender. Both $\tau$ and $D$ are random variables with positive integer outcomes. The PGFs $\Phi_{\tau}(z)$ and $\Phi_D(z)$ of $\tau$ and $D$ are derived using their MGFs by pre- and postmultiplications of row and column vectors, respectively. We will discuss how to obtain the MSFGs and the MGFs for the transmission time and delay in Sect. \ref{nocodingnosoftcombining}. We now discuss in detail how to obtain $\Phi(z)$'s from $\matr{\Phi}(z)$'s.

For the GE channel model, the probability of transmitting a new packet depends on the channel state. From Fig. \ref{GEchannel}, given that $\epsilon_G=0$, the probability of transmitting a new packet in state $G$ is $\pi_G (1-q)+\pi_B r$. Similarly, the probability of transmitting a new packet in state $B$ is $(\pi_G q+\pi_B (1-r)) (1-\epsilon_B)$. Therefore, the probability vector of transmitting a new packet is 
\begin{align}
\pi_I=\pi \matr{P}_0=[\pi_G (1-q)+\pi_B r, \, (\pi_G q+\pi_B (1-r)) (1-\epsilon_B)].\nonumber
\end{align}
In Appendix \ref{GEmatrices}, %\cite[Appendix A]{MalMedYeh2018tinycodesarxiv}, 
we detail how to compute the probability vectors $\matr{P}_0$ and $\matr{P}_1$ for the GE channel model as well as the probabilities for the memoryless channel model.

\begin{prop}
\label{throughputanalysis} 
{\bf Distribution of transmission time $\tau$ \cite{AusNos2007}.}
The PGF of the transmission time $\tau$ is
\begin{align}
\label{generatingfunctiontransmissiontime}
\phi_{\tau}(z)=\frac{\pi_I \matr{\Phi}_{\tau}(z) \matr{1}}{\pi_I \matr{1}}
=\frac{1}{1-\epsilon} \pi \matr{P}_0 \matr{\Phi}_{\tau}(z) \matr{1},
\end{align}
where $\matr{\Phi}_{\tau}(z)$ is the MGF of $\tau$, $\pi_I=\pi \matr{P}_0$ is the probability vector of transmitting a new packet, and $\matr{1}$ is a column vector of ones.
\end{prop}

The average transmission time $\bar{\tau}$ is found by evaluating the first derivative of $\phi_{\tau}(z)$ at $z=1$, i.e. $\bar{\tau}=\phi'_{\tau}(1)$. We define the throughput $\eta$ as the reciprocal of the average transmission time, i.e., $\eta=1/\bar{\tau}$, which is indeed a lower bound on the actual throughput $\mathbb{E}[1/\tau]$ due to the convexity of $1/\tau$, $\tau\geq 0$.

\begin{cor}
{\bf Throughput for memoryless channels.} 
When both the forward and reverse links are memoryless, $\eta = 1/\matr{\Phi}'_{\tau}(1)$ since $\pi=1$, $\matr{P}_0=1-\epsilon$, and $\phi_{\tau}(z)=\matr{\Phi}_{\tau}(z)$.      
\end{cor}         
         
%%%%%%%%%%%%%%%%%%%%%%%%%%%%%%%%%%%%%%%%%%
\begin{prop}
\label{delayanalysis}  
{\bf Distribution of delay $D$ \cite{AusNos2007}.} 
The PGF of the delay $D$ is given as
\begin{align}
\label{generatingfunctiondelay}
\matr{\phi}_D(z)=\frac{\pi_{\matr{I}} \matr{\Phi}_D(z) \matr{1}}{\pi_{\matr{I}} \matr{1}},
\end{align}
where $\matr{\Phi}_D(z)$ is the MGF of the delay. The average delay $\bar{D}$ is found by evaluating the first derivative of the PGF $\phi_D(z)$ at $z=1$, i.e. $\bar{D}=\phi'_{D}(1)$.
\end{prop}

\begin{cor}
{\bf Average delay for memoryless channels.} 
When both the forward and reverse links are memoryless, $\bar{D} =\matr{\Phi}'_{D}(1)$ since $\pi=1$, $\matr{P}_0=1-\epsilon$, and $\phi_{D}(z)=\matr{\Phi}_{D}(z)$.  
\end{cor}   

In reality, the feedback is lossy and delayed, burst errors occur, and the fluctuations in the RTT can cause a high variability in the delay. To understand these effects, we exploit the three-sigma rule\footnote{Even for non-normally distributed variables, at least 88.8\% of cases should fall within properly calculated three-sigma intervals, which follows from Chebyshev's Inequality. For unimodal distributions, the probability of being within the interval is at least 95\% \cite{Pukelsheim1994}.}. The $3\sigma$ heuristic is justifiable when the distribution of the delay is sub-Gaussian. If there are constants $C>0$, $v>0$ such that $\mathbb{P}(D>d) \leq C e^{-vd^2}$ for every $d>0$, then the probability distribution of a random variable $D$ is sub-Gaussian. A sub-Gaussian distribution has strong tail decay property since the tails decay at least as fast as the tails of a Gaussian. In this case, the guaranteeable delay $\hat{D}$ of a protocol is upper bounded by the guaranteeable delay of a Gaussian distribution with the same mean and variance as the distribution of $D$. Later in Sect. \ref{sims}, we demonstrate via numerical simulations that the tails of the delay distribution are dominated by the tails of a Gaussian distribution.
\begin{defi}
\label{GuarDelay}
{\bf Guaranteeable delay.} The guaranteeable delay of the ARQ protocol -- given that the distribution of the delay is sub-Gaussian -- is defined as
\begin{align}
\label{guaranteeabledelay}
\hat{D}=\bar{D}+3\sigma_D,
\end{align}
where $\bar{D}$ is the average delay, and $\sigma^2_D$ is the variance of the delay, which is calculated as $\sigma^2_D=\phi''_{D}(1)+\bar{D}-\bar{D}^2$, where the term $\phi''_{D}(1)$ is the second derivative of $\phi_D(z)$ evaluated at $z=1$.
\end{defi}

%%%%%%%%%%%%%%%%%%%%%%%%%%%%%%%%%%%%%%%%%%%%%%%

%%%
\subsection{Selective-Repeat ARQ}
\label{nocodingnosoftcombining}
Each packet is independently transmitted and acknowledged. Hence, it suffices to consider a MSFG of SR protocol for a single packet  \cite{AusNos2007}. The HMMs for throughput and delay analysis of uncoded SR ARQ in unreliable feedback are detailed in \cite{AusNos2007}, which is the baseline model for our paper. The flow graphs are illustrated in Fig. \ref{NoCodingnoHARQthroughputanddelay}. In the flow graph for the transmission time, as shown in Fig. \ref{NoCodingnoHARQthroughputanddelay}-(a), every time a packet is transmitted, the branch gain is multiplied by $z\matr{P}^{k-1}$ since transmission time is defined as the number of packets transmitted per successful packet. In the delay graph, however, as shown in Fig. \ref{NoCodingnoHARQthroughputanddelay}-(b), every time a packet is transmitted, the branch gain is multiplied by $z^{k-1}\matr{P}^{k-1}$ because the RTT is $k$ slots. In these MSFGs, nodes $I$ and $O$ represent the input and output nodes, and nodes $A$, $B$, $C$, $G$ denote the hidden states. The possibilities are:
\begin{itemize}
\item {\bf State $I$.} This state represents transmission of a new packet by the sender. The probability vector of transmitting a new packet is $\pi_I=\pi \matr{P}_0$ (see %\cite[Appendix A]{MalMedYeh2018tinycodesarxiv}).
Appendix \ref{GEmatrices}).
\item {\bf Transition to state $A$.} After sending a new packet, the transmitter receives a feedback message $k-1$ time slots later. This state is represented by node $A$. The branch gains for $\tau$ and $D$ are
\begin{align}
\label{IA}
\tbg(I\to A)=z\matr{P}^{k-1},\quad \dbg(I\to A)=z^{k-1}\matr{P}^{k-1}.
\end{align}
\item {\bf Transition to state $O$.} If the feedback is an error-free ACK, which occurs with probability $\matr{P}_{00}$, or if it is an erroneous ACK but an error-free ACK/NACK is received before timer expiration, which occurs with probability $\sum\nolimits_{j=0}^{d-1}{\matr{P}_{01}\matr{P}_{x1}^{j}\matr{P}_{x0}}$, then the system transits to state $O$ and the packet is removed from the system. The branch gains for $\tau$ and $D$ are
\begin{align}
\label{AO}
\tbg(A\to O)=\matr{P}_{00}+\sum\nolimits_{j=0}^{d-1}{\matr{P}_{01}\matr{P}_{x1}^{j}\matr{P}_{x0}},\quad \dbg(A\to O)=z\matr{P}_{00}.
\end{align}

\begin{figure*}[t!]
\centering
\includegraphics[width=1\textwidth]{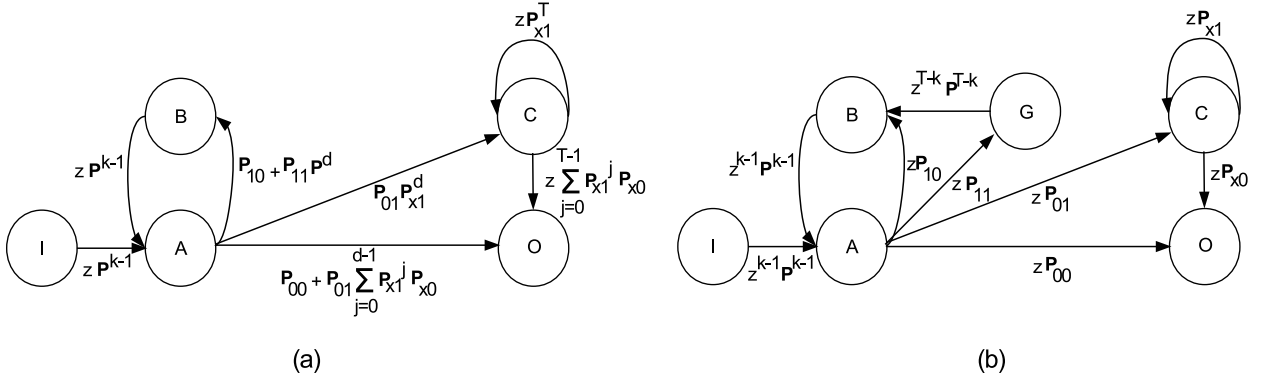}
\caption{\small{MSFGs for (a) throughput and (b) delay analysis in unreliable feedback for uncoded ARQ.}\label{NoCodingnoHARQthroughputanddelay}}
\end{figure*}

\item {\bf Transition to state $B$.} 
This state represents retransmission of an erroneous packet. Denote by $\pi_B$ its probability vector. For the HMM of the throughput, if the feedback is an error-free NACK, with probability $\matr{P}_{10}$, or a NACK is lost and the timer expires, with probability\footnote{When a NACK is lost, the transmitter waits for timeout to retransmit the corresponding packet, which involves a delay of $d=T-k$.} $\matr{P}_{11}\matr{P}^d$, the system goes to state $B$. Hence, the branch gain for $\tau$ is
\begin{align}
\label{ABt}
\tbg(A\to B)=\matr{P}_{10}+\matr{P}_{11}\matr{P}^d. 
\end{align}
In the delay graph however, state $B$ represents only the reception of an error-free NACK at $A$: 
\begin{align}
\label{ABd}
\dbg(A\to B)=z^{k-1}\matr{P}^{k-1}z\matr{P}_{10}.
\end{align}
The loop between $A$ and $B$ models retransmission of the erroneous packet until correctly received.
\item {\bf Transition to state $G$ (delay graph).\footnote{State $G$ is only relevant in the delay graph because losing a NACK causes an additional delay since the sender waits till timeout. However, in the throughput graph, we look at the number of packets transmitted per successful packet. When the forward link is bad, irrespective of whether the NACK is successfully received at the sender or is lost, the packet has to be retransmitted. Therefore, the states $G$ and $B$ can be clumped together into a single state $B$ as shown in Fig. \ref{NoCodingnoHARQthroughputanddelay}-(a).}} If the feedback is an erroneous NACK, the transmitter waits for timeout. The packet will be retransmitted after the timer expires, which involves a delay (transition from $G$ to $B$). Hence, the branch gain for $D$ is
\begin{align}
\label{AG}
\dbg(A\to G)=z^{k-1}\matr{P}^{k-1}z\matr{P}_{11}z^{T-k}\matr{P}^{T-k}.
\end{align} 
In our delay analysis, we use the following shorthand notation to incorporate the states $B$ and $G$:
\begin{align}
\dbg(A\to B,\, G)=z\matr{P}_{1x}^{\rm D}=z^{k}\matr{P}^{k-1}\matr{P}_{10}+z^{T}\matr{P}^{k-1}\matr{P}_{11}\matr{P}^{T-k}.
\end{align}
\item {\bf Transition to state $C$.} State $C$ represents retransmission of a packet that was correctly received, but with its timer expiring. Denote by $\pi_C$ its probability vector. If the feedback is an erroneous ACK and the timer expires before receiving any error-free ACKs/NACKs, the system transits to state $C$, the packet is retransmitted, modeled by the self-loop at state $C$ (the self-loop represents the delay from losing subsequent ACKs/NACKs), and the timeout is reset. The packet is acknowledged when a succeeding ACK/NACK is correctly received. Hence, the branch gains for $\tau$ and $D$ are
\begin{align}
\label{AC}
\tbg(A\to C)=\matr{P}_{01}\matr{P}_{x1}^d,\quad \dbg(A\to C)=z\matr{P}_{01}.
\end{align}
Finally, incorporating the self-loop at $C$, the branch gains for $\tau$ and $D$ from state $C$ and $O$ are
\begin{align}
\label{CO}
\tbg(C\to O)=(\matr{I}-z\matr{P}_{x1}^T)^{-1}z\sum\nolimits_{j=0}^{T-1}{\matr{P}_{x1}^{j}\matr{P}_{x0}},\quad
\dbg(C\to O)=(\matr{I}-z\matr{P}_{x1})^{-1}z\matr{P}_{x0}.
\end{align}
\end{itemize}

Referring to the MSFG for throughput analysis of uncoded ARQ in Fig. \ref{NoCodingnoHARQthroughputanddelay}-(a), a packet will be transmitted only in states $I$, $B$, and $C$. The probability vectors of states $I$, $B$, and $C$ are denoted by $\pi_I$, $\pi_B$, and $\pi_C$, respectively. These vectors can be found by solving the following equations \cite{AusNos2007}:
\begin{align}
\label{stationarydistribution}
\pi_B&=(\pi_I+\pi_B)\matr{P}^{k-1}(\matr{P}_{10}+\matr{P}_{11}\matr{P}^{T-k}),\nonumber\\
\pi_C&=(\pi_I+\pi_B)\matr{P}^{k-1}\matr{P}_{01}\matr{P}_{x1}^{T-k}+\pi_C\matr{P}_{x1}^T,
\end{align}
where $\pi$ satisfies $\pi=\pi_I+\pi_B+\pi_C$, which comes from the fact that the transmitter always has a packet to transmit. Solving for $\pi_I$ from the system (\ref{stationarydistribution}), the PGF $\phitauU(z)$ is derived using  (\ref{generatingfunctiontransmissiontime}). We refer the reader to %\cite[Appendices B, C]{MalMedYeh2018tinycodesarxiv} 
Appendices \ref{App:AppendixavgtransmissiontimeuncodedARQ} and \ref{App:AppendixavgdelayuncodedARQ}
for the derivations of the MGFs $\PhitauU(z)$ and $\PhiDU(z)$ \cite{AusNos2007}.

In the following, we investigate various extensions of uncoded ARQ, and analyze the throughput $\eta$ and the guaranteeable delay $\hat{D}$ by deriving the PGFs using MSFGs. Our analysis is based on the GE channel model. However, to simplify notation and have a better understanding of the main results in Sections \ref{matrixflowgraphs}-\ref{codedARQ}, throughput and delay results are given in closed form for memoryless channels only. Note that one may follow the analysis and obtain the results for the GE channel model, as detailed in Sect. \ref{model}. Later in Sect. \ref{sims}, based on the GE channel model, we will provide a numerical comparison of different ARQ protocols in terms of their throughputs and delay guarantees.

%%%%%%%%%%%%%%%%%%%%%%%%%%%%%%%%%%%%%%%%%%%%%
\section{Uncoded Hybrid ARQ with Soft Combining}
\label{nocodingHARQ}   
The Hybrid ARQ (HARQ) protocol with soft combining is a repetition-based uncoded transmission scheme, in which incorrectly received packets are stored, and the (re)transmitted packets are combined at the receiver \cite{DahlParkSkoBem2008}. While it is possible that two given transmissions cannot be independently decoded without error, the combination of the previously erroneously received transmissions may give enough information to be correctly decoded. Hence, this protocol is an improvement on uncoded ARQ in which incorrectly received packets are discarded \cite{AusNos2007}.

We use HARQ with Chase combining such that every (re)transmission contains the same data bits. Because all transmissions are identical, this protocol can be seen as additional repetition coding. Every (re)transmission adds extra energy to the received transmission through an increased $E_b/N_0$. The receiver uses maximal-ratio combining (MRC) to combine the received bits with the same bits from previous transmission attempts, to successfully decode the transmitted packet. However, since it is a repetition-based data transmission method, it is suboptimal. Different methods, such as incremental redundancy, can be used such that multiple coded packets are generated, each representing the same data packet, and the (re)transmission uses a different coded packet than the previous transmission \cite{SacBetNisApeUsh2015}. Thus, at every (re)transmission the receiver gains extra information.
 
\begin{figure}[t!]
\centering
\includegraphics[width=0.5\textwidth]{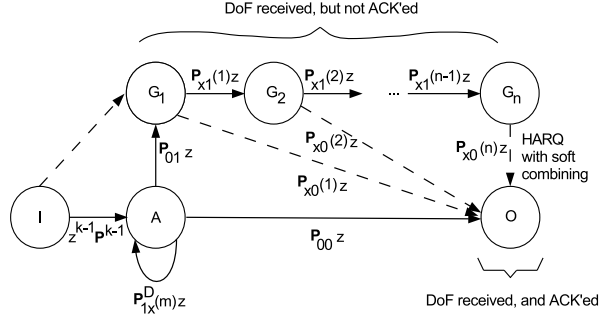}
\caption{\small{MSFG for delay analysis of Hybrid ARQ with soft combining at the receiver.}\label{NoCodingwithHARQ}}
\end{figure}

We illustrate the HMM for the HARQ scheme with soft combining in Fig. \ref{NoCodingwithHARQ}. The proposed model improves the chance of successful reception at every attempt. Therefore, the erasure rates decrease at each retransmission attempt. We assume that the erasure rate for state $G$ is $\epsilon_G(m) =0$, and to compute the erasure rate for state $B$, we exploit the Chase combiner output signal-to-noise ratio (SNR). We have shown in \cite[Proposition 1]{MalHuaAnd2017} that the combiner output SNR for a total of $m$ transmissions is $\textrm{SNR}(m)=\rho m$, where $\rho=P_td^{-\beta}/\sigma^2$ is the received SNR at the receiver per transmission, where $P_t$ is the transmit power, $\sigma^2$ is the received noise power, $d$ is the distance between the sender and the receiver, $\beta$ is the path loss exponent. We then have the following relation:
\begin{align}
\epsilon_B(m)=\mathbb{P}(\textrm{SNR}(m)<\gamma) = \mathbb{P}(m P_t d^{-\beta}h/\sigma^2   <\gamma)
=\mathbb{P}(h <\gamma P_t^{-1} \sigma^2 d^{\beta}/m ) 
=1-e^{-\mu\gamma P_t^{-1} \sigma^2 d^{\beta}/m }. \nonumber
\end{align}
where $\textrm{SNR}(m)$ is the soft combined SNR at the receiver as a result of $m$ retransmissions, $\gamma$ is the decoding threshold, and $h\sim \exp(\mu)$ is the channel power gain with parameter $\mu$ in the presence of Rayleigh fading. Using the relation $\rho=P_t d^{-\beta}/\sigma^2$ and clumping all the parameters into $\alpha=\mu\gamma/\rho$, we can obtain that the erasure rate for state $B$ satisfies the relation $\epsilon_B(m)=1-e^{-\alpha/m}$ as a function of the retransmission attempt $m$, where the parameter $\alpha$ controls the erasure rate. As $\alpha\to 0$ ($\rho\to\infty$), $\epsilon_B(m)=0$, and as $\alpha\to\infty$ ($\rho\to 0$), $\epsilon_B(m)=1$. Hence, as $\alpha$ decreases or $m$ increases, $\epsilon_B(m)$ drops.

On a retransmission attempt $m$, the branch gains for $\tau$ and $D$ for receiving an error-free NACK by the end of the RTT, or receiving an erroneous NACK before the timer expires, equal
\begin{align}
\tbg(A\to B)&=\matr{P}_{10}(m)+\matr{P}_{11}(m)\matr{P}^d,\nonumber\\
\dbg(A\to B,\, G)&=z\matr{P}_{1x}^{\rm D}(m)=z^{k}\matr{P}^{k-1}\matr{P}_{10}(m)+z^{T}\matr{P}^{k-1}\matr{P}_{11}(m)\matr{P}^{T-k},
\end{align}
where $\matr{P}_{xy}(m)$ is the composite channel matrix for $X_m^{(c)}=xy$ on attempt $m$. Note that in the uncoded scheme, the matrices $\matr{P}_{x0}(i)$'s and $\matr{P}_{x1}(i)$'s do not change with the transmission attempt $i$.

For the derivation of the MGFs of the transmission and delay times of HARQ with soft combining, reader is referred to %\cite[Appendix E]{MalMedYeh2018tinycodesarxiv} and \cite[Appendix F]{MalMedYeh2018tinycodesarxiv},
Appendix \ref{App:AppendixthroughputHARQ} and Appendix \ref{App:AppendixavgdelayHARQ},
respectively. We now present the closed form expressions for throughput and average delay of memoryless channels.

\begin{prop}\label{throughputHARQ}
The throughput of HARQ for memoryless channels is given by 
\begin{align}
\label{M1HARQthroughput}
&\etaH
=1\Big/\left\{
\sum\limits_{j=0}^{\infty}\Big(\prod_{i=0}^j{\epsilon(i)}\Big)
\left[(1-\epsilon(j^*))\epsilon(j^*)\Big(\prod\limits_{i=0}^d{\epsilon(j^*+i)}\Big)\right.\right.\nonumber\\
&\left.\times\sum\limits_{j=0}^{\infty}(j+1)\prod\limits_{i=1}^j\Big(\prod\limits_{l=(i-1)T+1}^{iT}{\epsilon(j^*+l)}\Big)
\sum\limits_{j'=0}^{T-1}{\Big(\prod\limits_{i=0}^{j'}\epsilon(j^*+jT+i)\Big)}(1-\epsilon(j^*+jT+j'+1))\right] \nonumber\\
&+
\Big(\sum\limits_{j=0}^{\infty}(j+1)\prod_{i=0}^j{(\epsilon(i))}\Big)
\left[(1-\epsilon(j^*))^2
+(1-\epsilon(j^*))\epsilon(j^*)\sum\limits_{j=0}^{d-1}{\Big(\prod\limits_{i=0}^j{\epsilon(j^*+i)}\Big)}(1-\epsilon(j^*+j+1))\right.\nonumber\\
&+(1-\epsilon(j^*))\epsilon(j^*)\Big(\prod\limits_{i=0}^d{\epsilon(j^*+i)}\Big)
\sum\limits_{j=0}^{\infty}\prod\limits_{i=1}^j\Big(\prod\limits_{l=(i-1)T+1}^{iT}{\epsilon(j^*+l)}\Big)\nonumber\\
&\left.\left.\times\sum\limits_{j'=0}^{T-1}{\Big(\prod\limits_{i=0}^{j'}\epsilon(j^*+jT+i)\Big)}(1-\epsilon(j^*+jT+j'+1))\right]\right\},       
\end{align}
where $\epsilon(0)=1$, $\epsilon(i)$ is the channel erasure rate at retransmission attempt $i$, and $j^*$ is the required number of forward (re)transmissions for successful decoding.
\end{prop}

\begin{proof}
See Appendix \ref{App:AppendixthroughputHARQ}.
\end{proof}

\begin{prop}\label{avgdelayHARQ}
The average delay of HARQ for memoryless channels is given by \cite{AusNos2007} 
\begin{align}
\label{M1HARQdelay}
\DeH&=(1-\epsilon(j^*))
\times\left\{\Big((1-\epsilon(j^*))+\epsilon(j^*)\sum\limits_{j'=0}^{\infty}\Big(\prod\limits_{i=0}^{j'}\epsilon(j^*+i)\Big)(1-\epsilon(j^*+j'+1))\Big)\right.\nonumber\\
&\times\Big(\sum\limits_{j=0}^{\infty}\sum\limits_{i^*=0}^j (k\epsilon(i^*)(1-\epsilon(i^*))+T\epsilon(i^*)^2)\Big(\prod\limits_{i=0 \neq i^*}^{j}\epsilon(i)\Big)\Big)\nonumber\\
&+
\Big(\sum\limits_{j=0}^{\infty} \prod\limits_{i=0}^j \epsilon(i)\Big)\left[k(1-\epsilon(j^*))+(k+1)\Big(\epsilon(j^*)\sum\limits_{j'=0}^{\infty}\Big(\prod\limits_{i=0}^{j'}\epsilon(j^*+i)\Big)(1-\epsilon(j^*+j'+1))\Big)\right.\nonumber\\
&\left.\left.+\epsilon(j^*)\sum\limits_{j'=0}^{\infty}j'\Big(\prod\limits_{i=0}^{j'}\epsilon(j^*+i)\Big)(1-\epsilon(j^*+j'+1))\right]\right\},
\end{align}
where $\epsilon(0)=1$, $\epsilon(i)$ is the channel erasure rate at retransmission attempt $i$, and $j^*$ is the required number of forward (re)transmissions for successful decoding.
\end{prop}

\begin{proof}
See Appendix \ref{App:AppendixavgdelayHARQ}.
\end{proof}

Uncoded ARQ is a special case of HARQ with soft combining where $\epsilon(i)=\epsilon$ for all channel uses $i\in\mathbb{Z}^+$. Following Propositions \ref{throughputHARQ} and \ref{avgdelayHARQ}, respectively, we can derive the following compact results.
 
\begin{prop}\label{throughputuncodedARQ}
The throughput $\eta$ for uncoded ARQ for memoryless channels is given by 
\begin{align}
\label{M1uncodedARQthroughput}
\etaU=\frac{1-\epsilon}{1+\epsilon^{d+1}(1-\epsilon)/(1-\epsilon^T)}.	        
\end{align}
\end{prop}

\begin{proof}
See Appendix \ref{App:AppendixavgtransmissiontimeuncodedARQ}.
\end{proof}

\begin{prop}\label{avgdelayuncodedARQ}
The average delay for uncoded ARQ for memoryless channels is given by 
\begin{align}
\label{M1uncodedARQdelay}
\DU=k+\frac{\epsilon}{1-\epsilon}(1+T\epsilon)+k\epsilon.  
\end{align}
\end{prop}

\begin{proof}
See Appendix \ref{App:AppendixavgdelayuncodedARQ}.
\end{proof}

\begin{prop}\label{secondmomentdelayuncodedARQ}
The variance of delay for uncoded ARQ for memoryless channels is 
\begin{align}
\label{variabilityuncodedARQ}
\varDU=k^2\frac{(\epsilon+\epsilon^2+\epsilon^3)}{1-\epsilon}-\frac{k}{1-\epsilon}\Big(1-2\epsilon+2\epsilon^2+2T\epsilon^2\Big(\frac{1-2\epsilon-\epsilon^2}{1-\epsilon}\Big)\Big)+\mathcal{O}(1).
\end{align}
\end{prop}

\begin{proof}
See Appendix \ref{App:AppendixsecondmomentdelayuncodedARQ}.
\end{proof}

From (\ref{M1uncodedARQthroughput}), throughput $\etaU=1/\avgtauU$ is upper bounded by $1-\epsilon$ as $T, d\to\infty$, and there is no such bound for the average delay in (\ref{M1uncodedARQdelay}). From (\ref{variabilityuncodedARQ}), we observe that the standard deviation of delay scales with the RTT $k$. Hence, RTT causes significant variability in delay.

We next consider a cumulative feedback-based SR ARQ protocol, and analyze its MSFGs.  

\begin{figure*}[t!]
\centering
\begin{minipage}[t]{.44\textwidth}
\includegraphics[width=\textwidth]{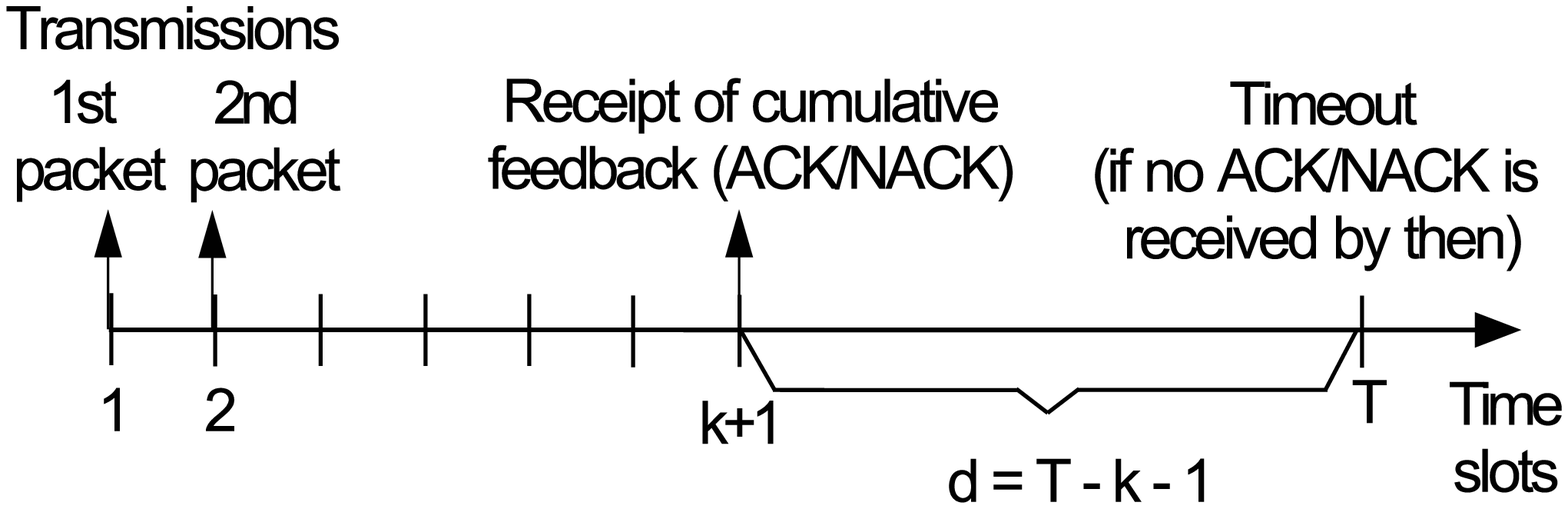}
\caption{\small{CF ARQ protocol description.}\label{CFprotocol}}
\end{minipage}
\begin{minipage}[t]{.54\textwidth}
\centering
\includegraphics[width=\textwidth]{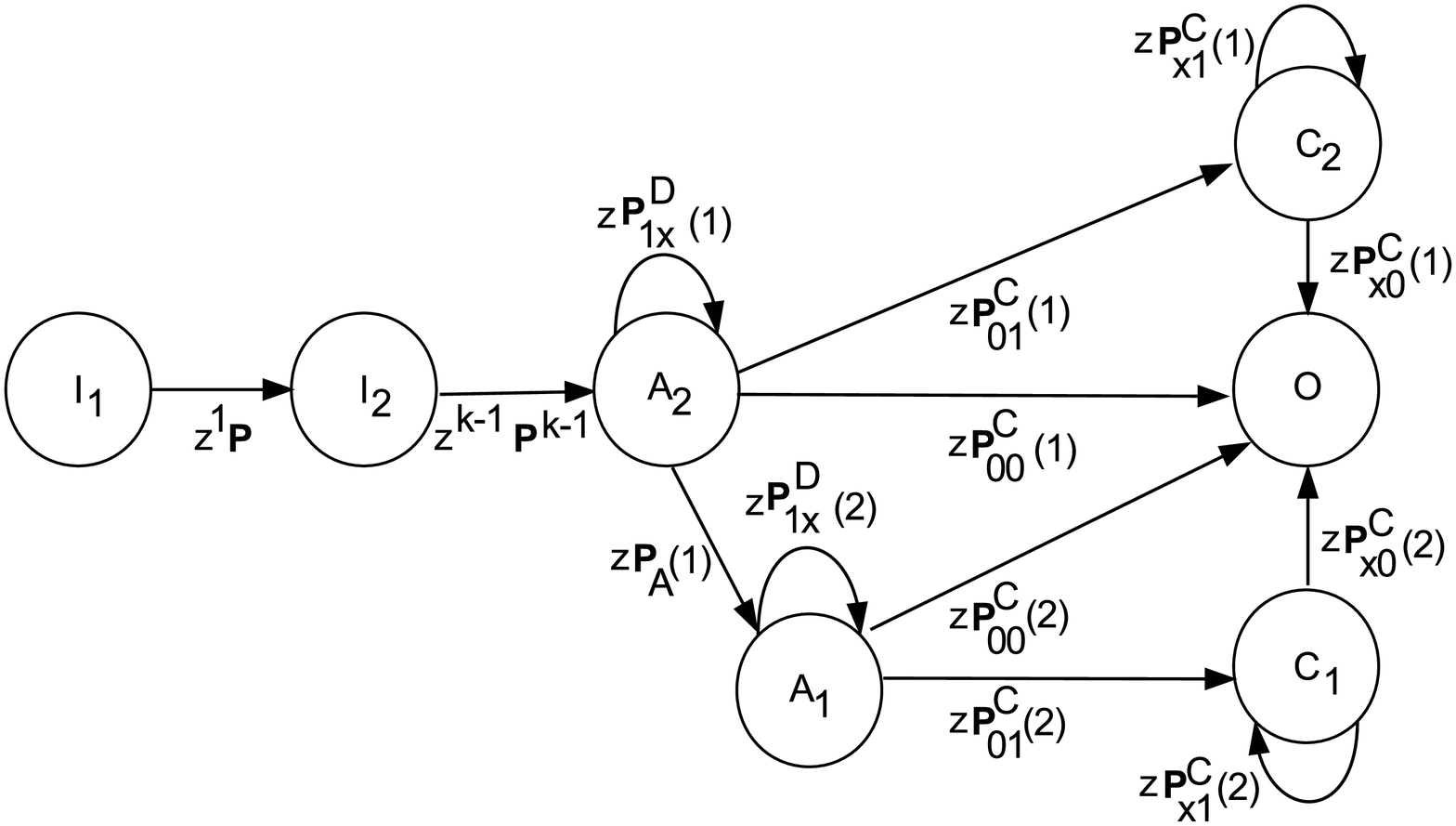}
\caption{\small{Matrix-flow graph for delay of CF ARQ.}
\label{SR_ARQ_CF}}
\end{minipage}
\end{figure*}

%%%%%%%%%%%%%%%%%%%%%%%%%%%%%%%%%%%%%%%
\section{Uncoded ARQ with Cumulative Feedback}
\label{CFcoding}
We propose a cumulative feedback-based ARQ (CF ARQ) scheme, where the transmitted packets are uncoded. The sender has a coding bucket. When it is ready to send a packet to the receiver, it transmits all the uncoded packets in the bucket to the receiver. The receiver sends a cumulative feedback to indicate the set of successfully received packets in the bucket. If the receiver successfully collects all packets in the coding bucket, and the sender successfully receives the cumulative ACK message, it then purges all the packets in the bucket and moves new packets into the bucket \cite{ZenNgMed2012}.

We consider minimum coding, i.e., with a sliding window of size $M=2$. While the protocol can be generalized to packet streams with $M>2$, this is left as future work. In this scheme, the transmitted packet stream is maximum distance separable (MDS) coded, and the feedback acknowledges all correctly received packets, and is cumulative for $M=2$ coded packets. However, the transmission scheme is repetition-based, i.e. the transmission rate is not adjusted based on the cumulative feedback. The receiver needs both coded packets to reconstruct the transmitted packet stream, i.e., the degrees of freedoms (DoFs) required at the receiver is $N=2$. We do not assume in-order packet delivery. Thus, the transmitted packets in the coding bucket will be successfully decoded when both of the coded packets are successfully received and ACK'ed by the receiver. 
 
The feedback is cumulative for $2$ packets, and it takes $k-1$ time slots between the transmission of the second packet and receipt of the feedback. Hence, the RTT of CF ARQ is ${\rm RTT}=k+1$ slots. If the feedback was not cumulative, i.e., the first feedback was received $k-1$ slots after the first packet was transmitted, then the RTT would have been $k$ slots. If the sender does not receive an ACK before the timeout, it retransmits both packets until it receives an ACK. The protocol is shown in Fig. \ref{CFprotocol}.

The combined observation set for CF ARQ with $M=2$ packets is all 3-tuples of $\mathbb{Z}_2=\{0,1\}$, i.e., $\mathcal{X}^{(c)}=\mathbb{Z}_2^3$. For example, $X_t^{(c)}=001$ means that the forward channel is good for both packets and the reverse channel is erroneous, i.e., the ACK for both packets is lost at time $t$. Since the feedback is cumulative for $M=2$ packets, it is possible that both packets are successfully acknowledged, or they both need to be retransmitted or only one of the packets has to be retransmitted.

The HMM for delay analysis of CF ARQ is shown in Fig. \ref{SR_ARQ_CF}.  The states $I_1$ and $O$ are the input and output nodes, respectively, and nodes $I_2$, $A_1$, $A_2$, $C_1$, $C_2$ represent the hidden states. States $I_1$ and $I_2$ represent transmission of the first new packet and the second packet one time slot later, respectively. The possibilities upon the transmission of $M=2$ coded packets are:
\begin{itemize}
\item {\bf Transition to state $A_2$.} Node $A_2$ denotes the reception of the first feedback. The coded packets are retransmitted until the forward link is successful and at least one packet is successfully transmitted. The retransmission is modeled by the self-loop at $A_2$, where the branch gain for delay is
\begin{align}
\dbg(A_2\to B_2,\, G_2)&=z\matr{P}_{1x}^{\rm D}(1)=z^{{\rm RTT}+1}\matr{P}^{\rm RTT}(\matr{P}_{10}^{\CF}(1)+z^d\matr{P}_{11}^{\CF}(1) \matr{P}^d),\nonumber
\end{align}
where $d=T-{\rm RTT}$ is the time for timer expiration upon the reception of the first feedback. The branch gain for delay using transition probability matrix $\matr{P}_{10}^{\CF}(1)$ for transmitting $2$ packets is
\begin{align}
\dbg(A_2\to B_2)=z\matr{P}_{10}^{\CF}(1)=z(\matr{P}_{10}\matr{P}_{10}+\matr{P}_{10}\matr{P}_{01}+\matr{P}_{01}\matr{P}_{10}), \nonumber
\end{align}
which models the error-free NACK. It combines the different cases such that the feedback is an error-free NACK, i.e., the forward link was bad for both packets and the reverse link was good (first term), or the forward link was bad for either one of the packets only and the reverse link was good (second and third terms). We assume the cumulative feedback is error-free as long as the reverse link is good before the forward transmission is over.

The transition probability matrix $\matr{P}_{11}^{\CF}(1)$ is given by  
\begin{align}
\matr{P}_{11}^{\CF}(1)&=\matr{P}_{11}\matr{P}_{11}+\matr{P}_{11}\matr{P}_{10}+\matr{P}_{10}\matr{P}_{11},\nonumber
\end{align}
which models the erroneous NACK. It combines the cases in which the feedback is an erroneous NACK, i.e., the forward link was bad for both packets and the reverse link was also bad.

In CF ARQ, unless both packets are successfully acknowledged, we always need retransmissions. Hence, it is suboptimal. Furthermore, the erasure rate of CF ARQ is not the same as the erasure rate of uncoded ARQ. For example, for the case of symmetric memoryless channels, the relationship between the erasure rate for CF ARQ with $M=2$ packets, i.e. $\epsilon_{\rm CF}$, and of the erasure rate of the uncoded ARQ in \cite{AusNos2007}, i.e. $\epsilon$, can be computed as $\epsilon_{\rm CF} = \sqrt{\epsilon^4+2\epsilon^3 (1-\epsilon)}$. Hence, $\epsilon_{\rm CF}\geq \epsilon^2$.

\item {\bf Transition to state $A_1$.} When the first feedback is received at node $A_2$, if the number of DoFs acknowledged by the receiver equals $1$, then the system transits to state $A_1$. The matrix 
\begin{align}
\dbg(A_2\to A_1)=z\matr{P}^{\CF}_A(1)=z(\matr{P}_{00}\matr{P}_{10}+\matr{P}_{10}\matr{P}_{00}+\matr{P}_{11}\matr{P}_{01}+\matr{P}_{01}\matr{P}_{11}+\matr{P}_{00}\matr{P}_{11}+\matr{P}_{11}\matr{P}_{00})\nonumber
\end{align} 
denotes the branch gain and $\matr{P}^{\CF}_A(1)$ is the transition probability matrix from $A_2$ to $A_1$. Hence, if the system goes into state $A_1$, the additional number of DoFs required by the receiver is $1$, i.e., only one packet needs to be retransmitted, which is modeled by the self-loop at $A_1$, where
\begin{align}
\dbg(A_1\to B_1,\, G_1)=z\matr{P}_{1x}^{\rm D}(2)= (z\matr{P})^{{\rm RTT}-1}
(z\matr{P}_{10}^{\CF}(2)+z\matr{P}_{11}^{\CF}(2)z^{d+1}\matr{P}^{d+1}),\nonumber
\end{align}
where the probability matrices $\matr{P}_{10}^{\CF}(2)$ and $\matr{P}_{11}^{\CF}(2)$ model the error-free and the erroneous NACK, respectively. At node $A_1$, since only one packet is retransmitted, the transition probability matrices satisfy $\matr{P}_{xy}^{\CF}(2)=\matr{P}_{xy}$, where $\matr{P}_{xy}$'s, $x,y\in\{0,1\}$ are same as the ones for uncoded ARQ in \cite{AusNos2007}. 

\item {\bf Transition to state $O$.} If $N=2$ DoF's are received, the stream can be successfully decoded. If $N=2$ DoF's are acknowledged (with probability $\matr{P}_{00}^{\CF}(1)=\matr{P}_{00}\matr{P}_{00}$), the system transits to $O$. 

\item {\bf Transition to state $C_2$.} If $N=2$ DoF's are received, but the feedback is an erroneous ACK (with probability $\matr{P}_{01}^{\CF}(1)=\matr{P}_{01}\matr{P}_{01}+\matr{P}_{01}\matr{P}_{00}+\matr{P}_{00}\matr{P}_{01}$, where the branch gain for delay satisfies $\dbg(A_2\to C_2)=z \matr{P}_{01}^{\CF}(1)$), then the system transits to $C_2$, where the sender waits till it receives an error-free ACK/NACK, modeled by the self-loop at $C_2$. 

\item {\bf Transition to state $C_1$.}  If $N=2$ DoF's are received, but only one packet is successfully acknowledged and the feedback for the other packet is an erroneous ACK (with probability $\matr{P}_{01}^{\rm C}(2)$, where $\dbg(A_1\to C_1)=z \matr{P}_{01}^{\CF}(2)$), then the system transits to $C_1$, where the sender waits till it receives an error-free ACK/NACK. This is modeled by the self-loop at $C_1$. 
\end{itemize}
 
The success and error probability matrices on the reverse channel for CF ARQ are given as
\begin{align}
\matr{P}_{x0}^{\CF}(n)=\matr{P}_{00}^{\CF}(n)+\matr{P}_{10}^{\CF}(n),\quad 
\matr{P}_{x1}^{\CF}(n)=\matr{P}_{01}^{\CF}(n)+\matr{P}_{11}^{\CF}(n),\,\, n\in\{1,2\}, \nonumber
\end{align}
respectively, given the transition probabilities, where $n-1$ is the number of DoFs acknowledged by the receiver, i.e., $2-(n-1)$ DoFs are needed at the receiver.

For the proposed scenario, both $\tauCF$ and $\DCF$ are random variables with positive integer outcomes. The matrix gain of the graph in Fig. \ref{SR_ARQ_CF} is calculated using the basic simplification rules. 

For the derivation of the MGFs of the transmission and delay times of CF ARQ, see %\cite[Appendix G]{MalMedYeh2018tinycodesarxiv} and \cite[Appendix H]{MalMedYeh2018tinycodesarxiv}. 
Appendix \ref{App:AppendixthroughputCF-ARQ} and Appendix \ref{App:AppendixavgdelayCF-ARQ}.
Using the PGFs, the throughput $\etaCF$, which is the reciprocal of the average value of $\tauCF$, and the average value of $D_{\sf{CF-ARQ}}$, i.e., $\DCF$, can be calculated. We now present the closed form expressions for throughput and average delay of the memoryless channels.

\begin{prop}\label{throughputCF-ARQ}
The throughput for CF ARQ for memoryless channels is given by
\begin{align}
\label{M2CFARQthroughput}
\etaCF\approx \frac{2(1-\epsilon)}{1+\alpha_{\CF}(\epsilon)^{-1}\epsilon^{d}(1-\epsilon )/(1-\epsilon^T)},
\end{align}
where $d=T-k$, and $\alpha_{\CF}(\epsilon)$ is given by
\begin{align}
\alpha_{\CF}(\epsilon)&=\frac{1+3\epsilon-2\epsilon^2+20\epsilon^3-18\epsilon^4+28\epsilon^5-60\epsilon^6+72\epsilon^7-40\epsilon^8+8\epsilon^9}{2(2-\epsilon)(1-\epsilon+4\epsilon^2-2\epsilon^3)(1+\epsilon-2\epsilon^2+2\epsilon^3)},\nonumber
\end{align}
where it can be easily verified that $\etaCF\geq \etaU$.
\end{prop}

\begin{proof}
See Appendix \ref{App:AppendixthroughputCF-ARQ}.
\end{proof}

\begin{prop}\label{avgdelayCF-ARQ}
The average delay of CF ARQ for memoryless channels is given by
\begin{align}
\label{M2CFARQdelay}
\DCF=k + 1+ (2k + 8)\epsilon- (3k + 11)\epsilon^2+ (6T + 10k + 26)\epsilon^3+ \mathcal{O}(\epsilon^4),\quad \epsilon\to 0.
\end{align}
\end{prop}

\begin{proof}
See Appendix \ref{App:AppendixavgdelayCF-ARQ}.
\end{proof}

Comparing this with the average delay of uncoded ARQ, we observe that $\DCF-\DU=1+(k+7)\epsilon-(3k+T+12)\epsilon^2+\mathcal{O}(\epsilon^3)$ as $\epsilon\to 0$, which is due to the cumulative feedback. On the other hand, when $\epsilon$ is large, $\DCF$ becomes less than $\DU$, as we demonstrate in Sect. \ref{sims}.

Note that the MGFs for CF ARQ given in %\cite[Appendix G]{MalMedYeh2018tinycodesarxiv} and \cite[Appendix H]{MalMedYeh2018tinycodesarxiv} 
Appendix \ref{App:AppendixthroughputCF-ARQ} and Appendix \ref{App:AppendixavgdelayCF-ARQ}
with $M=N=1$ is equivalent to the MGFs of uncoded ARQ given in \cite{AusNos2007}.

%%%%%%%%%%%%%%%%%%%%%%%%%%%%%%%%%%%%%%%%%%%%%
\section{Coded ARQ}
\label{codedARQ}
In this section, we propose a Coded ARQ scheme, where the transmitted packets are coded. The coding scheme is similar to the generation-based random linear network coding in \cite{HoMedKoeKarEffShiLeo2006}. The sender has a coding bucket, and when it is ready to send a packet to the receiver, it produces a coded packet by forming a random linear combination of all the packets in the bucket. The encoded packet is then transmitted to the receiver. The receiver sends a cumulative feedback to indicate the set of successfully received encoded packets in the coding bucket. If the receiver successfully collects a sufficient number of encoded packets to decode all packets in the coding bucket, and the sender successfully receives the cumulative ACK message, it then purges the successfully ACK'ed encoded packets in the coding bucket and partially updates the coding bucket by moving new packets. 

We consider minimum coding, i.e., with a sliding window of size $M=2$. Different from uncoded ARQ, HARQ with soft combining, and CF ARQ, the transmission scheme is adaptive, i.e. the transmission rate is adjusted based on the cumulative feedback for $M=2$ MDS coded packets in the transmitted packet stream. The receiver needs both coded packets to reconstruct the transmitted packet stream, i.e., the DoFs required at the receiver is $N=2$. We do not assume in-order packet delivery. Therefore, the transmitted packets in the bucket will be successfully decoded when both of the coded transmitted packets are successfully received and ACK'ed by the receiver. While the model can be extended to $M>2$ using a recursion, the state space scales exponentially, and the analysis becomes prohibitively complicated without any additional insights. Therefore, it is left as future work.

The combined observation set for Coded ARQ with $M=2$ packets is $\mathcal{X}^{(c)}=\mathbb{Z}_2^3$. For example, $X_t^{(c)}=001$ means that the forward channel is good for both packets and the reverse channel is erroneous, i.e., the ACK for $M=2$ packets is lost. The HMM for the delay of Coded ARQ is shown in Fig. \ref{SR_ARQ_Coding_alphabet3}. Similar to previous models, $I_1$ and $O$ are the input and output nodes, and other nodes are the hidden states, and $I_1$ and $I_2$ represent transmission of the first new packet and the second packet one time slot later, respectively. The possibilities upon the transmission of the $2$ packets are:

\begin{figure*}[t!]
\centering
\includegraphics[width=\textwidth]{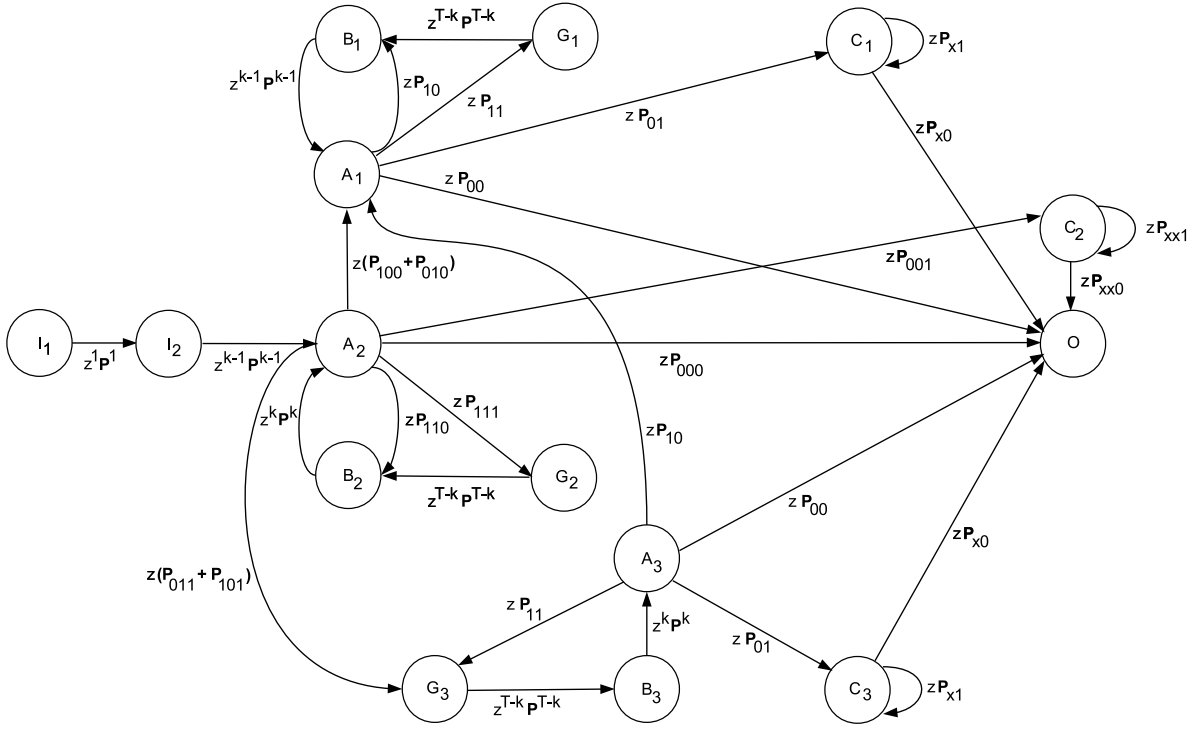}
\caption{\small{Matrix-flow graph for delay analysis of SR ARQ in unreliable feedback with coding.} 
\label{SR_ARQ_Coding_alphabet3}}
\end{figure*}

\begin{itemize}
\item {\bf Transition to state $A_2$.} After sending the new packets ($M=2$), the transmitter receives a feedback message $k-1$ time slots later. This state is represented by node $A_2$.
\item {\bf Transition to state $O$.} If the feedback is an error-free ACK (with probability $\matr{P}_{000}=\matr{P}_{0x}\matr{P}_{00}$), then the system transits to state $O$.
\item {\bf Transition to state $B_2$.} If the feedback is a successful NACK for both packets (with probability $\matr{P}_{110}=\matr{P}_{1x}\matr{P}_{10}$), then the system transits to state $B_2$, where both packets have to be retransmitted. 
\item {\bf Transition to state $C_2$.} If the feedback is an erroneous ACK (with probability $\matr{P}_{001}=\matr{P}_{0x}\matr{P}_{01}$) and the timer expires before receiving any error-free ACKs/NACKs, the system will transit to state $C_2$, the packets will be retransmitted, and the timeout will be reset. The packets will then be acknowledged when a succeeding ACK/NACK is correctly received. 
\item {\bf Transition to state $A_1$.} If only one of the packets is successfully transmitted and the feedback is an error-free ACK (with probability $\matr{P}_{100}+\matr{P}_{010}=\matr{P}_{1x}\matr{P}_{00}+\matr{P}_{0x}\matr{P}_{10}$), the system goes to state $A_1$. This state is equivalent to the state $A$ for the uncoded ARQ model as shown in Fig. \ref{NoCodingnoHARQthroughputanddelay}-(b). Hence, the rest of the analysis follows from the uncoded ARQ analysis in \cite{AusNos2007}. 
\item {\bf Transition to state $G_2$.} Node $G_2$ indicates that  a NACK is lost (with probability $\matr{P}_{111}=\matr{P}_{1x}\matr{P}_{11}$), both packets are lost, and the transmitter waits for timeout (node $B_2$). See also Fig. \ref{NoCodingnoHARQthroughputanddelay}.
\item {\bf Transition to state $G_3$.} Node $G_3$ indicates that a NACK is lost, but only one of the packets is successfully transmitted and the other one is lost, (with probability $\matr{P}_{011}+\matr{P}_{101}=\matr{P}_{0x}\matr{P}_{11}+\matr{P}_{1x}\matr{P}_{01}$), and the transmitter waits for timeout (node $B_3$). Node $A_3$ denotes the retransmission of both packets, and the receiver only needs one of the packets. Therefore, once the system goes to state $A_3$, the rest of the analysis follows from the uncoded ARQ analysis in \cite{AusNos2007}. 
\end{itemize} 

In this paper, since we use tiny codes, i.e. sliding window by coding with just $2$ packets, the available redundancy rate in terms of the packets in the encoding window is 50\%. However, we do a finer-grained control over the redundancy rate via the feedback which is cumulative. This can be observed from Fig. \ref{SR_ARQ_Coding_alphabet3}. For example, if the CF acknowledges the successful reception of $1$ packet only, i.e., the system transits to state $A_1$, then the rate is adaptively adjusted to retransmit $1$ packet only. On the other hand, if only $1$ packet is successfully transmitted and the CF is lost, then the system has a transition to $G_3$, and then to $A_3$ that represents the retransmission of both packets while the receiver only needs one of the packets. In this case, upon the successful reception of the CF in the succeeding time slots, the system either transits to state $A_1$, i.e. $1$ packet has to be retransmitted again, or to state $0$, i.e. no retransmission is required. Therefore, the redundancy rate of the model is not always 50\%, and a finer-grained control is provided through the feedback.

The matrix gain of the graph in Fig. \ref{SR_ARQ_Coding_alphabet3} can be calculated using the basic simplification rules. For the derivation of the MGFs of the transmission and delay times, hence characterization of the throughput and delay performance, of Coded ARQ, see %\cite[Appendices I, J]{MalMedYeh2018tinycodesarxiv}, 
Appendices \ref{App:AppendixthroughputcodedARQ} and \ref{App:AppendixavgdelaycodedARQ}, 
along with the relations (\ref{generatingfunctiontransmissiontime}) and (\ref{generatingfunctiondelay}). We now present closed form expressions for throughput and delay of memoryless channels.

\begin{prop}\label{throughputcodedARQ}
The throughput for Coded ARQ for memoryless channels is given by
\begin{align}
\label{M2codedARQthroughput}
\etaC=\frac{(1-\epsilon)}
{\alpha_C(\epsilon) + \epsilon^{d+1}(1-\epsilon)\beta_C(\epsilon)/(1-\epsilon^T)},
\end{align}
where $\alpha_C(\epsilon)=(1+\epsilon+7\epsilon^2/2-\epsilon^3/2-3\epsilon^4+\epsilon^6)/(1+\epsilon)^2$, $\beta_C(\epsilon)=1/2+\epsilon^2(1- \epsilon)$, and $d=T-k$. 
\end{prop}

\begin{proof}
See Appendix \ref{App:AppendixthroughputcodedARQ}.
\end{proof}

Note that in (\ref{M2codedARQthroughput}), it is easy to show that $3/4\leq \alpha_C(\epsilon)\leq 1$, and $1/2\leq \beta_C(\epsilon)< 13/20$. Therefore, we can conclude that $\etaC$ is always higher than $\etaU$ in (\ref{M1uncodedARQthroughput}) for any given $T$, $d$. Furthermore, $\etaC$ is upper bounded by $(1-\epsilon)/\alpha_C(\epsilon)$ as $T, d\to\infty$. This implies that for memoryless systems, with minimum coding, it is possible to achieve a gain of more than $30\%$ compared with uncoded ARQ. The gain becomes higher if the channel has memory, as demonstrated in Sect. \ref{sims}.

\begin{prop}\label{avgdelaycodedARQ}
The average delay of the Coded ARQ for memoryless channels is given by
\begin{align}
\label{codedARQavgdelay}
\DC=k + 1+3\epsilon+ (2T + 5k + 4)\epsilon^2+ (T - 7k + 5)\epsilon^3+\mathcal{O}(\epsilon^4),\quad \epsilon\to 0.
\end{align}
\end{prop}

\begin{proof}
See Appendix \ref{App:AppendixavgdelaycodedARQ}.
\end{proof}

The variance of delay for Coded ARQ is derived next.

\begin{prop}\label{secondmomentdelayCodedARQ}
The variance of delay for Coded ARQ for memoryless channels is
\begin{align}
\label{variabilityCodedARQ}
\varDC=k^2\epsilon^2(1+\epsilon-16\epsilon^2) + k\epsilon^2 (5-31\epsilon+43\epsilon^2-T\epsilon^2(4-6\epsilon+2\epsilon^2))+\mathcal{O}(1),\quad \epsilon\to 0.
\end{align}
\end{prop}

\begin{proof}
See Appendix \ref{App:AppendixsecondmomentdelayCodedARQ}.
\end{proof}

Comparing (\ref{codedARQavgdelay}) with the average delay of uncoded ARQ, we observe that $\DC-\DU=1+ (3-k-1)\epsilon+ (T + 5k + 3)\epsilon^2+\mathcal{O}(\epsilon^3)$ as $\epsilon\to 0$, which is due to the cumulative feedback. On the other hand, when $\epsilon$ is large, $\DC$ becomes comparable to $\DU$, as we demonstrate in Sect. \ref{sims}. However, Coded ARQ always provides better delay guarantees than uncoded ARQ. This provides insights in designing systems that are robust to the RTT fluctuations.

We next numerically evaluate the performance of the different ARQ protocols and outline the advantages of cumulative feedback and Coded ARQ over uncoded ARQ protocols.

%%%%%%%%%%%%%%%%%%%%%%%%%%%%%%%%%%%%%%%%%%%%%
\section{Numerical Results}
\label{sims}
We evaluate the performance of the SR ARQ schemes outlined in Sects. \ref{nocodingnosoftcombining}-\ref{codedARQ} by computing the MGFs of transmission and delay times via the MSFG approach detailed in Sect. \ref{matrixflowgraphs}, and provide a comparison of ARQ, HARQ, CF ARQ, and Coded ARQ schemes with feedback erasures. We also include the simulation results\footnote{The source code for simulation and analysis is available at \url{github.com/deryam/TinyCodesforDelayGuarantees}.} to validate our analytical models. The parameters for the numerical results are selected as follows. The RTT\footnote{The slot duration should be adjusted according to the transmission protocol. For example, if the transmission rate is 10 Mbits/sec, it takes 1 ms to transmit 10,000 bits over the channel. In that case, the RTT of $k=5$ time slots will be equivalent to 1 ms.} is $k=\{5,10\}$ time slots, timeout is $T=\{8,15\}$ slots when $k=5$, and $T=\{16,30\}$ slots when $k=10$, and we have the same parameters $\epsilon_B=1$, $\epsilon_G=0$, and $r$ for both the forward and reverse links, hence the same erasure rate, such that the proportion of the time spent in $G$ and $B$ can be computed using the stationary probabilities, given the erasure rate $\epsilon$. The performance metrics are the throughput $\eta$, the average delay $\bar{D}$, i.e. the per packet delay for uncoded ARQ and HARQ, and the delay corresponding to the transmission of $M=2$ packets in CF ARQ and Coded ARQ, and the guaranteeable delay $\hat{D}$ versus $\epsilon$ for varying RTT $k$, timeout $T$ and $r$. Unless otherwise specified, solid (Coded ARQ), dash-dot (CF ARQ), dashed (HARQ), dotted (ARQ) curves denote the analytical results, and unfilled circles denote the simulation results of this paper.

We next investigate the reliability of the protocols via numerically investigating the tail distribution of the delay. We illustrate the delay tail behavior in terms of the complementary cumulative distribution function (CCDF) in Fig. \ref{delay_CCDF_r03_k5} on a logarithmic scale, for the ARQ and HARQ protocols, and the CF ARQ and Coded ARQ protocols with $M=2$ packets. We demonstrate what these distributions look like both for memoryless and Gilbert-Elliott channels for $\epsilon=0.5$, $T=15$, $k=5$. From these semilogarithmic plots, it is clear that the tail decays exponentially. Thus, even though we did not prove analytically, the distribution of the delay is shown to be sub-Gaussian since $\mathbb{P}(D>d) \leq  e^{-vd^2}$ for $v=3\times10^{-4}$ and every $d>0$ as shown in marked curves in Fig. \ref{delay_CCDF_r03_k5}. Hence, the guaranteeable delay $\hat{D}$ of a protocol is upper bounded by the guaranteeable delay of a Gaussian distribution with the same first and second order parameters as $D$. Given that URLLC has different delay and reliability requirements, ranging from $10^{-5}$ to $10^{-9}$, we now discuss about guaranteed delays under different reliability requirements. Exploiting the sub-Gaussian behavior of the delay, a reliability requirement as high as $1-10^{-9}$ is guaranteed when we have that $\mathbb{P}(D\leq \hat{D}) \geq 1-e^{-v\hat{D}^2}$ for some $v>0$. Equivalently, the guaranteeable delay satisfies $\hat{D}\geq\sqrt{ \frac{9\log(10)}{v}}$. This result can be improved significantly when $\epsilon$ is smaller.
\begin{figure*}[t!]
\centering
\includegraphics[width=0.48\textwidth]{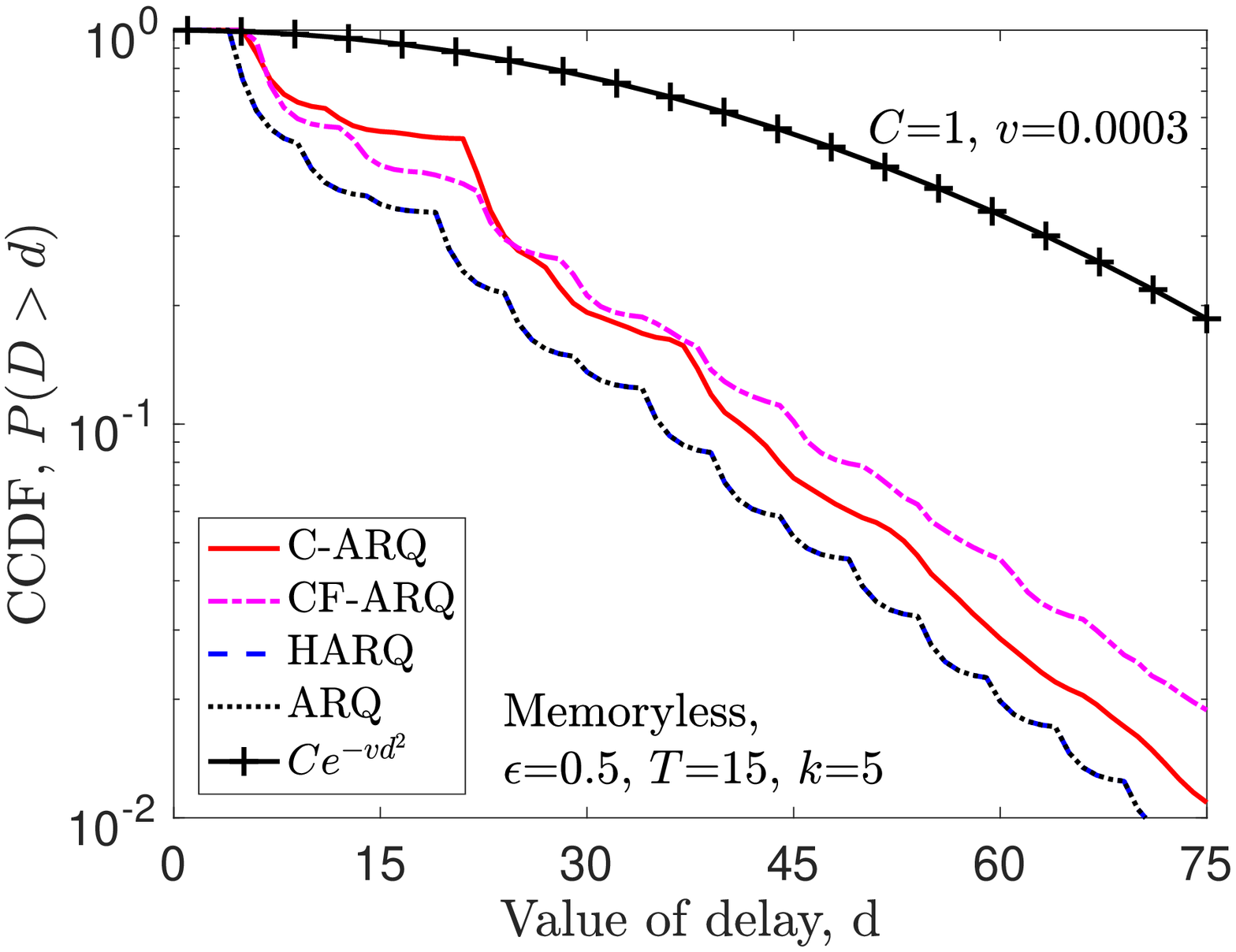}
\includegraphics[width=0.48\textwidth]{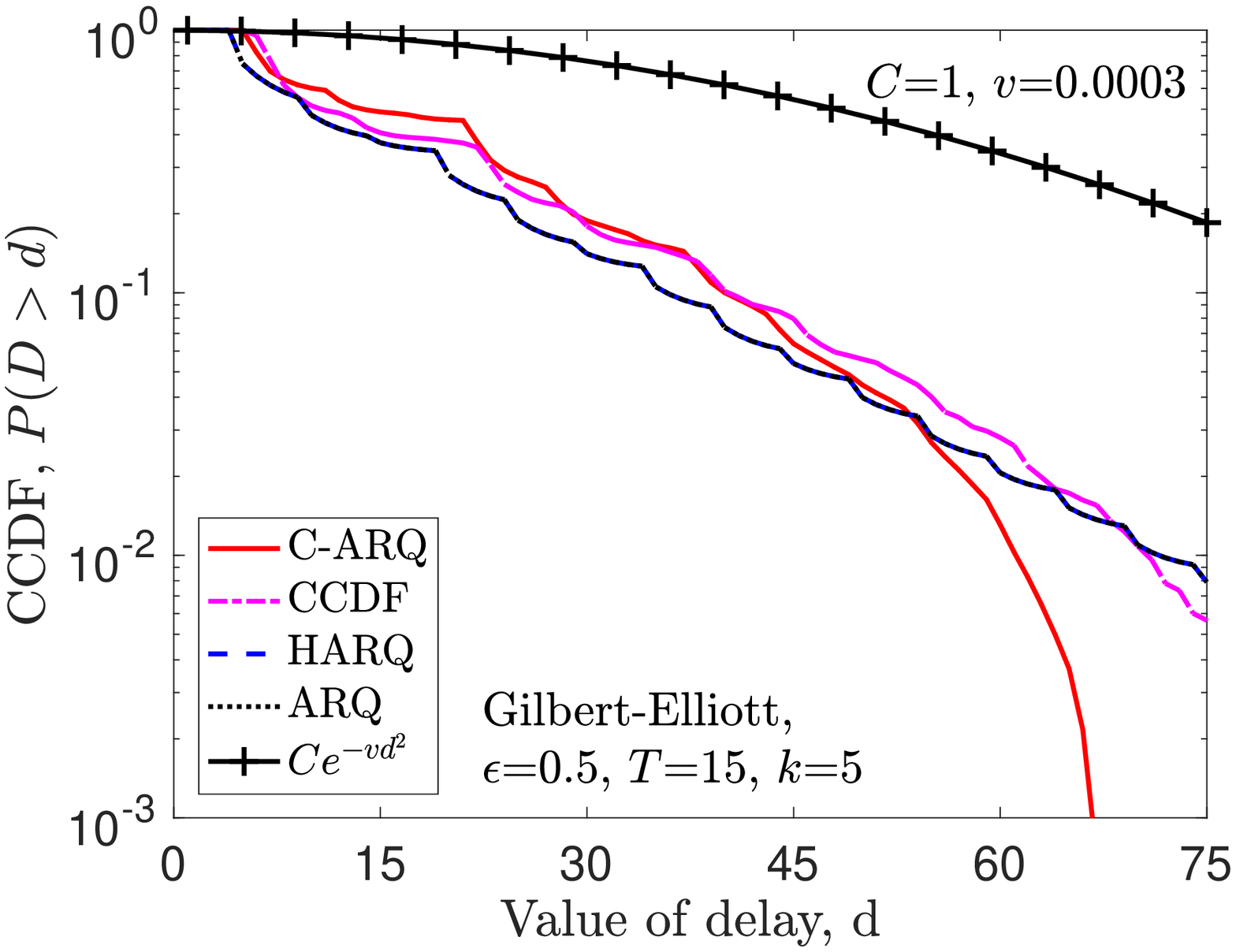}
\caption{\small{Delay CCDF for different ARQ schemes for $\epsilon=0.5$, $T=15$, $k=5$: (Left) Memoryless channel. (Right) Gilbert-Elliott channel, for $r=0.3$.}\label{delay_CCDF_r03_k5}}
\end{figure*}

\begin{figure*}[t!]
\centering
\includegraphics[width=0.48\textwidth]{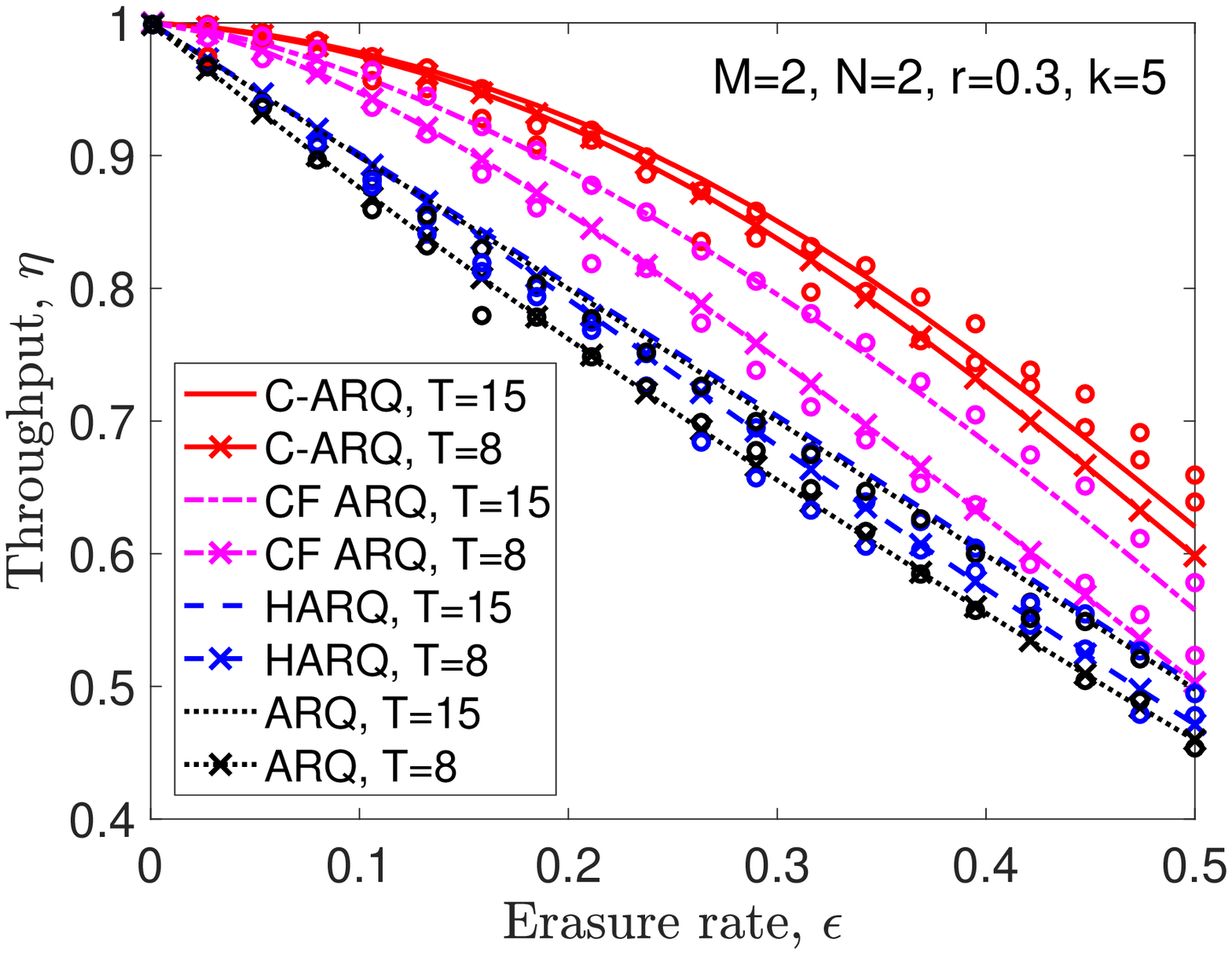}
\includegraphics[width=0.48\textwidth]{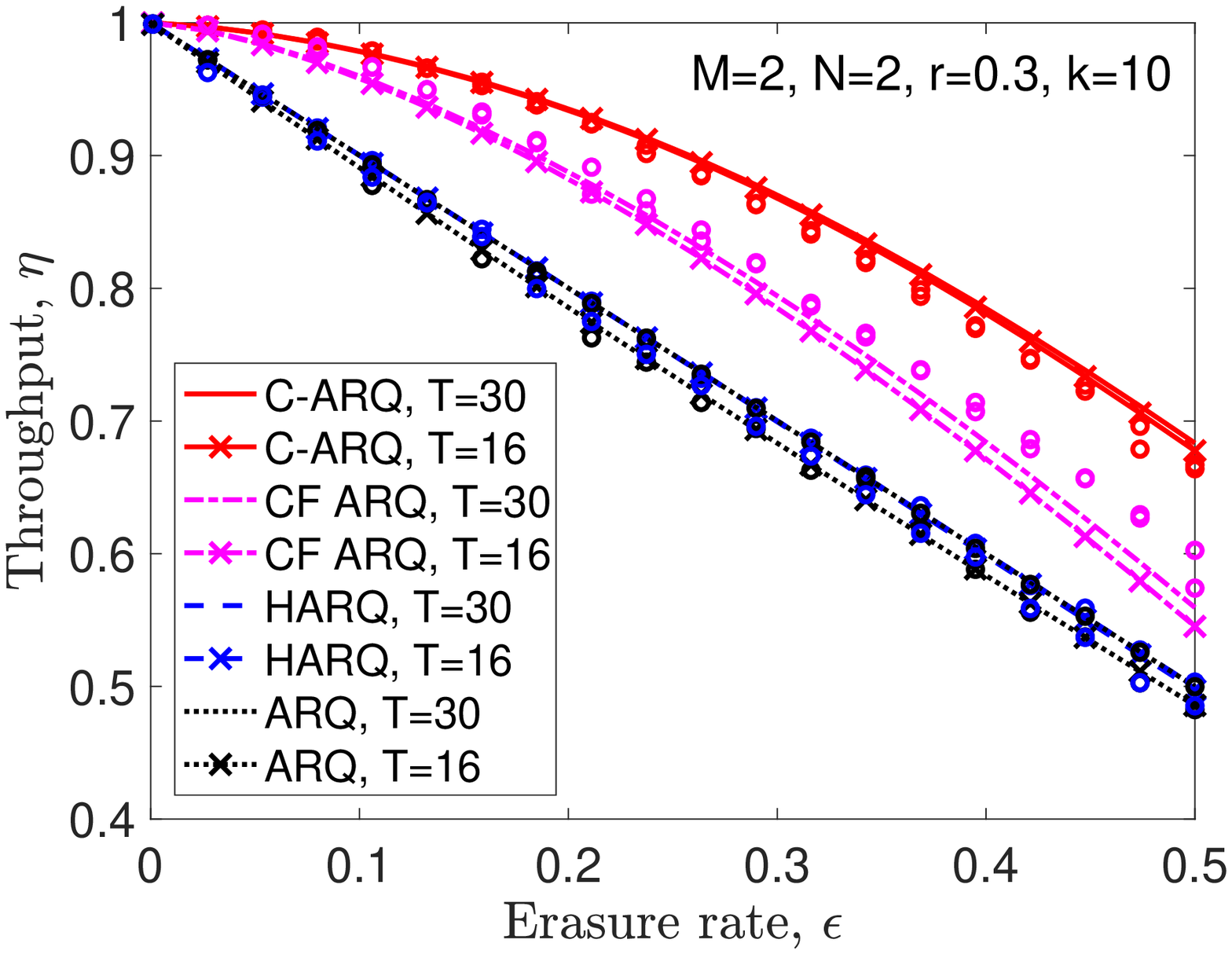}
\caption{\small{Throughput $\eta$ versus erasure rate $\epsilon$, in Markov errors for $r=0.3$, and $k=5$ and $k=10$.}\label{throughputvserasurer03}}
\end{figure*}

%%% DELAY AND THROUGHPUT VS ERASURE
The throughput and delay of the different ARQ protocols in the Markov channel for $r=0.3$ are shown in Figs. \ref{throughputvserasurer03} and \ref{delayvserasurer03_k5}-\ref{delayvserasurer03_k10}, respectively, for $k=\{5,10\}$, for different values of $T$. The baseline model is the uncoded ARQ scheme of \cite{AusNos2007}. For the HARQ scheme with soft combining at the receiver, $\epsilon_B(m)=1-e^{-\alpha/m}$ on a retransmission attempt $m$, where we assume $\alpha = 10\epsilon$, which is high, hence the erasure rate of state $B$. The HARQ scheme slightly improves the delay compared to the uncoded scheme, however its throughput is similar. In the Coded ARQ, more packets can be reliably transmitted even when the packet loss rate $\epsilon$ is large. As $\epsilon$ increases, throughput of Coded ARQ scheme decays more slowly than the other schemes because coding can compensate the packet losses. Hence, fewer retransmissions are required. Furthermore, delay is significantly lower than the uncoded ARQ schemes. In all ARQ models, when the timeout $T$ increases, both throughput and delay are higher. 

The simulation and analytical results agree for throughput and delay, as in \cite{AusNos2007}, except when the erasure or burst rates are high, or when RTT is comparable to timeout. Therefore, we did not include simulations for $r=0.1$. The throughput and delay of the protocols in the Markov channel for $r=0.1$ are shown in Figs. \ref{throughputvserasurer01} and \ref{delayvserasurer01_k5}-\ref{delayvserasurer01_k10}, respectively, for $k=\{5,10\}$ for different $T$. Comparing these results with the ones for $r=0.3$, the delay is higher and the throughput is lower for uncoded ARQ.

\begin{figure*}[t!]
\centering
\includegraphics[width=0.48\textwidth]{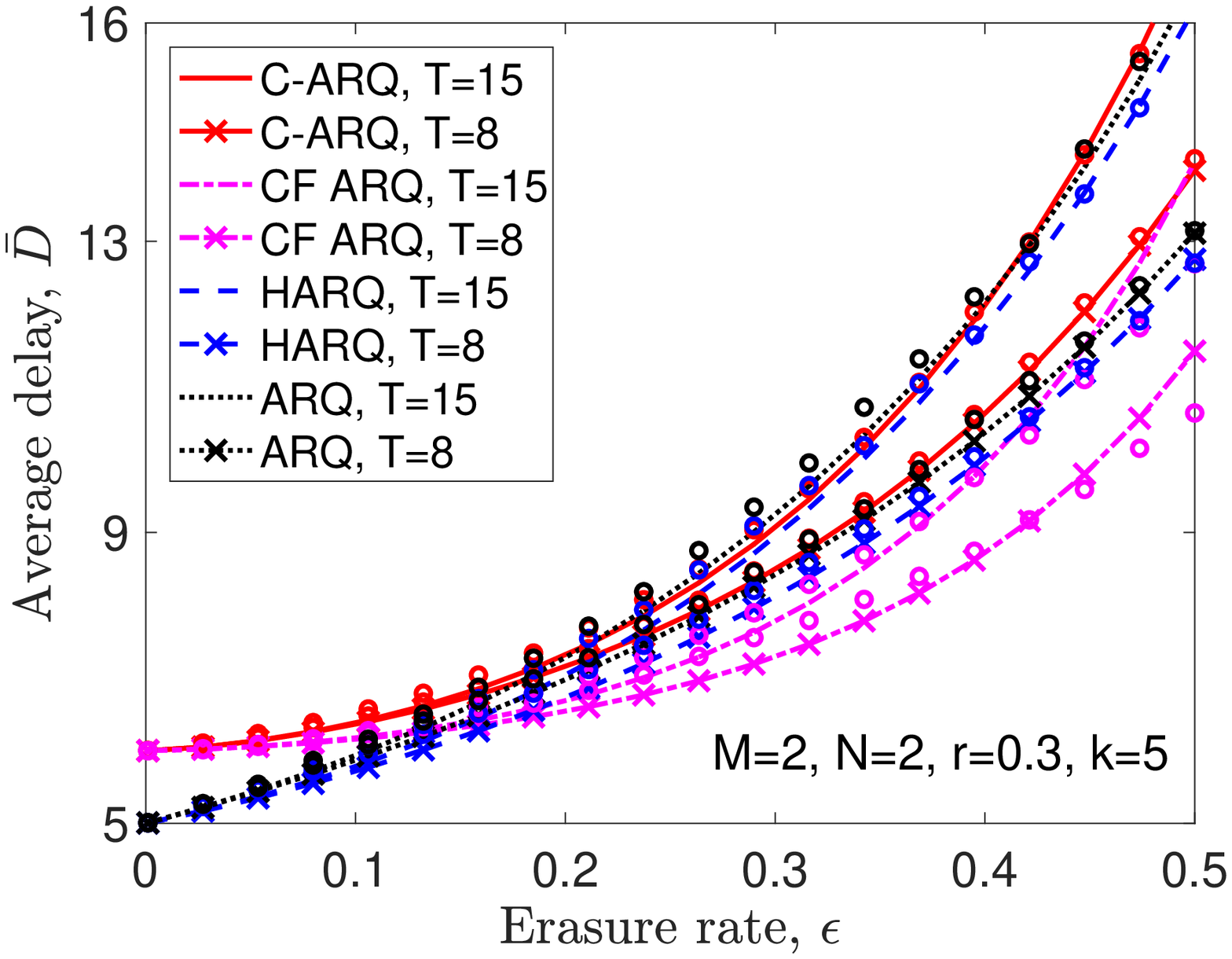}
\includegraphics[width=0.48\textwidth]{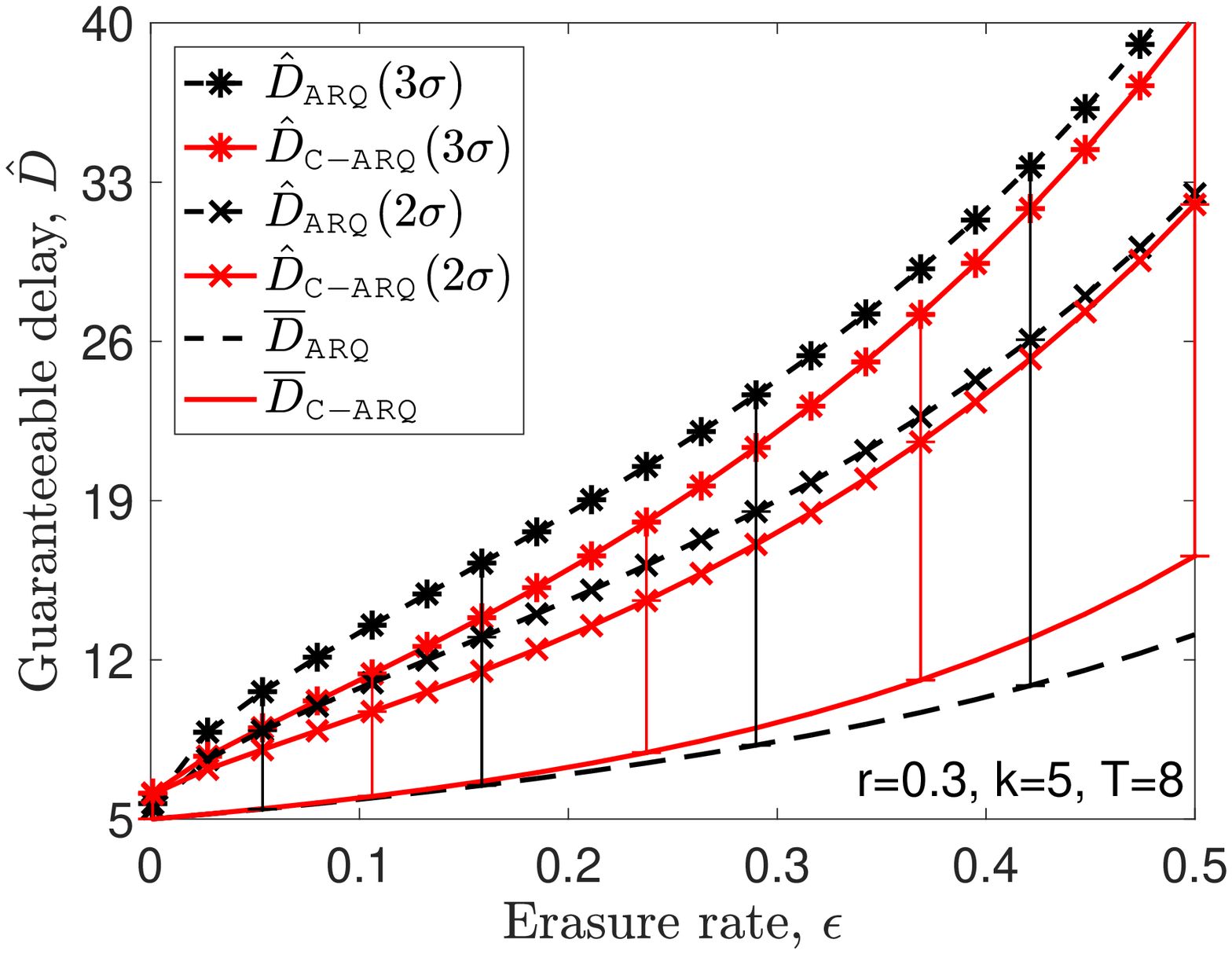}
\caption{\small{(Left) Average delay $\bar{D}$ versus erasure rate $\epsilon$ for various ARQ schemes. (Right) Guaranteeable delay $\hat{D}$ versus erasure rate $\epsilon$ for uncoded and Coded ARQ schemes for $T=8$, in Markov errors for $r=0.3$, and $k=5$.}\label{delayvserasurer03_k5}}
\end{figure*}

\begin{figure*}[t!] 
\centering
\includegraphics[width=0.48\textwidth]{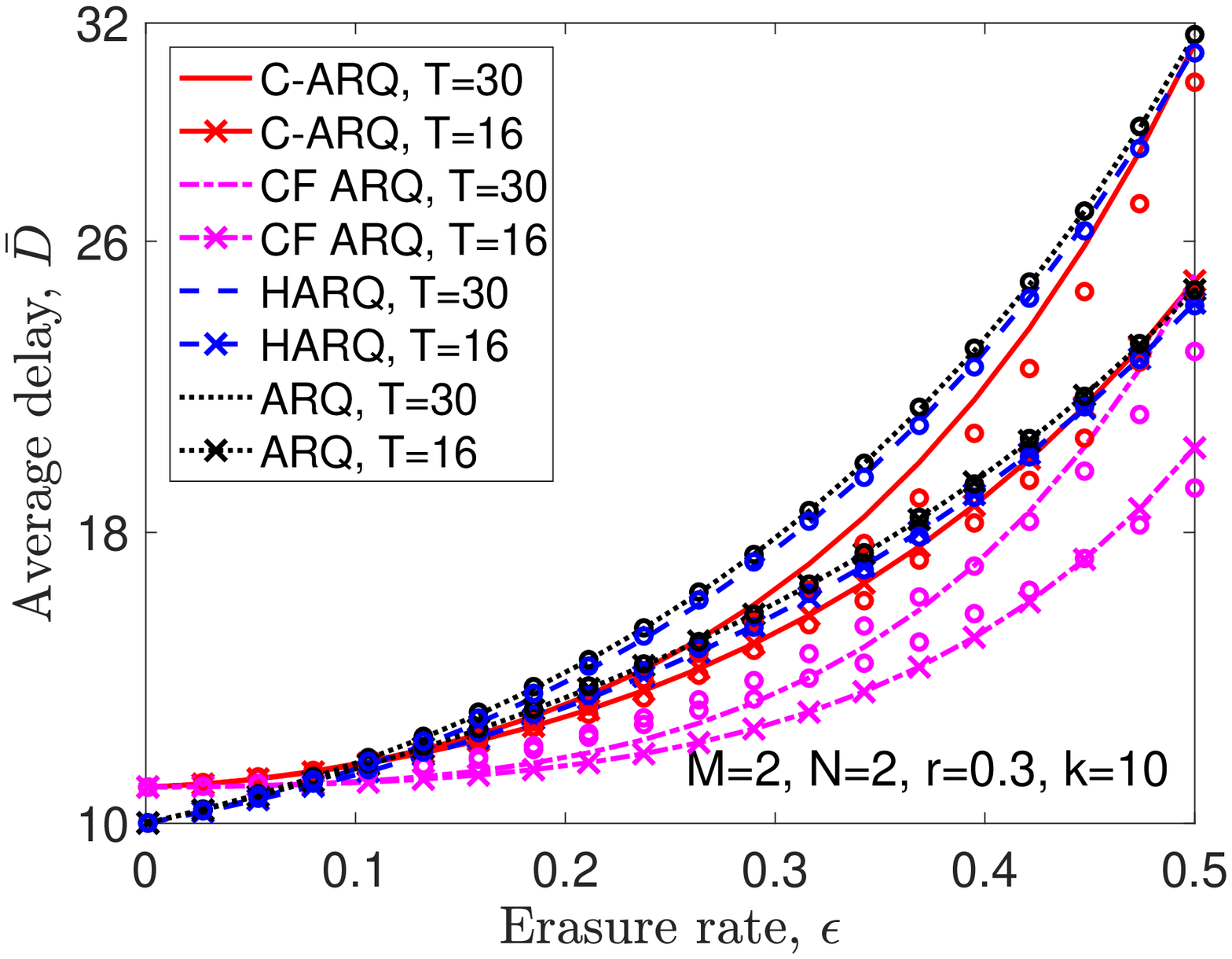}
\includegraphics[width=0.48\textwidth]{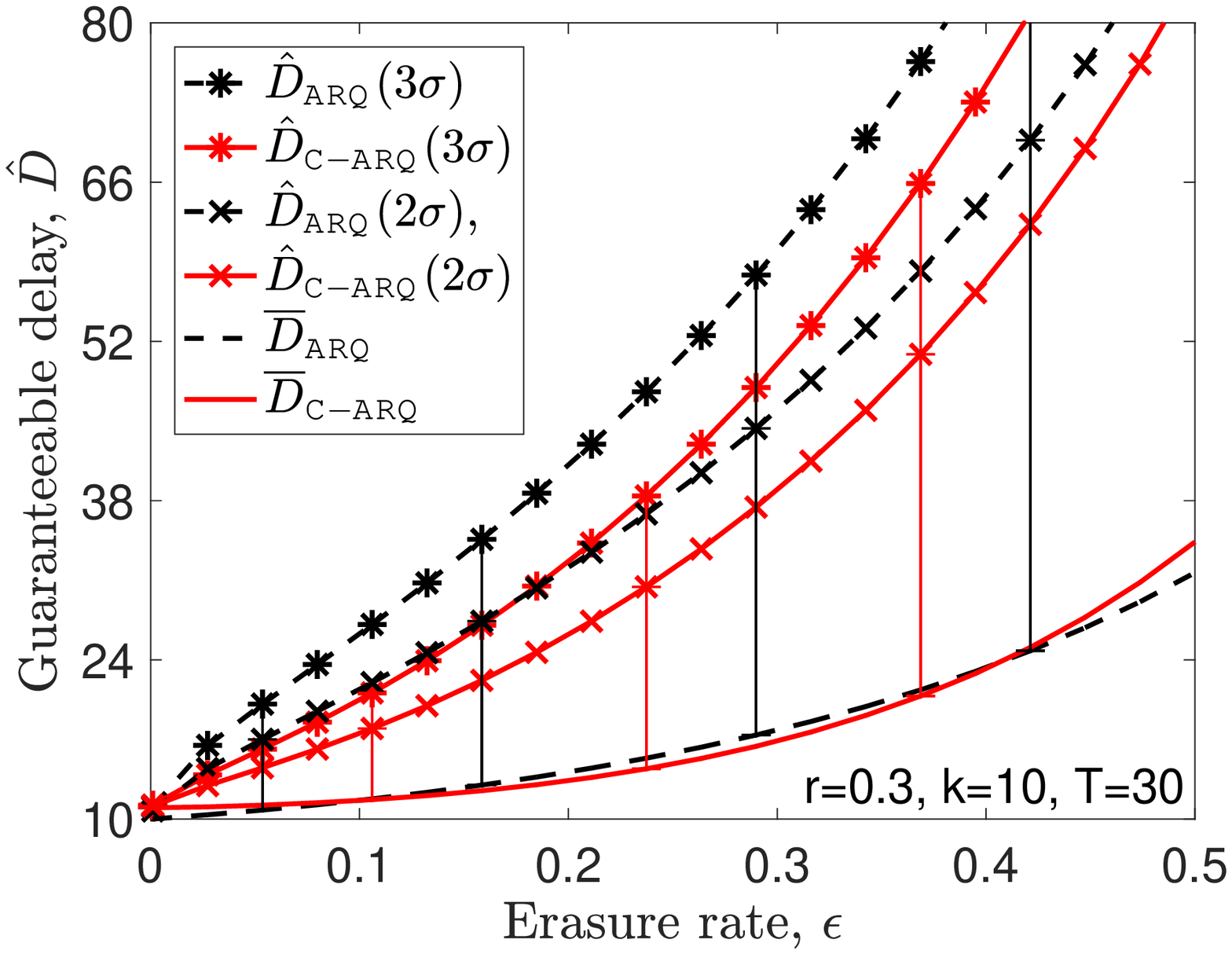}
\caption{\small{(Left) Average delay $\bar{D}$ versus erasure rate $\epsilon$ for various ARQ schemes. (Right) Guaranteeable delay $\hat{D}$ versus erasure rate $\epsilon$ for uncoded and Coded ARQ schemes for $T=30$, in Markov errors for $r=0.3$, and $k=10$.}\label{delayvserasurer03_k10}}
\end{figure*}

\begin{figure*}[t!]
\centering
\includegraphics[width=0.48\textwidth]{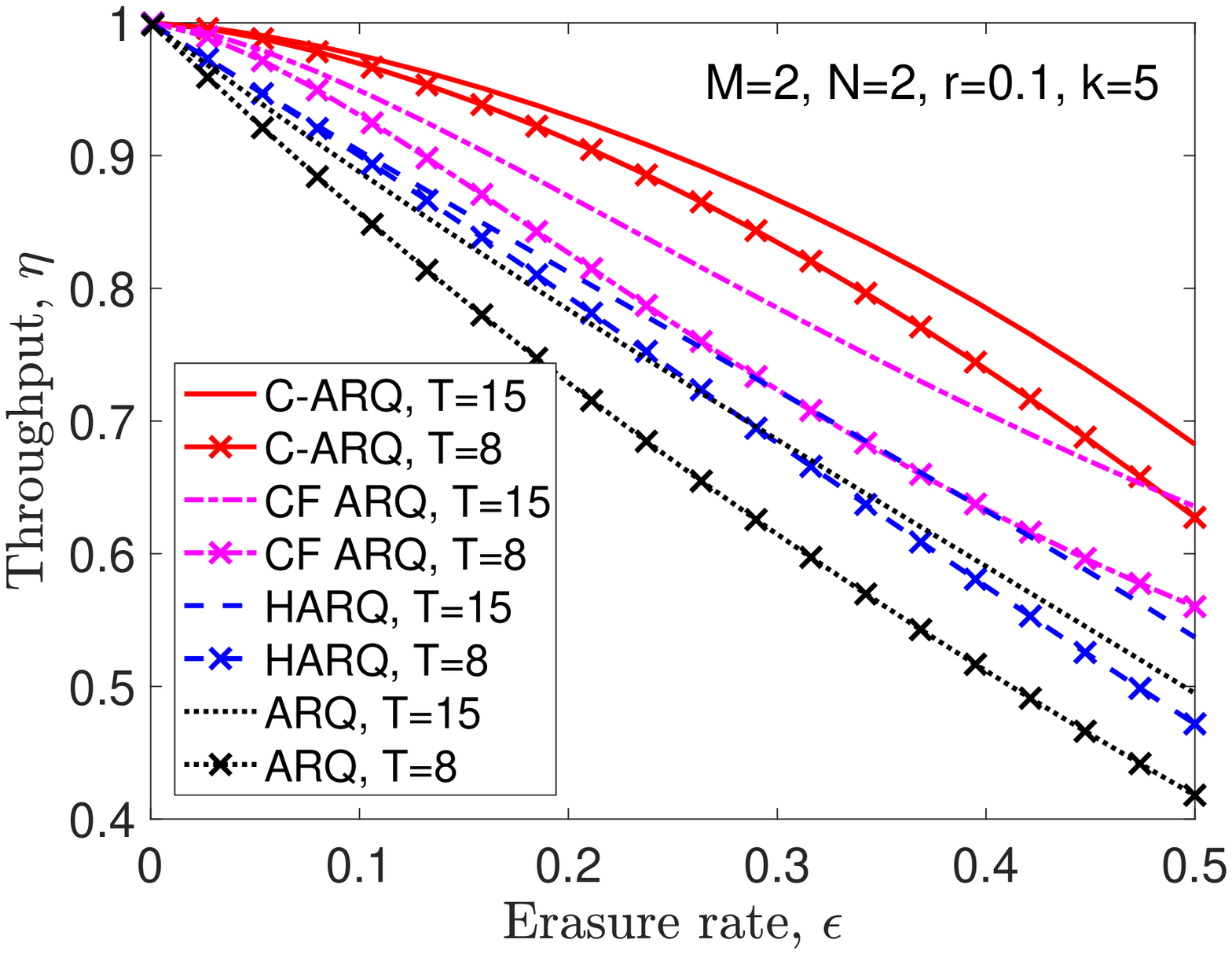}
\includegraphics[width=0.48\textwidth]{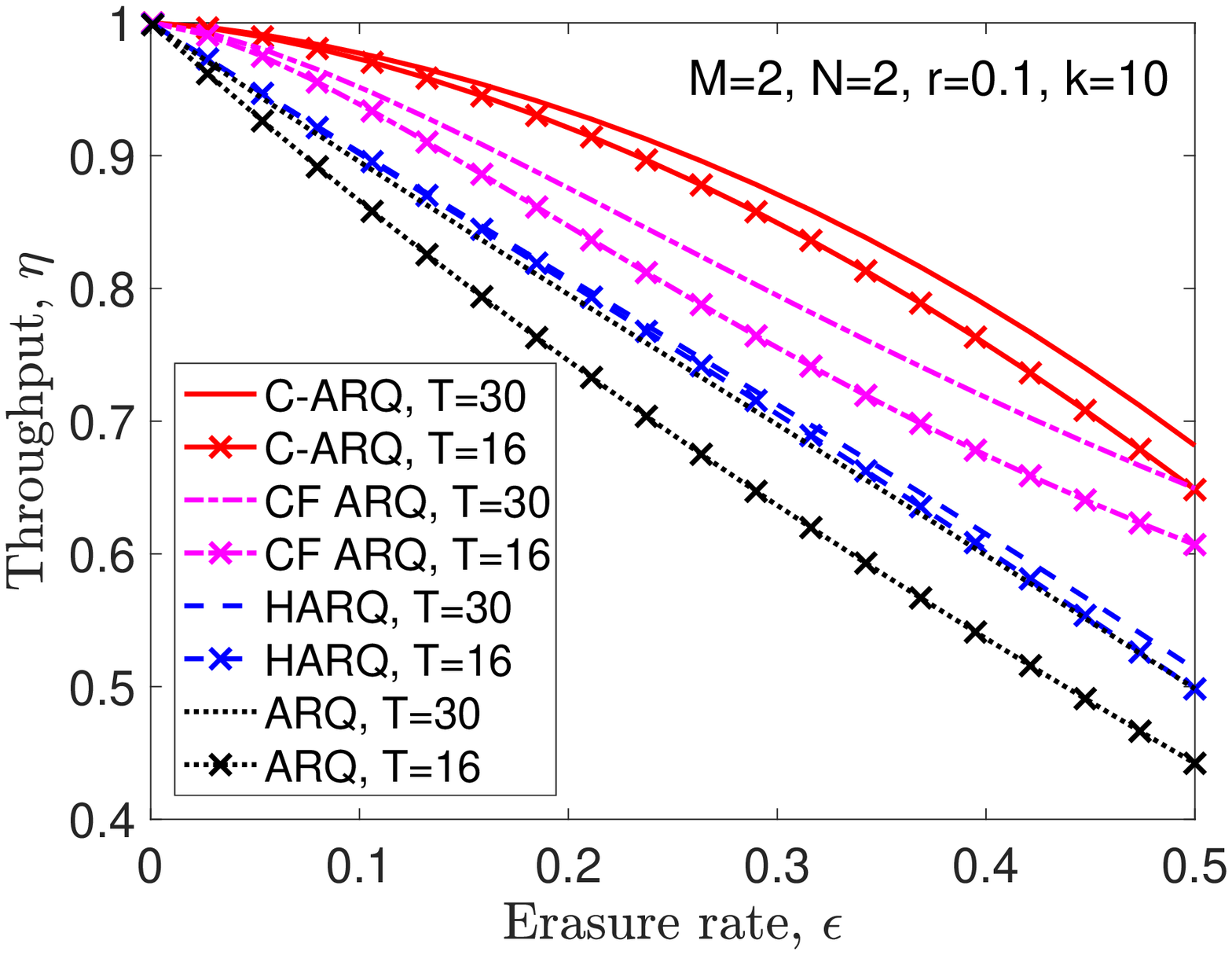}
\caption{\small{Throughput $\eta$ versus erasure rate $\epsilon$, in Markov errors for $r=0.1$, and $k=5$ and $k=10$.}\label{throughputvserasurer01}}
\end{figure*}

\begin{figure*}[t!]
\centering
\includegraphics[width=0.48\textwidth]{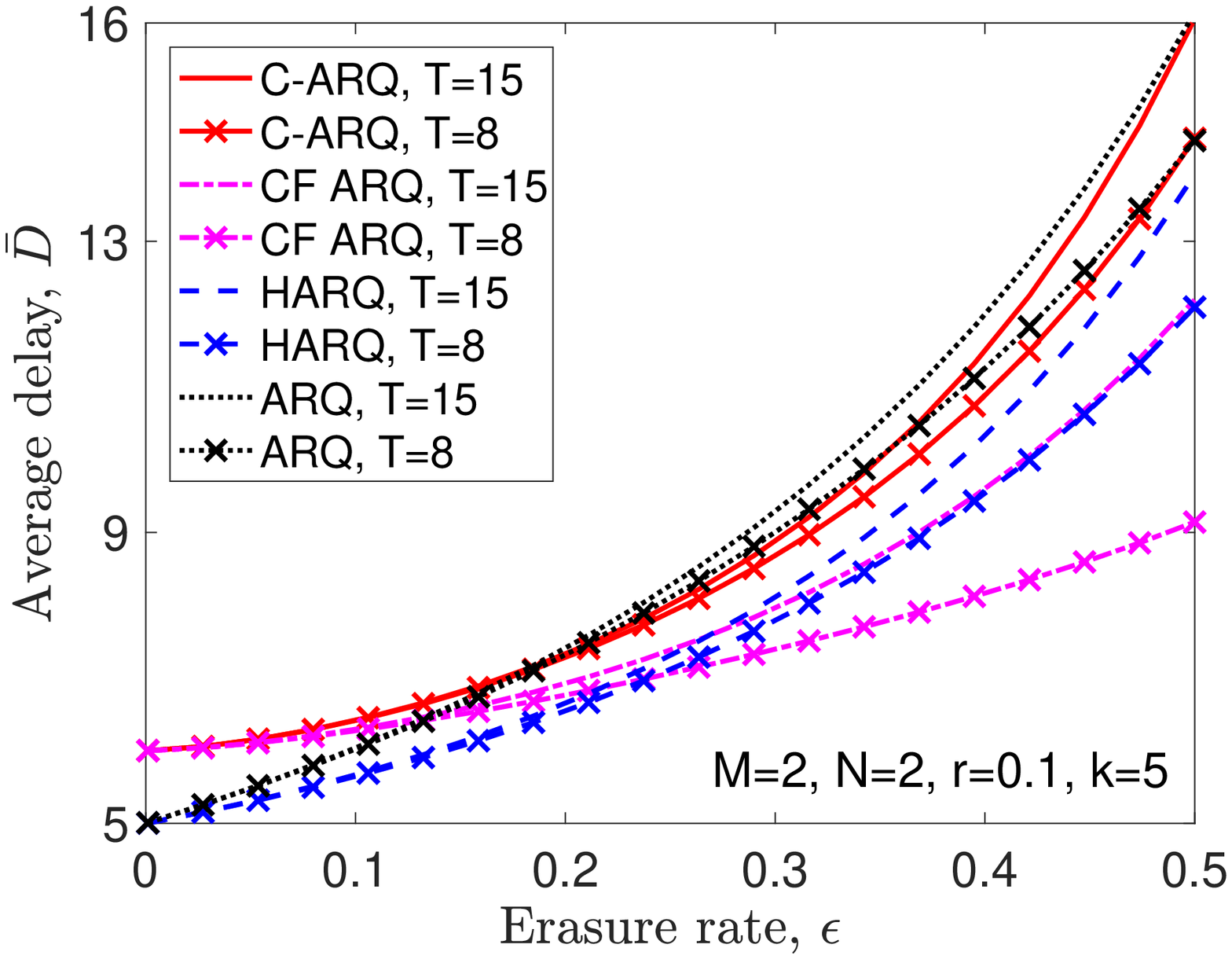}
\includegraphics[width=0.48\textwidth]{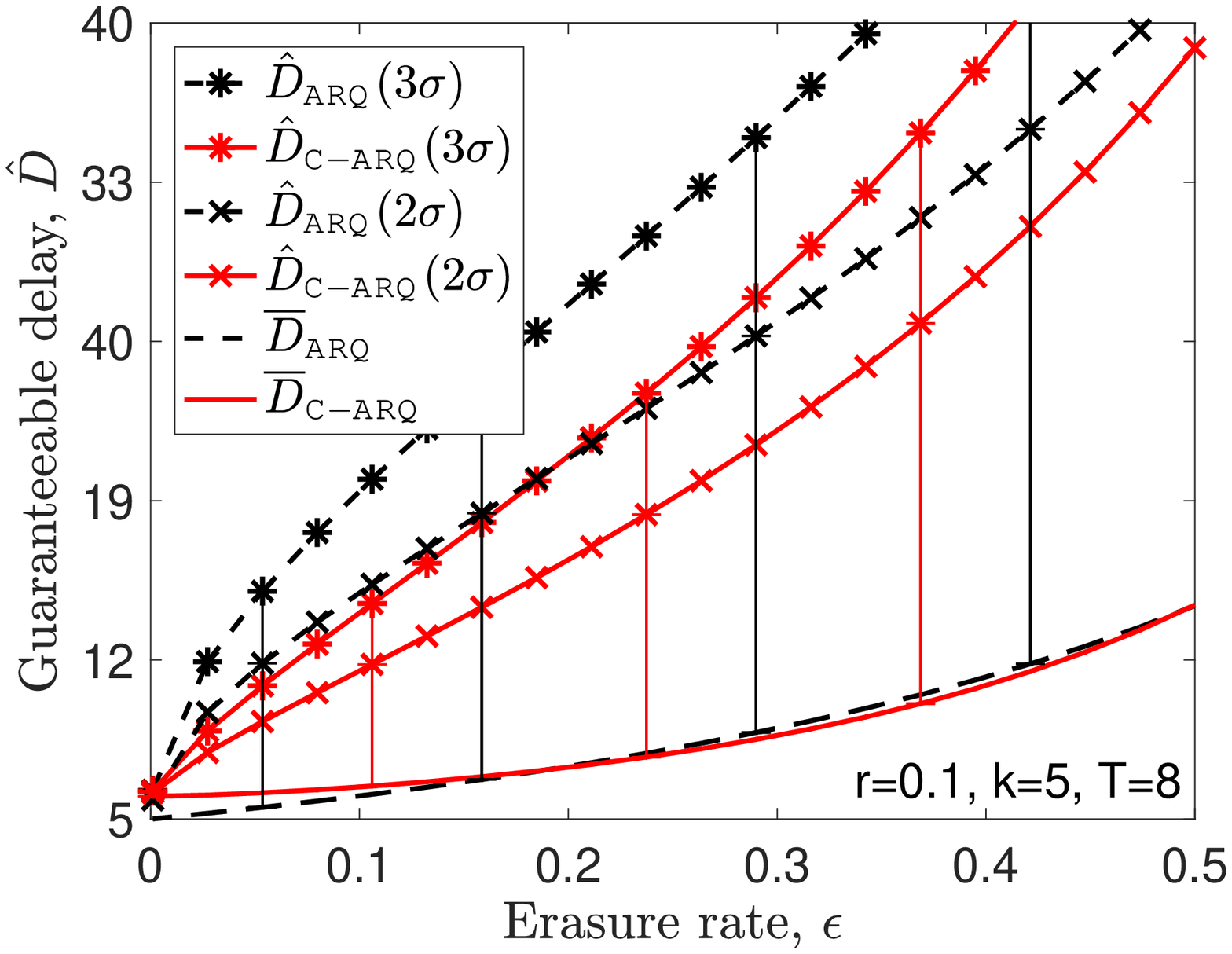}
\caption{\small{(Left) Average delay $\bar{D}$ versus erasure rate $\epsilon$ for various ARQ schemes. (Right) Guaranteeable delay $\hat{D}$ versus erasure rate $\epsilon$ for uncoded and Coded ARQ schemes for $T=8$, in Markov errors for $r=0.1$, and $k=5$.}\label{delayvserasurer01_k5}}
\end{figure*}

In CF ARQ and Coded ARQ with $M=2$ packets, since ${\rm RTT}=k+1$, the delay gap between the coded and uncoded schemes at $\epsilon=0$ is $1$ slot. Although the gap is indeed very small for small coding bucket sizes $M$, it also means that when the erasure rate is low, CF ARQ and Coded ARQ have higher average delays compared to the uncoded ARQ models. For CF ARQ, the throughput is higher than the throughput of the uncoded ARQ and HARQ with Chase combining since the feedback is cumulative, which decreases the number of packets being transmitted per a successful packet. Still, since CF ARQ is redundant in terms of transmissions, we do not have significant throughput gains compared to the uncoded ARQ. However, its average delay is lower than the average delays of the uncoded ARQ, HARQ and Coded ARQ under moderate or high erasures. Since this scheme has higher redundancy, it can achieve a better delay performance even under bursty errors. For Coded ARQ, the throughput is always higher than the throughput of the CF ARQ because the transmission rate is adapted based on the feedback received. However, Coded ARQ provides a higher average delay than CF ARQ. 

From \Cref{throughputvserasurer03,delayvserasurer03_k5,delayvserasurer03_k10,throughputvserasurer01,delayvserasurer01_k5,delayvserasurer01_k10}, we can observe that the higher the error burst ($r=0.1$), the lower the uncoded SR ARQ throughput in noisy feedback \cite{AusNos2007}. For uncoded ARQ and HARQ, the sensitivity of throughput to timeout $T$ increases as $r$ decreases, hence the throughput becomes very low when the timeout $T$ is very small. When $r=0.1$, for Coded ARQ, throughput becomes more sensitive to $T$, and it is possible to achieve significantly higher throughputs by increasing $T$.

As the burst rate increases, the average delay is higher for uncoded ARQ, HARQ and Coded ARQ. However, for both uncoded ARQ and HARQ models and Coded ARQ, the sensitivity of delay to timeout $T$ decreases as $r$ decreases, hence the variability of delay with timeout $T$ becomes less important under burst errors. For CF ARQ, however, when $r=0.1$, the sensitivity of the average delay to $T$ increases, and delay can be made significantly lower by keeping $T$ small.

\begin{figure*}[t!]
\centering
\includegraphics[width=0.48\textwidth]{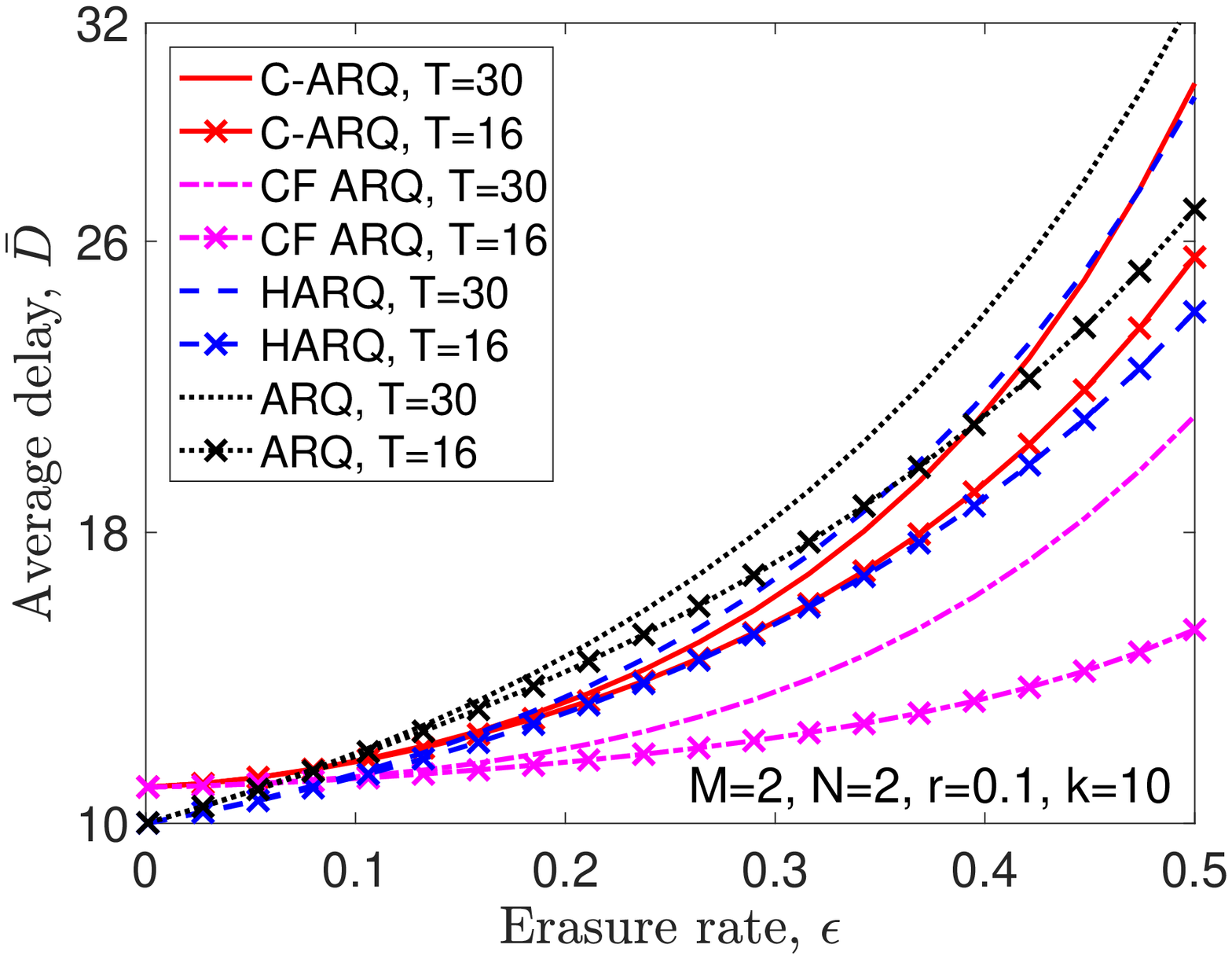}
\includegraphics[width=0.48\textwidth]{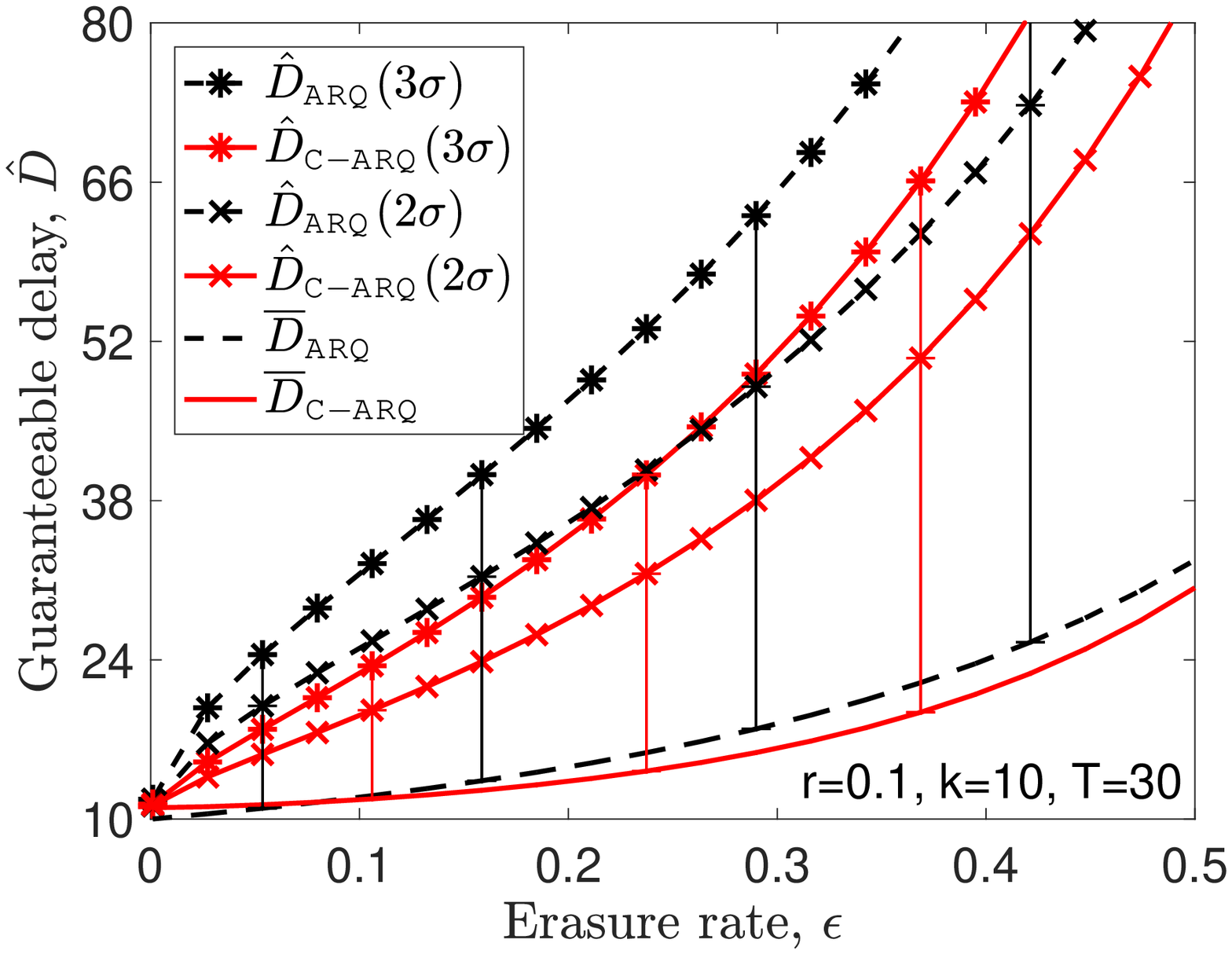}
\caption{\small{(Left) Average delay $\bar{D}$ versus erasure rate $\epsilon$ for various ARQ schemes. (Right) Guaranteeable delay $\hat{D}$ versus erasure rate $\epsilon$ for uncoded and Coded ARQ schemes for $T=30$, in Markov errors for $r=0.1$, and $k=10$.}\label{delayvserasurer01_k10}}
\end{figure*}

We next investigate the guaranteeable delay with respect to erasure rate, $\epsilon$, for different values of $r$, $k$, $T$ only for the uncoded ARQ and Coded ARQ. The guaranteeable delay $\hat{D}$ is illustrated in Figs. \ref{delayvserasurer03_k5} and \ref{delayvserasurer03_k10} for $r=0.3$, and in Figs. \ref{delayvserasurer01_k5} and \ref{delayvserasurer01_k10} for $r=0.1$, respectively, for different sets of $k$ and $T$. 

We observe that average and guaranteeable delay for both uncoded and Coded ARQ increases in $k$ and $T$. Although the Coded ARQ scheme might have a higher average delay $\bar{D}$ than the uncoded ARQ scheme, the guaranteeable delay $\hat{D}$ for Coded ARQ is lower than the one for uncoded ARQ. By increasing the timeout, the gap between the guaranteeable delays of both models can be made smaller.

Comparing Figs. \ref{delayvserasurer03_k5} and \ref{delayvserasurer03_k10} with Figs. \ref{delayvserasurer01_k5} and \ref{delayvserasurer01_k10}, we observe that under burst erasures ($r=0.1$), delay guarantees become more prevalent. Coded ARQ is also more reliable under high erasure rates. Therefore, we have the benefit of coding when we have burst or high rate of erasures. Coded ARQ is more guaranteeable across statistics, and hence is more stable.  

Our findings on various ARQ schemes suggest that the following design insights should enable more robust design for two-way erasure channels for wireless networks:
\begin{itemize}
\item Sensitivity of throughput and delay to timeout $T$ and RTT $k$ increases under burst errors.
\item Uncoded ARQ is very sensitive to error bursts. Therefore, the higher the burst rate, the lower its throughput and the higher its delay is.
\item CF ARQ provides significantly less average delay (when the erasure rate is high) than uncoded ARQ, HARQ and Coded ARQ, and a higher throughput than uncoded ARQ and HARQ, but its throughput performance is not as good as Coded ARQ.
\item Coded ARQ can provide throughput gains up to about 40\% more than the baseline uncoded ARQ.
\item Coded ARQ has higher average delay but lower variability. Guaranteeable delay for Coded ARQ is lower than the guaranteeable delay for uncoded ARQ. 
\item Coding has benefits under burst errors or higher erasure rates. Coded ARQ is more predictable across statistics, hence is more stable. This can help design robust systems when RTT is unreliable.
\end{itemize}

While the analysis is conducted only for tiny codes, this is the regime where substantial throughput gains can be achieved via coding. While the proposed technique becomes prohibitively complex to analyze for general coding window sizes due to the excessive number of hidden states, an exact analysis is possible for small window sizes using recursive formulas. Furthermore, using a network simulator, the technique can easily be extended to general window sizes. This can help understand the tradeoffs between the coding window size and system parameters along with the channel variations.

%%%%%%%%%%%%%%%%%%%%%%%%%%%%%%%%%%%%%%%%%%%%%
\section{Conclusions}
\label{conc}
We leveraged signal-flow techniques and coding theory to provide delay guarantees in SR ARQ. For tiny codes, we analyzed the distributions of transmission and delay times of HARQ with soft combining, cumulative-feedback based ARQ and Coded ARQ. Contrasting the performance of HARQ with soft combining and CF ARQ  with the uncoded ARQ scheme, we demonstrated their gains in terms of throughput and delay. For the given parameter setting, the CF ARQ scheme can provide a significant reduction in average delay, and a better throughput compared to the uncoded cases. Coded ARQ can provide gains up to about $40\%$ in throughput, and lower guaranteeable delay than the one for uncoded ARQ. This strategy also requires less feedback than that required by uncoded ARQ.

The insights can be applied to the design of mission-critical communications and industrial control for critical control messaging, which will be important use cases of 5G with ultra-reliability and ultra-low latency. Extensions include the optimization of the erasure coded schemes with minimal encoding and decoding complexity for asymmetric and bursty channels, and their code rate, and the development of more sophisticated coding schemes, such as sequential MDS or convolutional codes, Reed-Solomon codes for better FEC. Possible future directions also include the extension of the tiny coding scheme to long codes. Extending Coded ARQ to larger window sizes, we can investigate the scaling between the bucket size, the RTT and the DoFs required at the receiver, which will pave the way for protocol design with desirable throughput-delay tradeoffs.

%%%%%%%%%%%%%%%%%%%%%%%%%%%%%%%%%%%%%%%%%%%%

%
%%%%%%%%%%%%%%%%%%%%%%%%%%%%%%%%%%%%%%%%%%%%%
\begin{appendix}
\section{Appendices}
%%%%%%%%%%%%%%%%%%%%%%%%%%%%%%%%%%%%%%%%%%%%
\subsection{Transition Probability Matrices for the Gilbert-Elliott Channel and the Combined Channel}
\label{GEmatrices}
The state-transition matrix both for the forward and $\matr{P}$ for the reverse channels is denoted by $\matr{P}$. Since $\matr{P}$ is a stochastic matrix, $\matr{P}^n \matr{1}=\matr{1}$ for $n\geq 1$. The stationary vector of the state-transition matrix $\pi$ satisfies  $\pi \matr{1}=1$ and $\pi \matr{P}=\pi$. The combined state-transition matrix for the symmetric GE channels equals the Kronecker product of $\matr{P}$ with itself, i.e., $\matr{P}^{(c)}=\matr{P}\otimes \matr{P}$, which is given by
\begin{align}
\matr{P}^{(c)}=
\begin{bmatrix}
  (1-q)^2 &      q(1-q) &     q(1-q) &        q^2 \\
 r(1-q) & (1-q)(1-r) &             qr &  q(1-r) \\
 r(1-q) &            qr & (1-q)(1-r) & q(1-r) \\
        r^2 &      r(1-r) &      r(1-r) &  (1-r)^2 
\end{bmatrix}.
\end{align}

Let $\matr{P}_0$ and $\matr{P}_1$, respectively, be the success and the error probability matrices of an HMM. Both for the forward and reverse links, we have that $\matr{P}_{0}^{(f)}=\matr{P}_{0}^{(r)}=\matr{P}_0$ and $\matr{P}_{1}^{(f)}=\matr{P}_{1}^{(r)}=\matr{P}_1$, where
\begin{align}
\matr{P}_0&=\matr{P}\cdot {\rm diag}\{\matr{1}-\bm{\epsilon}\} = \begin{bmatrix}
    1-q & q  \\
    r & 1-r
\end{bmatrix}
\begin{bmatrix}
    1-\epsilon_G & 0  \\
    0 & 1-\epsilon_B
\end{bmatrix}=
\begin{bmatrix}
    (1-q)(1-\epsilon_G) & q(1-\epsilon_B)  \\
    r(1-\epsilon_G) & (1-r)(1-\epsilon_B)
\end{bmatrix},\nonumber\\
\matr{P}_1&=\matr{P}\cdot {\rm diag}\{\bm{\epsilon}\} = \begin{bmatrix}
    1-q & q  \\
    r & 1-r
\end{bmatrix}
\begin{bmatrix}
    \epsilon_G & 0  \\
    0 & \epsilon_B
\end{bmatrix}=
\begin{bmatrix}
    (1-q)\epsilon_G & q\epsilon_B  \\
    r\epsilon_G & (1-r)\epsilon_B
\end{bmatrix}.\nonumber
\end{align}
The probability vector of transmitting a new packet is $\pi_I=\pi \matr{P}_0$. Given the erasure rates $\bm{\epsilon}=[\epsilon_G, \epsilon_B]$, and $\epsilon=\pi\bm{\epsilon}^\intercal$, we have $\pi_I \matr{1}=\pi \matr{P}_0 \matr{1}=1-\pi \matr{P} \bm{\epsilon}^\intercal=1-\epsilon$, and $(\pi-\pi_I) \matr{1}=\pi \matr{P}_1 \matr{1}=\pi \matr{P} \bm{\epsilon}^\intercal=\epsilon$. The combined observation probabilities are given by the following $4\times 4$ matrices: 
\begin{align}
\matr{P}_{00}^{(c)}&=
\begin{bmatrix}
  (1-\epsilon_G)^2 \bar{q}^2 &    (1-\epsilon_B)(1-\epsilon_G) q\bar{q} &      (1-\epsilon_B)(1-\epsilon_G) q\bar{q} &        (1-\epsilon_B)^2 q^2 \\
 (1-\epsilon_G)^2 r\bar{q} & (1-\epsilon_B)(1-\epsilon_G)\bar{q} \bar{r} & (1-\epsilon_B)(1-\epsilon_G)qr &  (1-\epsilon_B)^2 q\bar{r} \\
 (1-\epsilon_G)^2 r \bar{q} & (1-\epsilon_B)(1-\epsilon_G)qr & (1-\epsilon_B)(1-\epsilon_G)\bar{q}\bar{r} &  (1-\epsilon_B)^2 q\bar{r} \\
        (1-\epsilon_G)^2 r^2   &      (1-\epsilon_B)(1-\epsilon_G) r\bar{r} &      (1-\epsilon_B)(1-\epsilon_G ) r\bar{r} &  (1-\epsilon_B)^2 \bar{r}^2
\end{bmatrix},\nonumber
\end{align}
\begin{align}
\matr{P}_{01}^{(c)}&=
\begin{bmatrix}
 \epsilon_G(1-\epsilon_G) \bar{q}^2 & \epsilon_B(1-\epsilon_G)q\bar{q} & \epsilon_G(1-\epsilon_B) q\bar{q} &       \epsilon_B(1-\epsilon_B)q^2 \\
  \epsilon_G(1-\epsilon_G) r\bar{q} & \epsilon_B(1-\epsilon_G)\bar{q}\bar{r} &  \epsilon_G(1-\epsilon_B)qr &  \epsilon_B(1-\epsilon_B) q\bar{r} \\
  \epsilon_G(1-\epsilon_G)r\bar{q} & \epsilon_B(1-\epsilon_G)qr & \epsilon_G(1-\epsilon_B)\bar{q} \bar{r} &  \epsilon_B(1-\epsilon_B)q\bar{r} \\
       \epsilon_G(1-\epsilon_G)r^2 & \epsilon_B(1-\epsilon_G)r\bar{r} & \epsilon_G(1-\epsilon_B)r\bar{r} & \epsilon_B(1-\epsilon_B) \bar{r}^2
\end{bmatrix},\nonumber
\end{align}
\begin{align}
\matr{P}_{10}^{(c)}&=
\begin{bmatrix}
 \epsilon_G(1-\epsilon_G)\bar{q}^2 &        \epsilon_G(1-\epsilon_B)q\bar{q} &        \epsilon_B(1-\epsilon_G)q\bar{q} &       \epsilon_B(1-\epsilon_B)q^2 \\
  \epsilon_G(1-\epsilon_G)r\bar{q} &  \epsilon_G(1-\epsilon_B)\bar{q}\bar{r} &             \epsilon_B(1-\epsilon_G)qr &  \epsilon_B(1-\epsilon_B)q\bar{r} \\
  \epsilon_G(1-\epsilon_G)r\bar{q} &             \epsilon_G(1-\epsilon_B)qr &  \epsilon_B(1-\epsilon_G)\bar{q}\bar{r} &  \epsilon_B(1-\epsilon_B)q\bar{r} \\
       \epsilon_G(1-\epsilon_G)r^2 &        \epsilon_G(1-\epsilon_B )r\bar{r} &        \epsilon_B(1-\epsilon_G)r\bar{r} & \epsilon_B(1-\epsilon_B)\bar{r}^2
\end{bmatrix},\nonumber
\end{align}
\begin{align}
\matr{P}_{11}^{(c)}&=
\begin{bmatrix}
  \epsilon_G^2\bar{q}^2 &      \epsilon_B\epsilon_Gq\bar{q} &      \epsilon_B\epsilon_Gq\bar{q} &        \epsilon_B^2q^2 \\
 \epsilon_G^2r\bar{q} & \epsilon_B\epsilon_G\bar{q}\bar{r} &             \epsilon_B\epsilon_Gqr &  \epsilon_B^2q\bar{r} \\
 \epsilon_G^2r\bar{q} &             \epsilon_B\epsilon_Gqr & \epsilon_B\epsilon_G\bar{q}\bar{r} &  \epsilon_B^2q\bar{r} \\
        \epsilon_G^2r^2 &      \epsilon_B\epsilon_Gr\bar{r} &      \epsilon_B\epsilon_Gr\bar{r} &     \epsilon_B^2\bar{r}^2 
\end{bmatrix},\nonumber
\end{align}
where we used the shorthand notation $\bar{q}=1-q$ and $\bar{r}=1-r$. 

\begin{comment}
Therefore, the composite channel for $X_t^{(c)}=11$ is given by
\begin{align}
\matr{P}_{11}^{(c)}&=\begin{bmatrix}
(1-q)\epsilon_G \matr{P}_{1}^{(r)} & q\epsilon_B \matr{P}_{1}^{(r)} \\
r\epsilon_G \matr{P}_{1}^{(r)} & (1-r)\epsilon_B \matr{P}_{1}^{(r)} \\
\end{bmatrix}.
\end{align}
\end{comment}

{\bf Transition Probability Matrices for the Symmetric Memoryless Channel.} 
Since the memoryless channel has only one state, $\matr{P}=1$ and its combined state-transition matrix is $\matr{P}^{(c)}=\matr{P}\otimes\matr{P}=1$. Hence, for memoryless channels with a symmetric erasure rate $\epsilon$, we have $\matr{P}_0=(1-\epsilon)$, and $\matr{P}_1=\epsilon$. Thus, the observation probabilities are $\matr{P}_{00}=(1-\epsilon)^2$, $\matr{P}_{01}=(1-\epsilon)\epsilon$, $\matr{P}_{10}=\epsilon(1-\epsilon)$, and $\matr{P}_{11}=\epsilon^2$.

%%%%%%%%%%%%%%%%%%%%%%%%%%%%%%%%%%%%%%%%%%%%
\subsection{Proof of Proposition \ref{throughputuncodedARQ}}
\label{App:AppendixavgtransmissiontimeuncodedARQ}
From (\ref{IA})-(\ref{ABt}), (\ref{AC}) and (\ref{CO}), the MGF of the transmission time for uncoded ARQ is given by \cite{AusNos2007} 
\begin{align}
\label{MGFtransmissiontimeNoCoding}
\PhitauU(z)&=z\matr{P}^{k-1}(\matr{I}-z\matr{P}_{10}\matr{P}^{k-1}-z\matr{P}_{11}\matr{P}^{T-1})^{-1}\nonumber\\
&\times\left[\matr{P}_{00}+\matr{P}_{01}\sum\limits_{j=0}^{d-1}{\matr{P}_{x1}^j} \matr{P}_{x0}+\matr{P}_{01}\matr{P}_{x1}^d (\matr{I}-z\matr{P}_{x1}^T)^{-1}z\sum\limits_{j=0}^{T-1}{\matr{P}_{x1}^j}\matr{P}_{x0}\right]. 
\end{align}

Computing the derivative of $\matr{\Phi}_{\tau}(z)$ at $z=1$, we have
\begin{align}
{\PhitauU}'(1)&=\matr{P}^{k-1}(\matr{I}-\matr{P}_{10}\matr{P}^{k-1}-\matr{P}_{11}\matr{P}^{T-1})^{-1}\nonumber\\
&\left\{\Big[\matr{I}+ (\matr{P}_{10}\matr{P}^{k-1}+\matr{P}_{11}\matr{P}^{T-1})   (\matr{I}-\matr{P}_{10}\matr{P}^{k-1}-\matr{P}_{11}\matr{P}^{T-1})^{-1}\Big]\right.\nonumber\\
&\times\left[\matr{P}_{00}+\matr{P}_{01}\sum\limits_{j=0}^{d-1}{\matr{P}_{x1}^j} \matr{P}_{x0}+\matr{P}_{01}\matr{P}_{x1}^d (\matr{I}-\matr{P}_{x1}^T)^{-1}\sum\limits_{j=0}^{T-1}{\matr{P}_{x1}^j}\matr{P}_{x0}\right]\nonumber\\
&\left.+\matr{P}_{01}\matr{P}_{x1}^d (\matr{I}-\matr{P}_{x1}^T)^{-1} \matr{P}_{x1}^T (\matr{I}-\matr{P}_{x1}^T)^{-1}\sum\limits_{j=0}^{T-1}{\matr{P}_{x1}^j}\matr{P}_{x0}+\matr{P}_{01}\matr{P}_{x1}^d (\matr{I}-\matr{P}_{x1}^T)^{-1}\sum\limits_{j=0}^{T-1}{\matr{P}_{x1}^j}\matr{P}_{x0}\right\},\nonumber
\end{align}
where we use the identity $\frac{d}{dz} (\matr{I}-\matr{A}z)^{-1} = (\matr{I}-\matr{A}z)^{-1} \matr{A} (\matr{I}-\matr{A}z)^{-1}$ for a square matrix $\matr{A}$.

Considering the case of memoryless channel, the mean transmission time is expressed as:
\begin{align}
\avgtauU&=(1-\epsilon(1-\epsilon)-\epsilon^2)^{-1}
\left\{\Big[1+ (\epsilon(1-\epsilon)+\epsilon^2)(1-\epsilon(1-\epsilon)-\epsilon^2)^{-1}\Big]\right.\nonumber\\
&\times\Big[(1-\epsilon)^2+(1-\epsilon)\epsilon\sum\limits_{j=0}^{d-1}{\epsilon^j} (1-\epsilon)+(1-\epsilon)\epsilon\epsilon^d (1-\epsilon^T)^{-1}\sum\limits_{j=0}^{T-1}{\epsilon^j}(1-\epsilon)\Big]\nonumber\\
&\left.+(1-\epsilon)\epsilon\epsilon^d (1-\epsilon^T)^{-1} \epsilon^T (1-\epsilon^T)^{-1}\sum\limits_{j=0}^{T-1}{\epsilon^j}(1-\epsilon)+(1-\epsilon)\epsilon\epsilon^d (1-\epsilon^T)^{-1}\sum\limits_{j=0}^{T-1}{\epsilon^j}(1-\epsilon)\right\}\nonumber\\
&=(1-\epsilon)^{-1}+\epsilon^{d+1}(1-\epsilon^T)^{-1}.\nonumber
\end{align}
Using this along with the relation $\etaU=1/\avgtauU$, the final expression for throughput can be obtained.

%%%%%%%%%%%%%%%%%%%%%%%%%%%%%%%%%%%%%%%%%%%%
\subsection{Proof of Proposition \ref{avgdelayuncodedARQ}}
\label{App:AppendixavgdelayuncodedARQ}
From (\ref{IA}), (\ref{AO}), (\ref{ABd}), (\ref{AG}), (\ref{AC}) and (\ref{CO}), the MGF of the delay for uncoded ARQ is given by  \cite{AusNos2007} 
\begin{align}
\label{MGFdelayNoCoding}
\matr{\PhiDU}(z)=z^{k-1}\matr{P}^{k-1}(\matr{I}-z^k \matr{P}_{10}\matr{P}^{k-1}-z^T \matr{P}_{11}\matr{P}^{T-1})^{-1}\times\left[z \matr{P}_{00}+z^2 \matr{P}_{01}(\matr{I}-z\matr{P}_{x1})^{-1}\matr{P}_{x0}\right]. 
\end{align}

Computing the derivative of $\PhiDU(z)$ at $z=1$, we have
\begin{align}
{\PhiDU}'(z)
&=\left((k-1)\matr{P}^{k-1}(\matr{I}-\matr{A}(1))^{-1}+\matr{P}^{k-1}(\matr{I}-\matr{A}(1))^{-1}\matr{A}'(1)(\matr{I}-\matr{A}(1))^{-1}\right)\nonumber\\
&\times\Big[\matr{P}_{00}+ \matr{P}_{01}(\matr{I}-\matr{P}_{x1})^{-1}\matr{P}_{x0}\Big]\nonumber\\
&+\matr{P}^{k-1}(\matr{I}-\matr{A}(1))^{-1}
\times\Big[\matr{P}_{00}+2 \matr{P}_{01} (\matr{I}-\matr{P}_{x1})^{-1} \matr{P}_{x0}+\matr{P}_{01} \matr{C}(1)\matr{P}_{x0}\Big],\nonumber
\end{align}
where $\matr{A}(z)=z^k \matr{P}_{10}\matr{P}^{k-1}+z^T \matr{P}_{11}\matr{P}^{T-1}$, we have $\matr{A}'(z)=kz^{k-1} \matr{P}_{10}\matr{P}^{k-1}+Tz^{T-1} \matr{P}_{11}\matr{P}^{T-1}$, and $\matr{C}(z)=(\matr{I}-z\matr{P}_{x1})^{-1} \matr{P}_{x1} (\matr{I}-z\matr{P}_{x1})^{-1}$. Hence, the mean delay equals $\DU=\frac{1}{1-\epsilon}\pi \matr{P}_0 \matr{\Phi}'_{D}(1) \matr{1}$.

Considering the case of memoryless channel, the mean delay is expressed as
\begin{align} 
\DU=\matr{\Phi}'_D(1)=k+\frac{\epsilon}{1-\epsilon}(1+T\epsilon)+k\epsilon.\nonumber
\end{align}

%%%%%%%%%%%%%%%%%%%%%%%%%%%%%%%%%%%%%%%%%%%%
\subsection{Proof of Proposition \ref{secondmomentdelayuncodedARQ}}
\label{App:AppendixsecondmomentdelayuncodedARQ}
The second derivative of $\PhiDU(z)$ at $z=1$ for uncoded ARQ can be computed as 
\begin{align}
{\PhiDU}''(1)&= \left((k-1)(k-2)\matr{P}^{k-1}(\matr{I}-\matr{A}(1))^{-1}+2(k-1)\matr{P}^{k-1}\matr{B}(1)+\matr{P}^{k-1}\matr{B}'(1)\right)\nonumber\\
&\times\Big[ \matr{P}_{00}+ \matr{P}_{01}(\matr{I}-\matr{P}_{x1})^{-1}\matr{P}_{x0}\Big]\nonumber\\
&+2((k-1)\matr{P}^{k-1}(\matr{I}-\matr{A}(1))^{-1}+\matr{P}^{k-1}\matr{B}(1))
\times\Big[\matr{P}_{00}+2 \matr{P}_{01}(\matr{I}-\matr{P}_{x1})^{-1}\matr{P}_{x0}+ \matr{P}_{01}\matr{C}(1)\matr{P}_{x0}\Big]\nonumber\\
&+\matr{P}^{k-1}(\matr{I}-\matr{A}(1))^{-1}
\times\Big[2 \matr{P}_{01} (\matr{I}-\matr{P}_{x1})^{-1} \matr{P}_{x0} +4 \matr{P}_{01} \matr{C}(1) \matr{P}_{x0}+ \matr{P}_{01} \matr{C}'(1) \matr{P}_{x0}\Big],\nonumber
\end{align}
where $\matr{A}(z)$ and $\matr{C}(z)$ are defined in Appendix \ref{App:AppendixavgdelayuncodedARQ}, and $\matr{B}(z)=(\matr{I}-\matr{A}(z))^{-1}\matr{A}'(z)(\matr{I}-\matr{A}(z))^{-1}$. Given these functions, we can compute their derivatives as follows:
\begin{align}
\matr{A}'(z)&=kz^{k-1} \matr{P}_{10}\matr{P}^{k-1}+Tz^{T-1} \matr{P}_{11}\matr{P}^{T-1},\nonumber\\
\matr{A}''(z)&=k(k-1)z^{k-2} \matr{P}_{10}\matr{P}^{k-1}+T(T-1)z^{T-2} \matr{P}_{11}\matr{P}^{T-1},\nonumber\\
\matr{B}'(z)&=2(\matr{I}-\matr{A}(z))^{-1}\matr{A}'(z)(\matr{I}-\matr{A}(z))^{-1}\matr{A}'(z)(\matr{I}-\matr{A}(z))^{-1}+(\matr{I}-\matr{A}(z))^{-1}\matr{A}''(z)(\matr{I}-\matr{A}(z))^{-1},\nonumber\\
\matr{C}'(z)&=2(\matr{I}-z\matr{P}_{x1})^{-1} \matr{P}_{x1} (\matr{I}-z\matr{P}_{x1})^{-1} \matr{P}_{x1} (\matr{I}-z\matr{P}_{x1})^{-1}.\nonumber
\end{align}

For the case of memoryless channel, simplifying above expressions, we have that $\matr{A}(z)=z^k \epsilon(1-\epsilon)+z^T \epsilon^2$, and $\matr{B}(z)=(1-\matr{A}(z))^{-2}\matr{A}'(z)$, and $\matr{C}(z)=\epsilon(1-z\epsilon)^{-2}$ and $\matr{C}'(z)=2\epsilon^2(1-z\epsilon)^{-3}$. Using the simplified functions, we obtain that $\Big[\matr{P}_{00}+ \matr{P}_{01}(\matr{I}-\matr{P}_{x1})^{-1}\matr{P}_{x0}\Big]=1-\epsilon$, $\Big[\matr{P}_{00}+2 \matr{P}_{01}(\matr{I}-\matr{P}_{x1})^{-1}\matr{P}_{x0}+ \matr{P}_{01}\matr{C}(1)\matr{P}_{x0}\Big]=
1$, and $\Big[2 \matr{P}_{01} (\matr{I}-\matr{P}_{x1})^{-1} \matr{P}_{x0} +4 \matr{P}_{01} \matr{C}(1) \matr{P}_{x0}+ \matr{P}_{01} \matr{C}'(1) \matr{P}_{x0}\Big]=2(1-\epsilon)^{-1}\epsilon$. Hence, using the MGF of $D$, the variance of the delay for uncoded ARQ is given as follows:
\begin{align}
\varDU&=\frac{1}{1-\epsilon}\pi \matr{P}_0 {\PhiDU}''(1) \matr{1}+\DU-{\DU}^2\nonumber\\
&=\Big\{\Big((k-1)(k-2)(1- \epsilon(1-\epsilon)- \epsilon^2)^{-1}+2(k-1)\matr{B}(1)+\matr{B}'(1)\Big)(1-\epsilon)\nonumber\\
&+2\Big((k-1)(1-\epsilon(1-\epsilon)-\epsilon^2)^{-1}+\matr{B}(1)\Big)+(1-\epsilon(1-\epsilon)-\epsilon^2)^{-1}
2(1-\epsilon)^{-1}\epsilon\Big\}+\bar{D}-\bar{D}^2\nonumber\\
&=-3k+2-k^2\epsilon^2+\frac{\epsilon^4}{(1-\epsilon)^3}2T^2+\frac{\epsilon}{(1-\epsilon)^2}(2+4kT\epsilon^2+T^2\epsilon+T\epsilon-\epsilon-2T\epsilon^2-T^2\epsilon^3)\nonumber\\
&+\frac{1}{1-\epsilon}( k^2\epsilon+2k-2- k\epsilon+\epsilon+2T\epsilon^2+2k^2\epsilon^2- 2kT\epsilon^2-2k\epsilon^2-2Tk\epsilon^3).\nonumber
\end{align}
Hence, the final result can be given in the form of (\ref{variabilityuncodedARQ}).

%%%%%%%%%%%%%%%%%%%%%%%%%%%%%%%%%%%%%%%%%%%%
\subsection{Proof of Proposition \ref{throughputHARQ}}
\label{App:AppendixthroughputHARQ}
The MGF of the transmission time in the case of no coding with HARQ combining is given by
\begin{align}
\label{MGFdelayNoCodingHARQ} 
\PhitauH(z)&=z\matr{P}^{k-1}
\sum\limits_{j=0}^{\infty}\prod_{i=0}^j{(z\matr{P}_{10}(i)\matr{P}^{k-1}+z\matr{P}_{11}(i)\matr{P}^{T-1})}
\Big[\matr{P}_{00}(j^*)\\
&+\matr{P}_{01}(j^*)\sum\limits_{j=0}^{d-1}{\Big(\prod\limits_{i=0}^j{\matr{P}_{x1}(j^*+i)}\Big)}\matr{P}_{x0}(j^*+j+1)+\matr{P}_{01}(j^*)\Big(\prod\limits_{i=0}^d{\matr{P}_{x1}(j^*+i)}\Big)\nonumber\\
&\times\sum\limits_{j=0}^{\infty}z^j\prod\limits_{i=1}^j\Big(\prod\limits_{l=(i-1)T+1}^{iT}{\matr{P}_{x1}(j^*+l)}\Big)
z\sum\limits_{j'=0}^{T-1}{\Big(\prod\limits_{i=0}^{j'}\matr{P}_{x1}(j^*+jT+i)\Big)}\matr{P}_{x0}(j^*+jT+j'+1)\Big],\nonumber
\end{align}
which is a generalization of uncoded ARQ in \cite{AusNos2007}. In the above, $j^*$ denotes attempt index where the forward link is successful, i.e. the initial $j^*-1$ attempts were not successful, $\matr{P}_{10}(i)$ and $\matr{P}_{11}(i)$ denote the probabilities of receiving an error-free NACK and an erroneous NACK on attempt $i\geq 0$, respectively, and $\matr{P}_{x0}(i)$ and $\matr{P}_{x1}(i)$ denote the success and failure probability matrices on attempt $i\in\{1,\hdots,d-1\}$, respectively. Furthermore, for notational convenience, we let $\matr{P}_{x1}(0)=\matr{I}$.

The value of $\avgtauH$ can be found using the relation (\ref{generatingfunctiontransmissiontime}), then computing the reciprocal of its first derivative. Simplifying above expression using the relations for the transition probability matrices for the symmetric memoryless channel, given in Appendix \ref{GEmatrices}, the value of throughput $\etaH={\avgtauH}^{-1}$ for the memoryless channel is given by (\ref{M1HARQthroughput}).

%%%%%%%%%%%%%%%%%%%%%%%%%%%%%%%%%%%%%%%%%%%%
\subsection{Proof of Proposition \ref{avgdelayHARQ}}
\label{App:AppendixavgdelayHARQ}
The MGF of the delay in the case of (uncoded) HARQ with soft combining is 
\begin{align}
\label{MGFdelayNoCodingHARQ}
\PhiDH(z)&=z^{k-1}\matr{P}^{k-1} \sum\limits_{j=0}^{\infty}{\prod\limits_{i=0}^j {(z^k \matr{P}_{10}(i)\matr{P}^{k-1}+z^T \matr{P}_{11}(i)\matr{P}^{T-1})}}\nonumber\\
&\times\left[z \matr{P}_{00}(j^*)+z^2 \matr{P}_{01}(j^*) \sum\limits_{j'=0}^{\infty}{z^{j'} \prod\limits_{i=0}^{j'}{\matr{P}_{x1}(j^*+i)}}  \matr{P}_{x0}(j^*+j'+1)\right].
\end{align}

Average delay $\DeH$ can be found using the relation (\ref{generatingfunctiondelay}), then computing its first derivative. Simplifying above expressions, average delay $\DeH$ for the memoryless channel is given by (\ref{M1HARQdelay}).

%%%%%%%%%%%%%%%%%%%%%%%%%%%%%%%%%%%%%%%%%%%%
\subsection{Proof of Proposition \ref{throughputCF-ARQ}}
\label{App:AppendixthroughputCF-ARQ}
For CF ARQ, the MGF of the transmission time for $M=2$ packets is given by
\begin{align}
\label{MGFtransmissiontimeCoding}
\PhitauCF(z)=z\matr{P}^{k}\Big[(\matr{I}-\matr{P}_{1x}^{\mathcal{T}}(1))^{-1}\matr{A}^{\CF}_1(z)
+\matr{P}^{\CF}_A(1){\prod\nolimits_{i=1}^2}{(\matr{I}-\matr{P}_{1x}^{\mathcal{T}}(i))^{-1}\matr{A}^{\CF}_2(z)}\Big],
\end{align}
where
\begin{align}
\matr{P}_{1x}^{\mathcal{T}}(1)=z(\matr{P}_{10}^{\CF}(1)+\matr{P}_{11}^{\CF}(1)(\matr{P}^2)^{d})\matr{P}^{k+1},\quad
\matr{P}_{1x}^{\mathcal{T}}(2)=z(\matr{P}_{10}^{\CF}(2)+\matr{P}_{11}^{\CF}(2)\matr{P}^{d-1})\matr{P}^{k},\nonumber
\end{align}
and the matrix $\matr{A}^{\CF}_n(z)$ that gives the gain of the transition from the state $A^{\CF}_{2-(n-1)}$ is computed as
\begin{align}
\matr{A}^{\CF}_n(z)&=\matr{P}_{00}^{\CF}(n)+\matr{P}_{01}^{\CF}(n)\left[\sum\nolimits_{i=1}^{d_n}{\matr{P}_{x1}^{\CF}(n)^{i-1}\matr{P}_{x0}^{\CF}(n)}\right.\nonumber\\
&\left.+\matr{P}_{x1}^{\CF}(n)^{d_n} (\matr{I}-z\matr{P}_{x1}^{\CF}(n)^T)^{-1}z\sum\nolimits_{i=0}^{T-1}{\matr{P}_{x1}^{\CF}(n)^i} \matr{P}_{x0}^{\CF}(n)\right],\quad n=\{1,2\},\nonumber
\end{align}
where $d_n=d-(n-1)$. The PGF of the transmission time of CF ARQ for $M=2$ coded packets is computed as $\phitauCF(z)=\pi_I\PhitauCF(z)\matr{1}/(\pi_I\matr{1})$ using the MGF $\matr{\Phi}_{\tau}(z)$ in (\ref{MGFtransmissiontimeCoding}), where $\matr{1}$ is a column vector of ones, $\pi_I$ is the probability vector of state $I$, and equals $\pi_I=\pi \matr{P}_0$. Finally, the throughput is the reciprocal of the derivative of $\phitauCF(z)$ at $z=1$, i.e., $\etaCF=1/\phitauCF'(1)$.

For CF ARQ with $M=2$, using the matrix-flow graph for throughput analysis (similar to the graph shown in Fig. \ref{SR_ARQ_CF} for delay), and using the basic simplification rules, the MGF of $\tauCF$, i.e., $\PhitauCF(z)$, can be computed as
\begin{align} 
&\PhitauCF(z)\nonumber\\
&=\matr{P}^{T-d-1}(\matr{I}-P_{1x,T1})^{-1}\Big[\matr{P}_{01}^{\CF}(1){\big(\matr{P}_{x1}^{\CF}(1)\big)}^d(\matr{I}-{\big(\matr{P}_{x1}^{\CF}(1)\big)}^T)^{-1} \Big(\sum\limits_{n=0}^{T-1}{{\big(\matr{P}_{x1}^{\CF}(1)\big)}^n}\Big) \matr{P}_{x0}^{\CF}(1) \nonumber\\
&+\matr{P}_{00}^{\CF}(1)+\Big(\sum\limits_{n=1}^{d}{\matr{P}_{01}^{\CF}(1){\big(\matr{P}_{x1}^{\CF}(1)\big)}^{n-1}\matr{P}_{x0}^{\CF}(1)}\Big)\nonumber\\
&+\matr{P}_0^*(\matr{I}-\matr{P}_{1x,T2}(1))^{-1}\Big(\matr{P}_{00}+\Big(\sum\limits_{n=1}^{d-1}{\matr{P}_{01}\matr{P}_{x1}^{n-1}\matr{P}_{x0}}\Big)+\matr{P}_{01}\matr{P}_{x1}^{d-1}(\matr{I}-\matr{P}_{x1}^T)^{-1} \Big(\sum\limits_{n=0}^{T-1}{\matr{P}_{x1}^n}\Big) \matr{P}_{x0}\Big)\Big]\nonumber\\
&+ \matr{P}^{T-d-1}(\matr{I}-\matr{P}_{1x,T1}(1))^{-1}\matr{P}'_{1x,T1}(1)(\matr{I}-\matr{P}_{1x,T1}(1))^{-1}\Big[\matr{P}_{01}^{\CF}(1) {\big(\matr{P}_{x1}^{\CF}(1)\big)}^d (\matr{I}-{\big(\matr{P}_{x1}^{\CF}(1)\big)}^T)^{-1} \nonumber
\end{align}
\begin{align}
&\Big(\sum\limits_{n=0}^{T-1}{{\big(\matr{P}_{x1}^{\CF}(1)\big)}^n}\Big) \matr{P}_{x0}^{\CF}(1)+\matr{P}_{00}^{\CF}(1)+\Big(\sum\limits_{n=1}^{d}{\matr{P}_{01}^{\CF}(1){\big(\matr{P}_{x1}^{\CF}(1)\big)}^{n-1}\matr{P}_{x0}^{\CF}(1)}\Big)\nonumber\\
&+\matr{P}_0^*(\matr{I}-\matr{P}_{1x,T2}(1))^{-1}\Big(\matr{P}_{00}+\Big(\sum\limits_{n=1}^{d-1}{\matr{P}_{01}\matr{P}_{x1}^{n-1}\matr{P}_{x0}}\Big)+\matr{P}_{01}\matr{P}_{x1}^{d-1}(\matr{I}-(\matr{P}_{x1}^{\CF}(2))^T)^{-1} \Big(\sum\limits_{n=0}^{T-1}{\matr{P}_{x1}^n}\Big) \matr{P}_{x0}\Big)\Big]\nonumber\\
&+ P^{T-d-1}(\matr{I}-\matr{P}_{1x,T1}(1))^{-1}
\Big[\matr{P}_{01}^{\CF}(1) {\big(\matr{P}_{x1}^{\CF}(1)\big)}^d(\matr{I}-{\big(\matr{P}_{x1}^{\CF}(1)\big)}^T)^{-1}{\big(\matr{P}_{x1}^{\CF}(1)\big)}^T(\matr{I}-{\big(\matr{P}_{x1}^{\CF}(1)\big)}^T)^{-1} \nonumber\\
&\times\Big(\sum\limits_{n=0}^{T-1}{{\big(\matr{P}_{x1}^{\CF}(1)\big)}^n}\Big) \matr{P}_{x0}^{\CF}(1)+\matr{P}_{01}^{\CF}(1) {\big(\matr{P}_{x1}^{\CF}(1)\big)}^d (\matr{I}-{\big(\matr{P}_{x1}^{\CF}(1)\big)}^T)^{-1} \Big(\sum\limits_{n=0}^{T-1}{{\big(\matr{P}_{x1}^{\CF}(1)\big)}^n}\Big) \matr{P}_{x0}^{\CF}(1)\nonumber\\                                  
&+\matr{P}_0^*(\matr{I}-\matr{P}_{1x,T2}(1))^{-1}\matr{P}'_{1x,T2}(1)(\matr{I}-\matr{P}_{1x,T2}(1))^{-1}\Big(\matr{P}_{00}+\Big(\sum\limits_{n=1}^{d-1}{\matr{P}_{01}\matr{P}_{x1}^{n-1}\matr{P}_{x0}}\Big)\nonumber\\
&+\matr{P}_{01}(\matr{P}_{x1})^{d-1}(\matr{I}-(\matr{P}_{x1})^T)^{-1} \Big(\sum\limits_{n=0}^{T-1}{(\matr{P}_{x1})^n}\Big) \matr{P}_{x0}\Big)\nonumber\\
&+\matr{P}_0^*(\matr{I}-\matr{P}_{1x,T2}(1))^{-1}\Big(\matr{P}_{01}\matr{P}_{x1}^{d-1}(\matr{I}-\matr{P}_{x1}^T)^{-1}\matr{P}_{x1}^T(\matr{I}-\matr{P}_{x1}^T)^{-1}\Big(\sum\limits_{n=0}^{T-1}{\matr{P}_{x1}^n}\Big)\matr{P}_{x0}\nonumber\\
&+\matr{P}_{01}(\matr{P}_{x1})^{d-1}(\matr{I}-(\matr{P}_{x1})^T)^{-1}\Big(\sum\limits_{n=0}^{T-1}{\matr{P}_{x1}^n}\Big)\matr{P}_{x0}\Big)\Big],\nonumber        
\end{align}
where we note that $\matr{P}_{xy}^{\CF}=\matr{P}_{xy}^{\CF}(1)$ and $\matr{P}_{xy}=\matr{P}_{xy}^{\CF}(2)$ for $x,y\in\{0,1\}$, and 
\begin{align}
\matr{P}_0^* &= \matr{P}_{00}\matr{P}_{10}+\matr{P}_{10}\matr{P}_{00}
       + \matr{P}_{00}\matr{P}_{01}+\matr{P}_{01}\matr{P}_{00}
       + \matr{P}_{00}\matr{P}_{11}+\matr{P}_{11}\matr{P}_{00}\nonumber\\
\matr{P}_{1x,T1}(z) &= z(\matr{P}_{10}^{\CF}(1)+\matr{P}_{11}^{\CF}(1)\matr{P}^d\matr{P}^d)\matr{P}^{T-d}\nonumber\\
\matr{P}_{1x,T2}(z) &= z(\matr{P}_{10}+\matr{P}_{11}\matr{P}^{d-1})\matr{P}^{T-d-1}.\nonumber
\end{align}
The value of $\avgtauCF$ for CF ARQ can be computed using the relation $\avgtauCF=\frac{1}{1-\epsilon}\pi_I \PhitauCF(z)\matr{1}$. Simplifying above expressions, the value of throughput $\etaCF=2{\avgtauCF}^{-1}$ for CF ARQ for $M=2$ for the case of memoryless channel is computed as
\begin{align}
\etaCF=\frac{(1+\epsilon-2\epsilon^2+2\epsilon^3)/(2-\epsilon)}{\alpha_{\CF}(\epsilon)-\beta_{\CF}(\epsilon,T,d) +\epsilon^{d}(1-\epsilon )/(1-\epsilon^T) },\nonumber
\end{align}
where $d=T-k$ and 
\begin{align}
\alpha_{\CF}(\epsilon)&=\frac{1+3\epsilon-2\epsilon^2+20\epsilon^3-18\epsilon^4+28\epsilon^5-60\epsilon^6+72\epsilon^7-40\epsilon^8+8\epsilon^9}{2(2-\epsilon)(1-\epsilon+4\epsilon^2-2\epsilon^3)(1+\epsilon-2\epsilon^2+2\epsilon^3)},\nonumber\\ 
\beta_{\CF}(\epsilon,T,d)&= \frac{(1-\epsilon)^2(1- 2\epsilon+4\epsilon^2) \epsilon^{d}(2-\epsilon)^{d-1}(1-2\epsilon+2\epsilon^2)^{d - 1}}{2(1-\epsilon+4\epsilon^2-2\epsilon^3)(\epsilon^T (2-\epsilon)^T (1-2\epsilon+2\epsilon^2)^T - 1)}.\nonumber
\end{align}
For $\epsilon\leq 0.5$, we have that $\alpha_{\CF}(\epsilon)\in[0.25, 0.8]$ and $\beta_{CF}(\epsilon,T,d)=0$ for all $T,d$ when $\epsilon=0$ and $\beta_{CF}(\epsilon,T,d)\to 0$ as $T\to \infty$ for all $\epsilon$, and $\beta_{CF}(\epsilon,T,d)\approx 0$ for any finite $T>k$. Using the relation $\frac{(\alpha_{\CF}(\epsilon)-\beta_{\CF}(\epsilon))(1-\epsilon)}{(1+\epsilon-2\epsilon^2+2\epsilon^3)/(2-\epsilon)}\approx 0.5$ when $T>k$, we have the final expression in (\ref{M2CFARQthroughput}).

%%%%%%%%%%%%%%%%%%%%%%%%%%%%%%%%%%%%%%%%%%%%
\subsection{Proof of Proposition \ref{avgdelayCF-ARQ}}
\label{App:AppendixavgdelayCF-ARQ}
For CF ARQ, the MGF of the delay for $M=2$ packets is given by
\begin{align}
\label{MGFdelayCoding}
\PhiDCF(z)=z^{k}\matr{P}^{k}\Big[ (\matr{I}-\matr{P}_{1x}^{\rm D}(1))^{-1}\matr{B}^{\CF}_1(z)
+z\matr{P}^{\CF}_A(1)\prod\nolimits_{i=1}^2 (\matr{I}-\matr{P}_{1x}^{\rm D}(i))^{-1}\matr{B}^{\CF}_2(z)\Big],
\end{align}
where $\matr{B}^{\CF}_n(z)$ for $n\in\{1,2\}$ can be computed using relation
\begin{align}
\matr{B}^{\CF}_n(z)=z\matr{P}_{00}^{\CF}(n)+z\matr{P}_{01}^{\CF}(n)(\matr{I}-z\matr{P}_{x1}^{\CF}(n))^{-1}z\matr{P}_{x0}^{\CF}(n).\nonumber
\end{align}

The PGF of delay of CF ARQ for $M=2$ coded packets can be computed as $\phiDCF(z)=\pi_I\PhiDCF(z)\matr{1}/(\pi_I\matr{1})$ using the MGF $\matr{\Phi}_D(z)$ in (\ref{MGFdelayCoding}). Finally, the average delay will be the derivative of $\phiDCF(z)$ at $z=1$, i.e., $\DCF=\phiDCF'(1)$.

Average delay $\DCF$ for CF ARQ for $M=2$ is given by
\begin{align}
\DCF&=\frac{1}{1-\epsilon}\pi_I\Big\{ k\matr{P}^k(\matr{I}-\matr{P}_{1x,D1})^{-1}\Big(\matr{P}_{00}^{\CF}(1)+\matr{P}_{01}^{\CF}(1)(\matr{I}-\matr{P}_{x1}^{\CF}(1))^{-1}\matr{P}_{x0}^{\CF}(1)\nonumber\\
&+\matr{P}_0^*(\matr{I}-\matr{P}_{1x,D2})^{-1}(\matr{P}_{00}+\matr{P}_{01}(\matr{I}-\matr{P}_{x1})^{-1}\matr{P}_{x0})\Big)\nonumber\\
&+ \matr{P}^k(\matr{I}-\matr{P}_{1x,D1})^{-1} \matr{P}'_{1x,D1}(\matr{I}-\matr{P}_{1x,D1})^{-1}\Big(\matr{P}_{00}^{\CF}(1)+\matr{P}_{01}^{\CF}(1)(\matr{I}-\matr{P}_{x1}^{\CF}(1))^{-1}\matr{P}_{x0}^{\CF}(1)\nonumber
\end{align}
\begin{align}
&+\matr{P}_0^*(\matr{I}-\matr{P}_{1x,D2})^{-1}(\matr{P}_{00}+\matr{P}_{01}(\matr{I}-\matr{P}_{x1})^{-1}\matr{P}_{x0})\Big)+ \matr{P}^k(\matr{I}-\matr{P}_{1x,D1})^{-1}\nonumber\\
&\Big(\matr{P}_{00}^{\CF}(1)+2\matr{P}_{01}^{\CF}(1)(\matr{I}-\matr{P}_{x1}^{\CF}(1))^{-1}\matr{P}_{x0}^{\CF}(1)+\matr{P}_{01}^{\CF}(1)(\matr{I}-\matr{P}_{x1}^{\CF}(1))^{-1}\matr{P}_{x1}^{\CF}(1)(\matr{I}-\matr{P}_{x1}^{\CF}(1))^{-1}\matr{P}_{x0}^{\CF}(1)\nonumber\\
&+\matr{P}_0^*(\matr{I}-\matr{P}_{1x,D2})^{-1}(\matr{P}_{00}+\matr{P}_{01}(\matr{I}-\matr{P}_{x1})^{-1}\matr{P}_{x0})\nonumber\\
&+\matr{P}_0^*(\matr{I}-\matr{P}_{1x,D2})^{-1}\matr{P}'_{1x,D2}(\matr{I}-\matr{P}_{1x,D2})^{-1}(\matr{P}_{00}+\matr{P}_{01}(\matr{I}-\matr{P}_{x1})^{-1}\matr{P}_{x0})\nonumber\\
&+\matr{P}_0^*(\matr{I}-\matr{P}_{1x,D2})^{-1}(\matr{P}_{00}+2\matr{P}_{01}(\matr{I}-\matr{P}_{x1})^{-1}\matr{P}_{x0}+\matr{P}_{01}(\matr{I}-\matr{P}_{x1})^{-1}\matr{P}_{x1}(\matr{I}-\matr{P}_{x1})^{-1}\matr{P}_{x0})\Big) \Big\}\matr{1},\nonumber
\end{align}
where
\begin{align}
\matr{P}_{1x,D1}&=z^{T-d}\matr{P}^{T-d}(z\matr{P}_{10}^{\CF}(1)+z\matr{P}_{11}^{\CF}(1) z^d \matr{P}^d)\nonumber\\
\matr{P}_{1x,D2}&=z^{T-d-1}\matr{P}^{T-d-1}(z\matr{P}_{10}+z\matr{P}_{11}z^{d+1}\matr{P}^{d+1}).\nonumber
\end{align}
Simplifying above expressions, the exact expression for average delay $\DCF$ for CF ARQ for $M=2$ for the case of memoryless channel is given by 
\begin{align}
&\DCF=(1+ k +\epsilon (8  + 2k)- \epsilon^2(9+ k) + 2\epsilon^3(27  + 13 k+3T)  -\epsilon^4 (49   + 49k + 3T)    \nonumber\\ 
&     + 2\epsilon^5(71+ 60 k+ 5T) - \epsilon^6 (393+283k- 83T)  + 4\epsilon^7(213+ 132k- 62T)- 2\epsilon^8(780 + 475 k- 271T)    \nonumber\\
&    + 4\epsilon^9(489+ 314 k- 200T) - 2\epsilon^{10}(757+ 521 k- 365T)+ 4\epsilon^{11}(171+ 129k- 101T)\nonumber\\
&   - 4\epsilon^{12}(41+ 35 k- 31T) + 16\epsilon^{13}(1+ k- T) )/(1+ \epsilon^2+ 6\epsilon^3- 12\epsilon^4+ 12\epsilon^5- 4\epsilon^6 )^2.\nonumber
\end{align}
Using this, the final expression can be derived as $\epsilon\to 0$ as given in (\ref{M2CFARQdelay}).

%%%%%%%%%%%%%%%%%%%%%%%%%%%%%%%%%%%%%%%%%%%%
\subsection{Proof of Proposition \ref{throughputcodedARQ}}
\label{App:AppendixthroughputcodedARQ}
For Coded ARQ, the MGF of the transmission time (for $M=2$ packets) is given by
\begin{align}
\label{CodedARQtransmissiontimeMGF}
\PhitauC(z)&=z^2\matr{P}^k (\matr{I}-\matr{g}^{\Cod}_2(z))^{-1} \times \nonumber\\
&\Big[[(\matr{P}_{100}+\matr{P}_{010})+(\matr{P}_{011}+\matr{P}_{101})(z\matr{P}^{T}+ (\matr{I}-\matr{g}^{\Cod}_3(z))^{-1}-\matr{I}) \matr{P}_{10}](\matr{I}-\matr{g}^{\Cod}_1(z))^{-1} [\matr{A}^{\Cod}_{00}+\matr{A}^{\Cod}_{01}]\nonumber\\
&+\matr{A}^{\Cod}_{000}+\matr{A}^{\Cod}_{001}+(\matr{P}_{011}+\matr{P}_{101})(z\matr{P}^{T}+ (\matr{I}-\matr{g}^{\Cod}_3(z))^{-1}-\matr{I})[\matr{A}^{\Cod}_{00}+\matr{A}^{\Cod}_{01}]\Big],
\end{align}
where the functions $\matr{g}^{\Cod}_1(z)$, $\matr{g}^{\Cod}_2(z)$ and $\matr{g}^{\Cod}_3(z)$ are the branch gains for the self-loops at states $A_1$, $A_2$ and $A_3$ as shown in Fig. \ref{SR_ARQ_Coding_alphabet3} respectively, and are given by 
\begin{align}
\label{g1g2g3}
\matr{g}^{\Cod}_1(z)&=(\matr{P}_{10}+\matr{P}_{11}\matr{P}^{T-k})z\matr{P}^{k-1},\nonumber\\
\matr{g}^{\Cod}_2(z)&=(\matr{P}_{110}+\matr{P}_{111}\matr{P}^{T-k}) z\matr{P}^k,\nonumber\\
\matr{g}^{\Cod}_3(z)&=\matr{P}_{11} \matr{P}^{T} z,
\end{align}
and the matrices in (\ref{CodedARQtransmissiontimeMGF}) are given by
\begin{align}
\matr{A}^{\Cod}_{00}+\matr{A}^{\Cod}_{01}&=\left[\matr{P}_{00}+\matr{P}_{01}\Big(\sum\limits_{n=0}^{T-k-1}{\matr{P}_{x1}^n}\Big)\matr{P}_{x0}\right]+\left[\matr{P}_{01}\matr{P}_{x1}^{T-k}(\matr{I}-z\matr{P}_{x1}^T)^{-1}z\Big(\sum\limits_{n=0}^{T-1}{\matr{P}_{x1}^n}\Big)\matr{P}_{x0}\right],\nonumber\\ 
\matr{A}^{\Cod}_{000}+\matr{A}^{\Cod}_{001}&=\left[\matr{P}_{000}+\matr{P}_{001}\Big(\sum\limits_{n=0}^{T-k-1}{ \matr{P}_{xx1}^n }\Big)\matr{P}_{xx0}\right]+\left[\matr{P}_{001}\matr{P}_{xx1}^{T-k}(\matr{I}-z\matr{P}_{xx1}^T)^{-1}z\Big(\sum\limits_{n=0}^{T-1}{\matr{P}_{xx1}^n}\Big)\matr{P}_{xx0}\right].\nonumber
\end{align}

The value of $\avgtauC$ for Coded ARQ for $M=2$ for the case of memoryless channel is given by
\begin{align}
\avgtauC&= \frac{1}{1-\epsilon}\pi_I\Big\{2\matr{P}^k(\matr{I}-\matr{g}^{\Cod}_2(1))^{-1}(\matr{B}^{\Cod}_1(1)+\matr{B}^{\Cod}_2(1)+\matr{B}^{\Cod}_3(1)+\matr{B}^{\Cod}_4(1))\nonumber\\
&+ \matr{P}^k(\matr{I}-\matr{g}^{\Cod}_2(1))^{-1}(\matr{g}^{\Cod}_2)'(1)(\matr{I}-\matr{g}^{\Cod}_2(1))^{-1}(\matr{B}^{\Cod}_1(1)+\matr{B}^{\Cod}_2(1)+\matr{B}^{\Cod}_3(1)+\matr{B}^{\Cod}_4(1))\nonumber\\
&+ \matr{P}^k(\matr{I}-\matr{g}^{\Cod}_2(1))^{-1}((\matr{B}^{\Cod}_1)'(1)+(\matr{B}^{\Cod}_2)'(1)+(\matr{B}^{\Cod}_3)'(1)+(\matr{B}^{\Cod}_4)'(1))\Big\}\matr{1},\nonumber
\end{align}
where 
\begin{align}   
 \matr{B}^{\Cod}_1(z) &= ((\matr{P}_{100}+\matr{P}_{010})+(\matr{P}_{011}+\matr{P}_{101})(z\matr{P}^T+(\matr{I}-\matr{g}^{\Cod}_3(z))^{-1}-\matr{I})\matr{P}_{10})(\matr{I}-\matr{g}^{\Cod}_1(z))^{-1}\nonumber\\
 &\times((\matr{P}_{00}+\matr{P}_{01}\Big(\sum\limits_{n=0}^{T-k-1}{\matr{P}_{x1}^n}\Big)\matr{P}_{x0})+(\matr{P}_{01}\matr{P}_{x1}^{T-k}(\matr{I}-z\matr{P}_{x1}^T)^{-1}z \Big(\sum\limits_{n=0}^{T-1}{\matr{P}_{x1}^n}\Big)\matr{P}_{x0}))\nonumber\\
    \matr{B}^{\Cod}_2(z) &= \matr{P}_{000}+\matr{P}_{001}\Big(\sum\limits_{n=0}^{T-k-1}{\matr{P}_{xx1}^n}\Big)\matr{P}_{xx0}\nonumber\\
    \matr{B}^{\Cod}_3(z) &= \matr{P}_{001}\matr{P}_{xx1}^{T-k}(\matr{I}-z\matr{P}_{xx1}^T)^{-1}z \Big(\sum\limits_{n=0}^{T-1}{\matr{P}_{xx1}^n}\Big)\matr{P}_{xx0}\nonumber\\
    \matr{B}^{\Cod}_4(z) &= (\matr{P}_{011}+\matr{P}_{101})(z\matr{P}^T+(\matr{I}-\matr{g}^{\Cod}_3(1))^{-1}-\matr{I})\nonumber\\
    &\times((\matr{P}_{00}+\matr{P}_{01} \Big(\sum\limits_{n=0}^{T-k-1}{\matr{P}_{x1}^n}\Big) \matr{P}_{x0})+(\matr{P}_{01}\matr{P}_{x1}^{T-k}(\matr{I}-z\matr{P}_{x1}^T)^{-1}z\Big(\sum\limits_{n=0}^{T-1}{\matr{P}_{x1}^n}\Big)\matr{P}_{x0})),\nonumber
\end{align}
where $\matr{g}^{\Cod}_i(z)$'s for $i\in\{1,2,3\}$ are given in (\ref{g1g2g3}). Finally, simplifying above expressions, throughput $\etaC=2{\avgtauC}^{-1}$ for the memoryless channel is given by (\ref{M2codedARQthroughput}).

%%%%%%%%%%%%%%%%%%%%%%%%%%%%%%%%%%%%%%%%%%%%
\subsection{Proof of Proposition \ref{avgdelaycodedARQ}}
\label{App:AppendixavgdelaycodedARQ}
For Coded ARQ, the MGF of the delay (for $M=2$ packets) is given by
\begin{align}
\label{M2uncodedARQaveragedelay}
\PhiDC(z)&=z^{k}\matr{P}^{k}(\matr{I}-\matr{f}^{\Cod}_2(z))^{-1}\times \nonumber\\
&\Big[[z(\matr{P}_{100}+\matr{P}_{010})+z(\matr{P}_{011}+\matr{P}_{101})(z^T\matr{P}^T+(\matr{I}-\matr{f}^{\Cod}_3(z))^{-1}-\matr{I})z\matr{P}_{10}](\matr{I}-\matr{f}^{\Cod}_1(z))^{-1}\matr{B}^{\Cod}_{00}\nonumber\\
&+\matr{B}^{\Cod}_{000}+z(\matr{P}_{011}+\matr{P}_{101})(z^T\matr{P}^T+(\matr{I}-\matr{f}^{\Cod}_3(z))^{-1}-\matr{I})\matr{B}^{\Cod}_{00}\Big],
\end{align}
where the functions $\matr{f}^{\Cod}_1(z)$, $\matr{f}^{\Cod}_2(z)$ and $\matr{f}^{\Cod}_3(z)$ are the branch gains for the self-loops at states $A_1$, $A_2$ and $A_3$ as shown in Fig. \ref{SR_ARQ_Coding_alphabet3} respectively, and are given by 
\begin{align}
\label{f1f2f3}
\matr{f}^{\Cod}_1(z)&=(z\matr{P}_{10}+z\matr{P}_{11}z^{T-k}\matr{P}^{T-k})z^{k-1}\matr{P}^{k-1},\nonumber\\
\matr{f}^{\Cod}_2(z)&=(z\matr{P}_{110}+z\matr{P}_{111}z^{T-k}\matr{P}^{T-k})z^{k}\matr{P}^{k},\nonumber\\
\matr{f}^{\Cod}_3(z)&=z\matr{P}_{11}z^{T}\matr{P}^{T},
\end{align}
and the matrices in (\ref{M2uncodedARQaveragedelay}) are given by
\begin{align}
\matr{B}^{\Cod}_{00}&=z\matr{P}_{00}+z\matr{P}_{01}(\matr{I}-z\matr{P}_{x1})^{-1}z\matr{P}_{x0},\nonumber\\
\matr{B}^{\Cod}_{000}&=z\matr{P}_{000}+z\matr{P}_{001}(\matr{I}-z\matr{P}_{xx1})^{-1}z\matr{P}_{xx0}.\nonumber
\end{align}
Average delay $\DC$ for Coded ARQ for $M=2$ for the case of memoryless channel is given by
\begin{align}
\DC&=\frac{1}{1-\epsilon}\pi_I\Big\{k\matr{P}^k(\matr{I}-\matr{f}^{\Cod}_2(1))^{-1} ( \matr{A}^{\Cod}_1(1) (\matr{I}-\matr{f}^{\Cod}_1(1))^{-1}\matr{A}^{\Cod}_2(1)+\matr{A}^{\Cod}_3(1)+\matr{A}^{\Cod}_4(1))\nonumber\\
           &+ \matr{P}^k(\matr{I}-\matr{f}^{\Cod}_2(1))^{-1} (\matr{f}^{\Cod}_2)'(1)(\matr{I}-\matr{f}^{\Cod}_2(1))^{-1}( \matr{A}^{\Cod}_1(1)(\matr{I}-\matr{f}^{\Cod}_1(1))^{-1}\matr{A}^{\Cod}_2(1)+\matr{A}^{\Cod}_3(1)+\matr{A}^{\Cod}_4(1))\nonumber\\
           &+ \matr{P}^k(\matr{I}-\matr{f}^{\Cod}_2(1))^{-1} \Big( (\matr{A}^{\Cod}_1)'(1)(\matr{I}-\matr{f}^{\Cod}_1(1))^{-1}\matr{A}^{\Cod}_2(1)\nonumber\\
           &+\matr{A}^{\Cod}_1(1) (\matr{I}-\matr{f}^{\Cod}_1(1))^{-1} (\matr{f}^{\Cod}_1)'(1)(\matr{I}-\matr{f}^{\Cod}_1(1))^{-1} \matr{A}^{\Cod}_2(1)\nonumber\\
           &+\matr{A}^{\Cod}_1(1) (\matr{I}-\matr{f}^{\Cod}_1(1))^{-1} (\matr{A}^{\Cod}_2)'(1)+(\matr{A}^{\Cod}_3)'(1)+(\matr{A}^{\Cod}_4)'(1)\Big)\Big\}\matr{1},\nonumber
\end{align}
where  
\begin{align}    
\matr{A}^{\Cod}_1(z) & = z(\matr{P}_{100}+\matr{P}_{010})+z(\matr{P}_{011}+\matr{P}_{101})(z^T\matr{P}^T+(\matr{I}-\matr{f}^{\Cod}_3(z))^{-1}-\matr{I})z\matr{P}_{10}\nonumber\\              
\matr{A}^{\Cod}_2(z) & = z\matr{P}_{00}+z\matr{P}_{01}(\matr{I}-z \matr{P}_{x1})^{-1}z\matr{P}_{x0}\nonumber
\end{align}
\begin{align}
\matr{A}^{\Cod}_3(z) & = z\matr{P}_{000}+z\matr{P}_{001}(\matr{I}-z\matr{P}_{xx1})^{-1}z\matr{P}_{xx0}\nonumber\\
\matr{A}^{\Cod}_4(z) &= z(\matr{P}_{011}+\matr{P}_{101})(z^T\matr{P}^T+
(\matr{I}-\matr{f}^{\Cod}_3(z))^{-1}-\matr{I})(z\matr{P}_{00}+z\matr{P}_{01}(\matr{I}-z\matr{P}_{x1})^{-1}z\matr{P}_{x0}),\nonumber
\end{align}
where $\matr{f}^{\Cod}_i(z)$'s are given in (\ref{f1f2f3}). The average delay for the memoryless channel is computed as
\begin{align}
\DC&=(4\epsilon + k + \epsilon k + 2T\epsilon^2 + 3T\epsilon^3 - T\epsilon^4 + 4T\epsilon^5 + 2T\epsilon^6 - 4T\epsilon^7 + 4\epsilon^2k - 3\epsilon^3k \nonumber\\ &- \epsilon^4k - 4\epsilon^5k - 2\epsilon^6k + 4\epsilon^7k + 6\epsilon^2 + 5\epsilon^3 - 6\epsilon^4 - 4\epsilon^5 + 2\epsilon^6 + 1)/((1-\epsilon)(1+\epsilon)^2).\nonumber
\end{align}

%%%%%%%%%%%%%%%%%%%%%%%%%%%%%%%%%%%%%%%%%%%%
\subsection{Proof of Proposition \ref{secondmomentdelayCodedARQ}}
\label{App:AppendixsecondmomentdelayCodedARQ}
The second derivative of $\PhiDC(z)$ at $z=1$ for Coded ARQ can be computed as 
\begin{align}
&{\PhiDC}''(1)= k(k-1)\matr{P}^k(\matr{I}-\matr{f}^{\Cod}_2(1))^{-1}(\matr{A}^{\Cod}_1(1) (\matr{I}-\matr{f}^{\Cod}_1(1))^{-1}\matr{A}^{\Cod}_2(1)+\matr{A}^{\Cod}_3(1)+\matr{A}^{\Cod}_4(1))\nonumber\\
&+k\matr{P}^k(\matr{I}-\matr{f}^{\Cod}_2(1))^{-1} (\matr{f}^{\Cod}_2(1))'(\matr{I}-\matr{f}^{\Cod}_2(1))^{-1}(\matr{A}^{\Cod}_1(1)(\matr{I}-\matr{f}^{\Cod}_1(1))^{-1}\matr{A}^{\Cod}_2(1)+\matr{A}^{\Cod}_3(1)+\matr{A}_4(1))\nonumber\\
&+k\matr{P}^k(\matr{I}-\matr{f}^{\Cod}_2(1))^{-1}( (\matr{A}^{\Cod}_1(1))' (\matr{I}-\matr{f}^{\Cod}_1(1))^{-1}\matr{A}_2(1)+\matr{A}^{\Cod}_1(1)(\matr{I}-\matr{f}^{\Cod}_1(1))^{-1}(\matr{f}^{\Cod}_1(1))'(\matr{I}-\matr{f}^{\Cod}_1(1))^{-1}\matr{A}^{\Cod}_2(1)\nonumber\\
&+\matr{A}^{\Cod}_1(1) (\matr{I}-\matr{f}^{\Cod}_1(1))^{-1}(\matr{A}^{\Cod}_2(1))'+(\matr{A}^{\Cod}_3(1))'+(\matr{A}^{\Cod}_4(1))')\nonumber\\
&+ k\matr{P}^k(\matr{I}-\matr{f}^{\Cod}_2(1))^{-1}(\matr{f}^{\Cod}_2(1))'(\matr{I}-\matr{f}^{\Cod}_2(1))^{-1}( \matr{A}^{\Cod}_1(1)(\matr{I}-\matr{f}^{\Cod}_1(1))^{-1}\matr{A}^{\Cod}_2(1)+\matr{A}^{\Cod}_3(1)+\matr{A}^{\Cod}_4(1))\nonumber\\
&+\matr{P}^k(\matr{I}-\matr{f}^{\Cod}_2(1))^{-1}(\matr{f}^{\Cod}_2(1))'(\matr{I}-\matr{f}^{\Cod}_2(1))^{-1}(\matr{f}^{\Cod}_2(1))'(\matr{I}-\matr{f}^{\Cod}_2(1))^{-1}( \matr{A}^{\Cod}_1(1)(\matr{I}-\matr{f}^{\Cod}_1(1))^{-1}\matr{A}^{\Cod}_2(1)\nonumber\\
&+\matr{A}^{\Cod}_3(1)+\matr{A}^{\Cod}_4(1))+\matr{P}^k(\matr{I}-\matr{f}^{\Cod}_2(1))^{-1}(\matr{f}^{\Cod}_2(1))'(\matr{I}-\matr{f}^{\Cod}_2(1))^{-1}( \matr{A}^{\Cod}_1(1)(\matr{I}-\matr{f}^{\Cod}_1(1))^{-1}\matr{A}^{\Cod}_2(1)+\matr{A}^{\Cod}_3(1)+\matr{A}^{\Cod}_4(1))\nonumber\\
&+\matr{P}^k(\matr{I}-\matr{f}^{\Cod}_2(1))^{-1}(\matr{f}^{\Cod}_2(1))''(\matr{I}-\matr{f}^{\Cod}_2(1))^{-1}( \matr{A}^{\Cod}_1(1) (\matr{I}-\matr{f}^{\Cod}_1(1))^{-1}\matr{A}^{\Cod}_2(1)+\matr{A}^{\Cod}_3(1)+\matr{A}^{\Cod}_4(1))\nonumber\\
&+\matr{P}^k(\matr{I}-\matr{f}^{\Cod}_2(1))^{-1}(\matr{f}^{\Cod}_2(1))'(\matr{I}-\matr{f}^{\Cod}_2(1))^{-1}(\matr{f}^{\Cod}_2(1))'(\matr{I}-\matr{f}^{\Cod}_2(1))^{-1}( \matr{A}^{\Cod}_1(1)(\matr{I}-\matr{f}^{\Cod}_1(1))^{-1}\matr{A}^{\Cod}_2(1)+\matr{A}^{\Cod}_3(1)\nonumber\\
&+\matr{A}^{\Cod}_4(1))+\matr{P}^k(\matr{I}-\matr{f}^{\Cod}_2(1))^{-1}(\matr{f}^{\Cod}_2(1))'(\matr{I}-\matr{f}^{\Cod}_2(1))^{-1}( (\matr{A}^{\Cod}_1(1))'(\matr{I}-\matr{f}^{\Cod}_1(1))^{-1}\matr{A}^{\Cod}_2(1)\nonumber\\
&+\matr{A}^{\Cod}_1(1)(\matr{I}-\matr{f}^{\Cod}_1(1))^{-1}(\matr{f}^{\Cod}_1(1))'(\matr{I}-\matr{f}^{\Cod}_1(1))^{-1}\matr{A}^{\Cod}_2(1)+\matr{A}^{\Cod}_1(1)(\matr{I}-\matr{f}^{\Cod}_1(1))^{-1}(\matr{A}^{\Cod}_2(1))'+(\matr{A}^{\Cod}_3(1))'+(\matr{A}^{\Cod}_4(1))')\nonumber\\
&+ k\matr{P}^k(\matr{I}-\matr{f}^{\Cod}_2(1))^{-1}( (\matr{A}^{\Cod}_1(1))'(\matr{I}-\matr{f}^{\Cod}_1(1))^{-1}\matr{A}^{\Cod}_2(1)+\matr{A}^{\Cod}_1(1)(\matr{I}-\matr{f}^{\Cod}_1(1))^{-1}(\matr{f}^{\Cod}_1(1))'(\matr{I}-\matr{f}^{\Cod}_1(1))^{-1}\matr{A}^{\Cod}_2(1)\nonumber\\
&+\matr{A}^{\Cod}_1(1)\matr{I}/(\matr{I}-\matr{f}^{\Cod}_1(1))(\matr{A}^{\Cod}_2(1))'+(\matr{A}^{\Cod}_3(1))'+(\matr{A}^{\Cod}_4(1))') \nonumber\\
&+ \matr{P}^k(\matr{I}-\matr{f}^{\Cod}_2(1))^{-1}(\matr{f}^{\Cod}_2(1))'(\matr{I}-\matr{f}^{\Cod}_2(1))^{-1}( (\matr{A}^{\Cod}_1(1))'(\matr{I}-\matr{f}^{\Cod}_1(1))^{-1} \matr{A}^{\Cod}_2(1)\nonumber\\
&+\matr{A}^{\Cod}_1(1) (\matr{I}-\matr{f}^{\Cod}_1(1))^{-1}(\matr{f}^{\Cod}_1(1))'(\matr{I}-\matr{f}^{\Cod}_1(1))^{-1} \matr{A}^{\Cod}_2(1)+\matr{A}^{\Cod}_1(1) (\matr{I}-\matr{f}^{\Cod}_1(1))^{-1}(\matr{A}^{\Cod}_2(1))'+(\matr{A}^{\Cod}_3(1))'+(\matr{A}^{\Cod}_4(1))')\nonumber
\end{align}
\begin{align}
&+ \matr{P}^k(\matr{I}-\matr{f}^{\Cod}_2(1))^{-1}( (\matr{A}^{\Cod}_1(1))'' (\matr{I}-\matr{f}^{\Cod}_1(1))^{-1}\matr{A}^{\Cod}_2(1)+(\matr{A}^{\Cod}_1(1))' (\matr{I}-\matr{f}^{\Cod}_1(1))^{-1}(\matr{f}^{\Cod}_1(1))'(\matr{I}-\matr{f}^{\Cod}_1(1))^{-1}\matr{A}^{\Cod}_2(1)\nonumber\\
&+(\matr{A}^{\Cod}_1(1))' (\matr{I}-\matr{f}^{\Cod}_1(1))^{-1}(\matr{A}^{\Cod}_2(1))'+(\matr{A}^{\Cod}_1(1))' (\matr{I}-\matr{f}^{\Cod}_1(1))^{-1}(\matr{f}^{\Cod}_1(1))'(\matr{I}-\matr{f}^{\Cod}_1(1))^{-1}\matr{A}^{\Cod}_2(1)\nonumber\\
&+\matr{A}^{\Cod}_1(1) (\matr{I}-\matr{f}^{\Cod}_1(1))^{-1}(\matr{f}^{\Cod}_1(1))'(\matr{I}-\matr{f}^{\Cod}_1(1))^{-1}(\matr{f}^{\Cod}_1(1))'(\matr{I}-\matr{f}^{\Cod}_1(1))^{-1}\matr{A}^{\Cod}_2(1)\nonumber\\
&+\matr{A}^{\Cod}_1(1)(\matr{I}-\matr{f}^{\Cod}_1)^{-1}(\matr{f}^{\Cod}_1(1))''(\matr{I}-\matr{f}^{\Cod}_1(1))^{-1}\matr{A}^{\Cod}_2(1)+\matr{A}^{\Cod}_1(1)(\matr{I}-\matr{f}^{\Cod}_1(1))^{-1}(\matr{f}^{\Cod}_1(1))'(\matr{I}-\matr{f}^{\Cod}_1(1))^{-1}\matr{A}^{\Cod}_2(1)\nonumber\\
&+\matr{A}^{\Cod}_1(1) (\matr{I}-\matr{f}^{\Cod}_1(1))^{-1}(\matr{f}^{\Cod}_1(1))'(\matr{I}-\matr{f}^{\Cod}_1(1))^{-1}(\matr{f}^{\Cod}_1(1))'(\matr{I}-\matr{f}^{\Cod}_1(1))^{-1}\matr{A}^{\Cod}_2(1)\nonumber\\
&+\matr{A}^{\Cod}_1(1) (\matr{I}-\matr{f}^{\Cod}_1)^{-1}(\matr{f}^{\Cod}_1(1))'(\matr{I}-\matr{f}^{\Cod}_1(1))^{-1}(\matr{A}^{\Cod}_2(1))' + (\matr{A}^{\Cod}_1(1))' (\matr{I}-\matr{f}^{\Cod}_1(1))^{-1} (\matr{A}^{\Cod}_2(1))' \nonumber\\
&+ \matr{A}^{\Cod}_1(1) (\matr{I}-\matr{f}^{\Cod}_1(1))^{-1}(\matr{f}^{\Cod}_1(1))'(\matr{I}-\matr{f}^{\Cod}_1(1))^{-1}(\matr{A}^{\Cod}_2(1))' + \matr{A}^{\Cod}_1(1) (\matr{I}-\matr{f}^{\Cod}_1(1))^{-1}(\matr{A}^{\Cod}_2(1))'' +(\matr{A}^{\Cod}_3(1))''+(\matr{A}^{\Cod}_4(1))''),\nonumber
\end{align}
where $\matr{A}^{\Cod}_j(1)$, $j\in\{1,2,3,4\}$, are given in Appendix \ref{App:AppendixavgdelaycodedARQ}, and $\matr{f}^{\Cod}_i(1)$'s for $i\in\{1,2,3\}$ are given in (\ref{f1f2f3}). Using the MGF of delay, we can compute the variance of the delay for memoryless channels as 
\begin{align}
\varDC&=\frac{\epsilon^2 k^2}{(1+\epsilon)^4}(1+ 5\epsilon- 6\epsilon^2-\epsilon^3 - 5\epsilon^4+ 10\epsilon^5+28\epsilon^6+ 20\epsilon^7+12\epsilon^{8}-16\epsilon^{9}- 16\epsilon^{10} )\nonumber\\
&-\frac{2\epsilon^4 T k}{(1-\epsilon)(1+\epsilon)^3}(2+ \epsilon- 5\epsilon^2+ 10\epsilon^3 -16\epsilon^4+12\epsilon^5-12\epsilon^6- 16\epsilon^{7}+16\epsilon^{8}) \nonumber\\ 
&+\frac{\epsilon^2  k}{(1-\epsilon)(1+\epsilon)^4}(5- 16\epsilon - 40\epsilon^2- 52\epsilon^3+ 3\epsilon^4+ 60\epsilon^5+ 96\epsilon^6+32\epsilon^7- 72\epsilon^8- 24\epsilon^{9} +16\epsilon^{10} )\nonumber\\ 
&+\frac{\epsilon^2 T^2}{(1-\epsilon)^2(1+\epsilon)^2}( 2+ \epsilon-8\epsilon^2+13\epsilon^3+7\epsilon^4-30\epsilon^5 +20\epsilon^6 +20\epsilon^7-52\epsilon^8+ 48\epsilon^9- 16\epsilon^{10}) \nonumber\\ 
&+ \frac{\epsilon^2 T}{(1-\epsilon)^2(1+\epsilon)^3}(4+ 5\epsilon- 21\epsilon^2- 3\epsilon^3+ 27\epsilon^4- 24\epsilon^5- 52\epsilon^6+104\epsilon^7+8\epsilon^8- 56\epsilon^{9}+16\epsilon^{10})\nonumber\\
&+\frac{\epsilon}{(1-\epsilon)^2(1+\epsilon)^4}(  3+ 4\epsilon- 4\epsilon^2- 33\epsilon^3- 29\epsilon^4+ 47\epsilon^5+ 100\epsilon^6- 8\epsilon^7- 72\epsilon^8+ 4\epsilon^9+ 16\epsilon^{10}-4\epsilon^{11}).  \nonumber        
\end{align}
Hence, as $\epsilon\to 0$, the variance of delay can be approximated as in (\ref{variabilityCodedARQ}).

\end{appendix}

\bibliographystyle{IEEEtran}
\bibliography{Derya}

\end{document}